\documentclass[12pt,twoside]{book}
\usepackage[mathscr]{eucal}
\usepackage{amsmath,amsfonts,amssymb,amsthm}
\usepackage[matrix,arrow]{xy}
\usepackage{times}
\usepackage{pdfsync}
\usepackage{graphicx}
\usepackage{epstopdf}
\usepackage{dcolumn}
\usepackage{hyperref}
\usepackage{empheq}

\def\d_Vphi{\text{d}_V\hspace{-0.06em}\phi}
\def\d_Vphibar{\text{d}_V\hspace{-0.06em}\bar\phi}
\def\d_Vxi{\text{d}_V\hspace{-0.06em}\xi}

\def\ndelta{\delta\hspace{-0.50em}\slash\hspace{-0.05em} }
\def\nG{G\hspace{-1.3em}-\hspace{+0.05em} }

\newcommand{\lc}{\boldsymbol{\varepsilon}}

\def\be{\begin{eqnarray}}
\def\ee{\end{eqnarray}}
\def\beann{\begin{eqnarray*}}
\def\eeann{\end{eqnarray*}}
\def\beq{\begin{equation}}
\def\eeq{\end{equation}}
\def\ba{\begin{array}}
\def\ea{\end{array}}
\def\ben{\begin{enumerate}}
\def\een{\end{enumerate}}
\def\bea{\begin{eqnarray}}
\def\eea{\end{eqnarray}}

\def\5{\bar }
\def\6{\partial }
\def\7{\hat }
\def\4{\tilde }
\advance\textheight\topskip
\renewcommand{\tilde}{\widetilde}
\renewcommand{\hat}{\widehat}
\newtheorem{prop}{Proposition}[section]

\newtheorem{definition}[prop]{Definition}
\newtheorem{theorem}[prop]{Theorem}

\renewcommand{\simeq}{\cong}

\newcommand{\T}{\mathrm{T}}

\newcommand{\bref}[1]{\textbf{\ref{#1}}}

\newcommand{\dd}{\partial}
\renewcommand{\d}{\partial}

\renewcommand{\geq}{\,{\geqslant}\,}

\newcommand{\binner}[2]{%
  {\langle}\kern-4.15pt{\langle}#1{,}\,#2{\rangle}\kern-4.15pt{\rangle}}

\newcommand{\half}{\mathchoice{%
    \ffrac{1}{2}}{\frac{1}{2}}{\frac{1}{2}}{\frac{1}{2}}}

\newcommand{\ffrac}[2]{\raisebox{.5pt}%
  {\footnotesize$\displaystyle\frac{#1}{#2}$}\kern1pt}

\newcommand{\dover}[2]{\ffrac{\dd #1}{\dd #2}}

\newcommand{\ddl}[2]{\ffrac{\dd #1}{\dd #2}}

\newcommand{\vddl}[2]{{\ffrac{\delta #1}{\delta #2}}}

\newcommand{\RR}{\mathbb{R}}

\newcommand{\ZZ}{\mathbb{Z}}
\newcommand{\NN}{\mathbb{N}}

\def\cA{\mathcal{A}}

\def\cD{\mathcal{D}}
\def\cE{\mathcal{E}}
\def\cF{\mathcal{F}}
\def\cG{\mathcal{G}}
\def\cH{\mathcal{H}}
\def\cI{\mathcal{I}}

\def\cK{\mathcal{K}}
\def\cL{\mathcal{L}}

\def\cN{\mathcal{N}}
\def\cO{\mathcal{O}}
\def\cP{\mathcal{P}}
\def\cQ{\mathcal{Q}}
\def\cR{\mathcal{R}}

\def\cW{\mathcal{W}}
\def\cX{\mathcal{X}}

\DeclareFontFamily{OT1}{rsfs}{} \DeclareFontShape{OT1}{rsfs}{m}{n}{
<-7> rsfs5 <7-10> rsfs7 <10-> rsfs10}{}
\DeclareMathAlphabet{\mycal}{OT1}{rsfs}{m}{n}
\numberwithin{equation}{section} \makeatletter
\@addtoreset{equation}{section}

\begin{document}

\def\mytitle{Aspects of duality in gravitational theories.}

\pagestyle{myheadings} \markboth{\textsc{\small Troessaert}}{%
  \textsc{\small Aspects of duality in gravitational theories.}} \addtolength{\headsep}{4pt}

\begin{centering}
\mbox{}

  \vspace{5cm}

  \textbf{\Large{\mytitle}}



  \vspace{1.5cm}

  {\large C\'edric Troessaert$^{a}$}

\vspace{.5cm}

\begin{minipage}{.9\textwidth}\small \it \begin{center}
   Physique Th\'eorique et Math\'ematique, Universit\'e Libre de
   Bruxelles\\ and \\ International Solvay Institutes, \\ Campus
   Plaine C.P. 231, B-1050 Bruxelles, Belgium \end{center}
\end{minipage}

\end{centering}

\vspace{5 cm}

\begin{minipage}{.9\textwidth}\small \begin{center}
   Th\`ese r\'ealis\'ee en vue de l'obtention du titre de docteur en
   sciences.\\ Directeur Glenn Barnich \end{center}
\end{minipage}

\vspace{1 cm}

\begin{minipage}{.9\textwidth}\small \begin{center}
  Ann\'ee acad\'emique 2010-2011\end{center}
\end{minipage}
\vfill

\noindent
\mbox{}
\raisebox{-3\baselineskip}{%
  \parbox{\textwidth}{\mbox{}\hrulefill\\[-4pt]}}
{\scriptsize$^a$ Research Fellow of the Fund for
  Scientific Research-FNRS.}

\thispagestyle{empty}

\begin{small}
{\addtolength{\parskip}{-1.5pt}
 \tableofcontents}
\end{small}

\chapter*{Remerciements}

Tout d'abord, je tiens \`a exprimer toute ma gratitude \`a mon promoteur Glenn Barnich sans qui cette th\`ese n'aurait pu voir le jour. Je lui suis tr\`es reconnaissant pour sa disponibilit\'e, pour les projets de recherche qu'il m'a propos\'es, pour le partage de son exp\'erience et de sa passion pour la physique. Son \'eternel enthousiasme m'a de nombreuses fois permis de sortir de mes p\'eriodes de doute ou de d\'ecouragement.
 
Un grand merci \`a mon coll\`egue de bureau, Fran\c cois Dehouck pour la bonne ambiance de travail qu'il a contribu\'e \`a installer, sa motivation contagieuse et les innombrables apr\`es-midis de discussion au bureau.
 
Je suis tr\`es reconnaissant ˆ Fabienne De Neyn et Marie-France Rogge pour leur efficacit\'e et leur disponibilit\'e.
 
Mes amis, ma famille, de nombreux coll\`egues m'ont soutenu et aid\'e durant ces quatre ann\'ees de travail, merci \`a tous.

\chapter{Introduction}

One of the biggest challenge of modern physics is the quantization
of gravity. The theory introduced by Einstein in 1916 is still the
best we have nearly 100 years later. Unfortunately, it is a
classical theory and we know that physics is fundamentally quantum. In
some sense it is even worse, Einstein's gravity predicts its own
downfall: it predicts black-holes, objects so massive that spacetime
breaks down and singularities appear.

Multiple attempts have been made to build a consistent quantum theory of
gravity. String theory is maybe the most successful but we think it is
fair to say that there is no complete solution yet. New ideas are needed.

A different way to solve the problem is to study the properties of
General Relativity. The hope is that a better understanding of the
classical theory will impose constraints on the form of the quantum
version or even hint at a solution. Easier said than done. The equations of motion of General
relativity are non-linear partial differential equations and those are
prone to very peculiar and complicated solutions.

The usual tool to deal with those kind of problems or at least to
obtain some control is the use of symmetries. In themselves, they
simplify the analysis but due to the famous Noether theorem, they also
give
rise to conserved quantities. In the most favorable scenario, this can even
lead to a complete analytic solution of the problem.

The concept of symmetry is central in modern physics and you will
encounter it throughout this work but it is not the focus point of our
analysis. As the title says, we are interested in another kind of symmetries: the concept of
dualities. The symmetries described above are internal symmetries,
dualities can be considered as ``external'' symmetries: they are
symmetries between different theories. Two theories are said to be
dual if they describe the same problem through different means, the
duality providing the dictionary to go from one description to the
other.

Dualities can be very powerful. For instance, some questions can be
very difficult to handle in one description and trivial in the dual
one. Those dualities can first be studied on the classical level but one can also hope for a dual version
of gravity that would be quantizable.

In this work, we considered two dualities. The first one is known as
the electromagnetic duality. Discovered in Maxwell's theory, this
duality exchanges the role of electric and magnetic
fields. Recently, it has been extended to gravity at least at the
linearized level. The second one is the gauge/gravity
correspondance. It links a gravity theory in $n$ dimensions with a
gauge theory in $n-1$ dimensions. In the past 10 years, it received a
lot of interest and it constitutes a very active subject of research.

\section{Electromagnetic duality}

Without sources, Maxwell's equations in 4 dimensions are invariant under the
exchange of electric and magnetic fields. If we add the usual
electric sources, this symmetry is broken. To restore it,
Dirac introduced the notion of a magnetic monopole: a source for the
magnetic field playing the dual role of the electric sources. Even if
nobody has observed them experimentally, the theoretical study of
these magnetic monopoles has brought a lot of interesting results.

One of the most surprising comes from Dirac himself. He proved that
the existence of magnetic monopoles could be an explanation of the
quantization of electric charge. He showed the following quantization
condition: All electric and magnetic charge, $e_i$ and $g_j$, must
satisfy: 
\begin{equation}
e_i g_j = \frac{1}{2}n_{ij} \quad \text{where } n_{ij} \in \mathbb{Z}.
\end{equation}
In other terms, the presence of one monopole of charge $g$ anywhere in
the universe will force every electric charge to be of the form: 
\begin{equation}
\frac{n}{2g} \quad \text{where } n \in \mathbb{Z}.
\end{equation}

For a long time after the introduction by Dirac of magnetic monopoles
in 1931, these monopoles were a mere theoretical curiosity. We were able to build theories
with monopoles but we didn't have to use them to describe
nature. Everything changed in 1974 when 't Hooft and Polyakov showed
that some quantum field theories relevant in particle physics
inevitably contain magnetic monopoles. In fact, all reasonable grand
unified theories necessarily contain them \cite{Coleman:1982cx}.

How is this? A usual quantum field theory will contain a bunch of
fundamental fields: some electrically charged fermions coupled to
gauge fields like the electromagnetic field. Because we want to
describe all interactions at the same time (so the name: grand unified
theories), the gauge fields will be more complicated but in the limit
of low energy one should recover an electromagnetic field. The gauge
fields describe fundamental particles, and the only charged particles
are electrically charged. In these kinds of theories, the magnetic
monopoles cannot arise as point particles. In fact, 't Hoof and
Polyakov showed the existence of extended objects (solitons) that
behave exactly as magnetic monopoles. If you are near one of them, you
would see a complex internal structure composed of more elementary
particles but if you look at it from a distance you would only see a
magnetic monopole.

Three years later, Montonen and Olive conjectured that the theory
considered by 't Hooft and Polyakov contained electric magnetic
duality in the following sense. The duality would not be a symmetry of
the theory, it would be the link between two theories: on one hand we
would find the `electric' theory described above and on the other hand
we would find another theory (`magnetic' theory) where the role of the
components are exchanged. The magnetic monopoles would be the
elementary particles whereas the electric particles would become
extended objects. One could do computations in any of the two
theories, the results would be the same provided one uses an
appropriate dictionary between the two theories.

This duality conjectured by Montonen-Olive would be one example of
what is called a strong-weak coupling duality. For the `electric'
theory, the electric charge $e$ is the coupling constant of the
theory, the magnetic one $g$ being some constant characterizing the
soliton. For the `magnetic' theory, the situation is reversed, the
coupling constant is given by the magnetic charge of our
monopoles. Following Dirac's argument, we know that the product of
electric and magnetic charges is quantized. In particular, we have: 
\begin{equation}
eg = \frac{1}{2},
\end{equation}
where we chose $n=1$ for simplicity. If we have a weakly coupled
`electric' theory, $e$ being very small, the coupling constant $g$ of
the dual theory is very large: the theory is strongly coupled. 

Those kinds of dualities are very interesting because, usually,
strongly coupled theories are not tractable. If one is able to find a
dual theory with a small coupling constant, one would be able to use
perturbation theory to compute physical quantities and, after
computations have been done, to use the dictionary to get back the
quantities of the original theory. 

\vspace{5mm}

There are some hints that the same duality exits for General
Relativity in 4 dimensions. In particular, the Taub-NUT solution seems to be the
electromagnetic dual to the Schwarzschild black-hole\cite{Tamburino:1966ys}. Unfortunately,
any precise definition of the duality is difficult due to the
non-linear form of the equations of motion. On the other hand, the
linearized Einstein's equations describing a free spin 2 particule are
invariant under the duality and, recently, it has been shown that it
is a symmetry of the associated reduced phase space action\cite{Henneaux:2005qf}.

The goal of our work is a generalization to spin $2$ of the so-called double
potential formalism for spin $1$ fields \cite{Deser:1976ve,Schwarz:1994cr},
which has been extended so as to include couplings to dynamical dyons
by using Dirac strings \cite{Deser:1997oq,Deser:1998kx}.  In the
original double potential formalism, Gauss's constraint is solved in
terms of new transverse vector potentials for the electric field so
that electromagnetism is effectively formulated on a reduced phase
space with all gauge invariance eliminated. Alternatively, one may
choose \cite{Barnich:2008ve} to double the gauge redundancy of
standard electromagnetism by using a description with independent
vector and longitudinal potentials for the magnetic and electric
fields and $2$ scalar potentials that appear as Lagrange multipliers
for the electric and magnetic Gauss constraints. In this framework,
the string-singularity of the solution describing a static dyon is
resolved into a Coulomb-like solution. Furthermore, magnetic charge no
longer appears as a topological conservation law but as a surface
charge on a par with electric charge.

The aim of the present work is to apply the same strategy to the spin
$2$ case. Doubling the gauge invariance by keeping all degrees of
freedom of symmetric tensors now leads to a second copy of
linearized lapse and shifts as Lagrange multipliers for the new
magnetic constraints. As a consequence, the string singularity of the
gravitational dyon, the linearized Taub-NUT solution is resolved and
becomes Coulomb-like exactly as the purely electric linearized
Schwarzschild solution. Furthermore, as required by manifest duality,
magnetic mass, momentum and Lorentz charges also appear as surface
integrals. 

Our work thus presents a manifestly duality invariant alternative to
\cite{Bunster:2006fk} where the coupling of spin $2$ fields to
conserved electric and magnetic sources has been achieved in a
manifestly Poincar\'e invariant way through the introduction of Dirac
strings.

We start in chapter \ref{ch:Elecmag} with a summary of the results obtained in
the spin $1$ case. In chapter \ref{chap:spin2}, we then apply the same analysis
to the spin $2$ case. The idea is that both problems are
closely related and the reader can consider the electromagnetic
case as a physical toy model for linearized gravity. 

\vspace{5mm}

As a result of this work, a really interesting link
between electromagnetic duality and soliton theory appeared. We show in
chapter \ref{ch:Integ} that massless higher spin gauge fields are
bi-Hamiltonian and consequently integrable systems. Even if
for now this only holds at the free level, this insight relates two very
different domains of physics. Our hope is that it can allow the use
of the techniques of soliton theory in the study of fundamental
fields.

\section{Gauge/Gravity conjecture}

The idea of a duality between gauge theories and gravitational
theories was first proposed by Maldacena
in his famous paper \cite{Maldacena:1998cr}. At some level, this duality relates a
supergravity theory in anti-de Sitter background to a gauge theory
defined on its boundary. Since then, people believe that the duality
extend to any background.

One example where this idea was most successful is the $3$-dimensional
case. In 1985, Brown and Henneaux showed that the infinitesimal symmetries of
asymptotically $AdS_3$ spacetimes
\cite{Henneaux:1985kl,Brown:1986zr,Henneaux:1985tg} provide a
representation of the algebra of conformal Killing vectors of the flat 
boundary metric. In this case, the boundary being $2$-dimensional,
this algebra is infinite dimensional. This was surprising as people were expecting this
symmetry group to be just $SO(2,2)$ the exact symmetry group of
$AdS_3$. Using Hamiltonian methods, they also showed that the algebra
of surface charges form a centrally extended representation of the
$2$-dimensional algebra. Those results imply that any consistent
quantum theory of gravity in $3$ dimension should be a conformal
theory. 

The presence of this infinite-dimensional algebra allows the use of
the powerful techniques of 2-dimensional conformal field theory. Maybe
the most interesting result was then obtained by Strominger \cite{Strominger:1998fk}. Using the Cardy
formula for the growth of states in a 2D conformal theory, he was able
to reproduce the value of the Bekenstein-Hawking area formula for the
entropy of BTZ black-holes.

Historically, the first example where the asymptotic symmetry group is
enhanced with respect to the isometry group of the background metric
and becomes infinite-dimensional is the one of asymptotically flat
$4$-dimensional spacetimes at null infinity
\cite{Bondi:1962zr,Sachs:1962ly,Sachs:1962fk}.  The
asymptotic symmetry group of non singular transformations is the
well-known Bondi-Metzner-Sachs group. It consists of the semi-direct
product of the  Lorentz group, times the infinite-dimensional abelian
normal subgroup of so-called supertranslations.

The starting point of our work is the following observation. If one focuses on infinitesimal
transformations and does not require the associated finite
transformations to be globally well-defined, then there is a further enhancement. The symmetry algebra is
then the semi-direct sum of the infinitesimal local conformal
transformations of the $2$-sphere with the abelian ideal of
supertranslations, and now both factors are infinite-dimensional.

The aim of the present work is to reconsider from the point of view of
local conformal transformations the $4$-dimensional case which is, in
some sense at least, of direct physical relevance. In particular, we
provide a detailed derivation of the natural generalization of the
$\mathfrak{bms}_4$ algebra discussed above. No modification of well
studied boundary conditions is needed and the transformations are
carefully distinguished from conformal rescalings. 

A major motivation for our investigation comes from Strominger's
derivation described above. More recently, a similar analysis has been used to derive
the Bekenstein-Hawking entropy of an extreme 4-dimensional Kerr black
hole \cite{Guica:2009fk}. One of our hopes is to make progress along
these lines in the non extreme case, either directly from an analysis
at null infinity or by making a similar analysis at the horizon, as
discussed previously for instance in
\cite{Carlip:1994hc,Solodukhin:1999ij,Carlip:1999bs,Carlip:1999fv,%
  Park:1999dz,Park:2002fu,Sachs:2001kl,Koga:2001qa,Carlip:2002mi,%
Silva:2002pi,Kang:2004ff,Koga:2007lh,Koga:2008fu}.

Related work includes \cite{Ashtekar:1981ye,Ashtekar:1987tt} on
asymptotic quantization where for instance the implications of
supertranslations for the gravitational $S$-matrix have been
discussed. Asymptotically flat spacetimes at null infinity in higher
spacetime dimensions have been investigated for instance in
\cite{Hollands:2005qo,Hollands:2003tw,Hollands:2004il,Tanabe:2010zt},
while various aspects of holography in $4$ dimensions have been
studied in some details
in~\cite{Susskind:1998jl,Polchinski:1999gb,Boer:2003mb,%
  Arcioni:2003cq,Arcioni:2004kh,Solodukhin:2004rq,Gary:2009fc}. In
particular, a symmetry algebra of the kind that we derive and study
here has been conjectured in \cite{Banks:2003ss}.

We will start in chapter \ref{chap:ads3} with an analysis of the
$AdS_3$ case. We basically rederive the results of Brown-Henneaux using a
different setup. In chapter \ref{chap:bms3}, we will do the same analysis for
spacetimes that are asymptotically flat at null infinity in 3
dimensions. Finally, in chapter \ref{chap:bms4}, we present our results
on $\mathfrak{bms}_4$. The reader should view the first two cases as intermediate
problems because the $\mathfrak{bms}_4$ case has features in common with both the
$AdS_3$ and the $\mathfrak{bms}_3$ cases.

\section{Conventions }

We will work with natural units: $c=1$, $\epsilon_0=1$ and $\hbar =1$. In the following, we will use 
index notation for vectors.  Summation over repeated indices will be
understood. The 1-forms $dx^\mu$ are treated as fermionic variables
and the skew-symmetric epsilon tensor in $n$-dimension is defined as $\epsilon_{0....(n-1)}=1$.

\section{Summary of original results}

\begin{itemize}
\item Introduction of an extended double potential formalism for spin
  2 and definition of surface charges in a duality invariant way.
\item Description of massless integer spin gauge fields as
  bi-Hamiltonian systems
\item Exact representation by bulk vector fields of the symmetry
  algebra of asymptotically $AdS_3$ spacetimes through the
  introduction of a modified Lie bracket. 
\item Definition of an extended $\mathfrak{bms}_4$ algebra with both supertranslation and superrotations:
non trivial extension of the Poincar\'e algebra containing Virasoro
algebra due to the presence of gravity.
\item Derivation of the $\mathfrak{bms}_4$ charge algebra containing a field dependent central extension.
\end{itemize}

The publications containing those results are
\begin{enumerate}
\item
G.~Barnich and C.~Troessaert, ``{Manifest spin 2 duality with electric and
  magnetic sources},'' {\em JHEP} {\bf 01} (2009) 030,
\href{http://www.arXiv.org/abs/0812.0552}{{\tt 0812.0552}}.

\item
G.~Barnich and C.~Troessaert, ``{Duality and integrability: Electromagnetism,
  linearized gravity and massless higher spin gauge fields as bi- Hamiltonian
  systems},'' {\em J. Math. Phys.} {\bf 50} (2009) 042301,
\href{http://www.arXiv.org/abs/0812.4668}{{\tt 0812.4668}}.

\item
G.~Barnich and C.~Troessaert, ``{Symmetries of asymptotically flat 4
  dimensional spacetimes at null infinity revisited},'' {\em Phys. Rev. Lett.}
  {\bf 105} (2010) 111103,
\href{http://www.arXiv.org/abs/0909.2617}{{\tt 0909.2617}}.

\item
G.~Barnich and C.~Troessaert, ``{Aspects of the BMS/CFT correspondence},'' {\em
  JHEP} {\bf 05} (2010) 062,
\href{http://www.arXiv.org/abs/1001.1541}{{\tt 1001.1541}}.

\item
G.~Barnich and C.~Troessaert, ``{Supertranslations call for superrotations},''
  \href{http://www.arXiv.org/abs/1102.4632}{{\tt 1102.4632}}.

\item
G.~Barnich and C.~Troessaert, ``{BMS charge algebra},''
  to appear soon
\end{enumerate}

\chapter{Double-potential formalism for spin 1} \label{ch:Elecmag}

In this chapter, we will present the double potential formalism for
electromagnetism. As said in the introduction, the reader should view
this as a physical toy model for the
spin 2. Indeed, the ideas are the same in the two cases; everything is
just more complicated for linearized gravity.

We start by a quick introduction to electromagnetic duality for
spin 1. Then, we describe the double potential formalism used to write
a duality invariant action. Finally, we introduce the extended double
potential formalism developed in \cite{Barnich:2008ve}  in order to couple to both electric
and magnetic sources in a duality invariant way.

\section{Electro-magnetic duality}

Classically, electromagnetism is described by two fundamental vector
fields: the electric field and the magnetic field. The dynamics of
these fields and their interaction with charged particles is governed
by the well known Maxwell equations (see \cite{Felsager:1981fk},\cite{Song:1996ys}). In empty space, they take the following form:
\begin{eqnarray}
\vec{\nabla} \cdot \vec{E} & = & 0,\\
\vec{\nabla} \cdot \vec{B} & = & 0,\\
\vec{\nabla} \times \vec{E} + \frac{\partial \vec{B}}{\partial t} & = & 0,\\
\vec{\nabla} \times \vec{B} - \frac{\partial \vec{E}}{\partial t} & = & 0.
\end{eqnarray}

 These equations possess an unusual symmetry consisting of the exchange of the electric and magnetic field. More precisely, the following map leaves the above equations invariant:
\begin{equation}
\vec{E}' = -\vec{B}, \quad \vec{B}' = \vec{E}.
\end{equation}
In addition to leaving the equations invariant, this map also doesn't change physical quantities like the total energy or the total momentum of the electromagnetic field:
\begin{eqnarray}
H & = & \int_V \, d^3x \, \frac{1}{2}\left( \vert \vec{E}\vert^2 + \vert \vec{B}\vert^2\right),\\
\vec{P} & = &  \int_V \, d^3x \, \vec{E} \times \vec{B},
\end{eqnarray}
where both integrals are evaluated over the entire space $V$. 
This symmetry is called electromagnetic duality.

\vspace{5mm}

Nature seems to break this symmetry. Indeed, the only charged matter
we are aware of is electrically charged: it only sources the electric
field. In the presence of these sources, the Maxwell equations become 
\begin{eqnarray}
\vec{\nabla} \cdot \vec{E} & = & \rho_e,\\
\vec{\nabla} \cdot \vec{B} & = & 0,\\
\vec{\nabla} \times \vec{E} + \frac{\partial \vec{B}}{\partial t} & = & 0,\\
\vec{\nabla} \times \vec{B} - \frac{\partial \vec{E}}{\partial t} & = & \vec{j}_e,
\end{eqnarray}
where $\rho_e$ and $J_e$ are respectively the charge and the current
densities of electric particles. This consistency of the above
equations implies the conservation of the electric charge: $\d_0
\rho_e+ \d_i j_e^i=0$. The duality map given above does not leave these equations invariant any longer. To restore the symmetry, Dirac postulated the existence of magnetic monopoles: particles that source the magnetic field. If these magnetic monopoles are present, the electromagnetic equations will take the form:
\begin{eqnarray}
\vec{\nabla} \cdot \vec{E} & = & \rho_e,\\
\vec{\nabla} \cdot \vec{B} & = &  \rho_m,\\
\vec{\nabla} \times \vec{E} + \frac{\partial \vec{B}}{\partial t} & =
& - \vec{j}_m,\\
\vec{\nabla} \times \vec{B} - \frac{\partial \vec{E}}{\partial t} & = & \vec{j}_e,
\end{eqnarray}
with $\rho_m$ and $\vec j_m$ respectively the charge and the current densities of magnetic particles. In this setting, one can see that, provided we also exchange the sources, the equations are again invariant under the duality:
\begin{eqnarray}
\vec{E}' = -\vec{B}, & \quad & \vec{B}' = \vec{E},\\
\vec{j_e}' = -\vec{j_m}, & \quad & \vec{j_m}' = \vec{j_e},\\
\rho_e' = -\rho_m, & \quad & \rho_m' = \rho_e.
\end{eqnarray}

\vspace{5mm}

The simplest solution to these equations is called a dyon: a
point-particule carrying both electric and magnetic charges,
respectively $e$ and $g$. The electromagnetic field induced by this
source is given by twice the well known Coulomb solution:
\begin{eqnarray}
\vec{E} & = & \frac{e}{4 \pi} \frac{\vec{r}}{r^3},\\
\vec{B} & = & \frac{g}{4 \pi} \frac{\vec{r}}{r^3}.
\end{eqnarray}

The above results can easily be written in term of the usual 4 dimensional
quantities. The
electric and magnetic tensors are put together in a 2-form
$F=F_{\mu\nu} dx^\mu dx^\nu$:
\begin{equation}
F_{i0} = E_i, \qquad F_{ij} = \epsilon_{ijk} B^k.
\end{equation}
Using this notation, one can easily write the Maxwell equations in a
manifestly Lorentz-invariant way:
\begin{eqnarray}
\half \epsilon^{\mu\nu\rho\sigma} \partial_\nu F_{\rho\sigma} & =& j^\mu_m,\label{eq:electroEOMrel1}\\
\partial_\mu F^{\nu\mu} & = &  j^\nu_e,\label{eq:electroEOMrel2}
\end{eqnarray}
where we have defined the electric and the
magnetic currents respectively as $j^\mu_e=(\rho_e,j_e^i)$ and
$j^\mu_m = (\rho_m,j_m^i)$.
Remark that for consistency of the equations of motion (\ref{eq:electroEOMrel1}) and (\ref{eq:electroEOMrel2}), both currents must be conserved.
The duality transformation on the electromagnetic fields takes the following form:
\begin{equation}
F'_{\mu\nu} = \star F_{\mu\nu} =  \frac{1}{2} \epsilon_{\mu\nu\rho\sigma} F^{\rho\sigma}.
\end{equation}

\vspace{10mm}

At the level of the classical equations of motion, the duality is well
behaved and the magnetic monopoles can be added quite easily. The next
questions are: How can one describe this in term of an action? and
Is it possible to write an action for electromagnetism in presence of both
electric and magnetic sources? This was done by Dirac in his famous paper
\cite{Dirac:1948vn} by introducing string-like singularities
known as Dirac strings.

If the magnetic sources are absent, Maxwell's equations
become
\begin{eqnarray}
\label{eq:Bianchi}
d F &=&0\\
\label{eq:EOM}
\d_\nu F^{\mu\nu} & = & j^\mu_e
\end{eqnarray}
in term of the external differential $d=dx^\mu
\partial_\mu$. Using the Poincar\'e Lemma in $\RR^4$
this identity implies the existence of a 1-form $A=A_\mu dx^\mu$ such
that $F=dA$ or $F_{\mu\nu} = \partial_\mu A_\nu - \partial_\nu
A_\mu$. Note that there is a redundancy in the description, the
transformation $\delta A_\mu  =\d_\mu \epsilon (x^\mu)$ leaves the tensor
$F_{\mu\nu}$ invariant for any function $\epsilon$. Such a
transformation is called a gauge transformation. The usual action
for electromagnetism is build from this vector potential $A_\mu$:
\begin{equation}
  \label{eq:actionEM}
S[A_\mu] = \int d^4x \, \left( -\frac{1}{4} F^{\mu\nu} F_{\mu\nu}
  +A_\mu j^\mu_e\right).
\end{equation}
When using the description of electromagnetism in terms of $A_\mu$, we
have chosen a side; all Maxwell's equations are not on the same
footing. Half of them (\ref{eq:Bianchi}) are consistency equations
known as the bianchi identities, they are coming from the identity
$d^2=0$ and the definition $F=dA$. The other half (\ref{eq:EOM}) are
dynamical equations coming from the variation of the action (\ref{eq:actionEM}).

\vspace{5mm}

Adding magnetic sources is problematic. The above
parametrisation is based on the identity $dF=0$ but in presence of
magnetic sources, it is no longer valid. The solution to this problem
was found by Dirac in his famous paper \cite{Dirac:1948vn}. The idea
is to introduce a parametrisation that is valid everywhere except on
some points where it must be corrected. The
electromagnetic field is now defined as 
\begin{equation}
F = dA + G.
\end{equation}
The 2-form $G$ is a function of the magnetic source,
zero nearly everywhere and such that
\begin{equation}
\star dG = j_m
\end{equation}
If $j_m$ is produced by a set of point-like monopoles,
$G_{\mu\nu}$ can be chosen to be non-zero on a set of strings (called
Dirac's strings), each one of them
attached to one of the monopoles and going to infinity. Notice that
the potential $A_i$ must be singular along the strings to compensate
the infinity of $G$ and to produce a $F$ regular outside the source.

In the case of one monopole sitting at the origin, the potential 4-vector and the associated string term are given by:
\begin{eqnarray}
A_0& =& 0,\\
\label{eq:potentialmonopole}
A_i & = & \frac{g}{4\pi} \left( \frac{y}{r(r-z)},-\frac{x}{r(r-z)},0\right),\\
G_{ij} & = & \epsilon_{ijz}  g \delta(x)\delta(y)\Theta(z),\\
G_{0i} & = & 0,
\end{eqnarray}
where $\delta(x)$ and $\Theta(z)$ are the Dirac and the Heaviside functions. The magnetic field produced by the potential (\ref{eq:potentialmonopole}) alone would be the following:
\begin{equation}
\tilde{B}^i  = \epsilon^{ijk}\partial_jA_k = \frac{g}{4\pi} \frac{x^i}{r^3} - g \delta(x)\delta(y)\Theta(z)\delta^{iz}.
\end{equation}
Without the string term, this magnetic field $\tilde{B}^i$ is the one
created by a semi-infinite solenoid which is infinitely thin; it describes a string of infinitely concentrated magnetic flux that spreads out from the origin to infinity and creates the magnetic flux of the monopole.
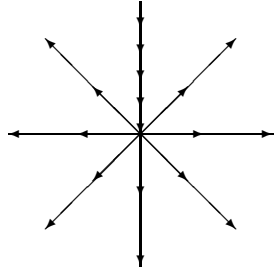
\begin{center}
\setlength{\unitlength}{2pt}
\begin{picture}(80,50)

\put(40,30){\vector(0,-1){5}}
\put(40,35){\vector(0,-1){5}}
\put(40,40){\vector(0,-1){5}}
\put(40,45){\vector(0,-1){5}}
\put(40,50){\vector(0,-1){5}}

\put(40,25){\vector(0,-1){25}}
\put(40,25){\vector(1,0){25}}
\put(40,25){\vector(-1,0){25}}
\put(40,25){\vector(1,-1){18}}
\put(40,25){\vector(-1,-1){18}}
\put(40,25){\vector(1,1){18}}
\put(40,25){\vector(-1,1){18}}

\put(40,25){\vector(0,-1){12}}
\put(40,25){\vector(1,0){12}}
\put(40,25){\vector(-1,0){12}}
\put(40,25){\vector(1,-1){9}}
\put(40,25){\vector(-1,-1){9}}
\put(40,25){\vector(1,1){9}}
\put(40,25){\vector(-1,1){9}}

\end{picture}
\end{center}
The string term is then added to remove this semi-infinite solenoid and get the magnetic monopole alone. 
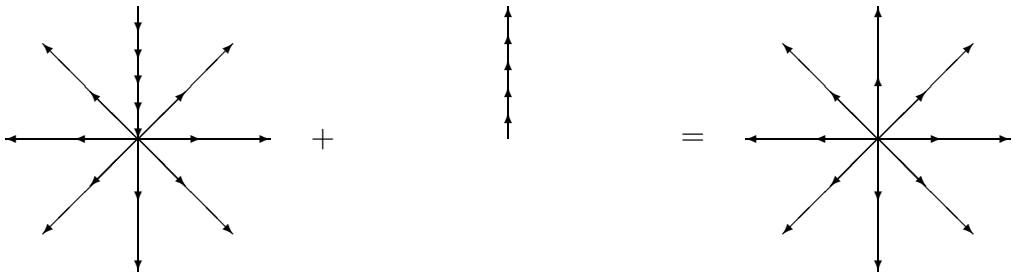
\begin{center}
\setlength{\unitlength}{2pt}
\begin{picture}(240,50)

\put(40,30){\vector(0,-1){5}}
\put(40,35){\vector(0,-1){5}}
\put(40,40){\vector(0,-1){5}}
\put(40,45){\vector(0,-1){5}}
\put(40,50){\vector(0,-1){5}}

\put(40,25){\vector(0,-1){25}}
\put(40,25){\vector(1,0){25}}
\put(40,25){\vector(-1,0){25}}
\put(40,25){\vector(1,-1){18}}
\put(40,25){\vector(-1,-1){18}}
\put(40,25){\vector(1,1){18}}
\put(40,25){\vector(-1,1){18}}

\put(40,25){\vector(0,-1){12}}
\put(40,25){\vector(1,0){12}}
\put(40,25){\vector(-1,0){12}}
\put(40,25){\vector(1,-1){9}}
\put(40,25){\vector(-1,-1){9}}
\put(40,25){\vector(1,1){9}}
\put(40,25){\vector(-1,1){9}}

\put(75,25){\makebox(0,0){$+$}}

\put(110,30){\vector(0,1){5}}
\put(110,35){\vector(0,1){5}}
\put(110,40){\vector(0,1){5}}
\put(110,45){\vector(0,1){5}}
\put(110,25){\vector(0,1){5}}

\put(145,25){\makebox(0,0){$=$}}

\put(180,25){\vector(0,1){25}}
\put(180,25){\vector(0,-1){25}}
\put(180,25){\vector(1,0){25}}
\put(180,25){\vector(-1,0){25}}
\put(180,25){\vector(1,-1){18}}
\put(180,25){\vector(-1,-1){18}}
\put(180,25){\vector(1,1){18}}
\put(180,25){\vector(-1,1){18}}

\put(180,25){\vector(0,1){12}}
\put(180,25){\vector(0,-1){12}}
\put(180,25){\vector(1,0){12}}
\put(180,25){\vector(-1,0){12}}
\put(180,25){\vector(1,-1){9}}
\put(180,25){\vector(-1,-1){9}}
\put(180,25){\vector(1,1){9}}
\put(180,25){\vector(-1,1){9}}
\end{picture}
\end{center}
In this example, we chose to place the string along the positive $z$
axis. This choice is arbitrary, we could have placed it anywhere. The
change in $A_\mu$ following the change of the position of the string
is given by a gauge transformation. In other words, the string is not
an observable physical object, it is an artefact of the
parametrisation used to describe electromagnetism. It is just a sign
of the fact that the usual
description in term of a vector potential $A_\mu$ is not adapted to
the duality. Usually, the vector potential for the
monopole is written in spherical coordinates:
\begin{equation}
\label{eq:monopspheric}
A_\mu dx^\mu = \frac{g}{4 \pi} (-1 - \cos \theta ) d\phi.
\end{equation}
In that case, the singularity is hidden in the shortcomings of the
parametrisation. Indeed, the spherical parametrisation of $\RR^3$ is
not well defined at the poles $\theta=0,\pi$ corresponding to the $z$
axis. 

Using this new definition for $F_{\mu\nu}$, one can use the same action as before to describe electromagnetism:
\begin{equation}
S\left[A_\mu\right] =\int_V \, d^4x\, \left(-\frac{1}{4}   F^{\mu\nu}F_{\mu\nu} + A_\mu j^\mu_e\right).
\end{equation}
 The two Maxwell equations concerning magnetic sources will follow
 from the definition of $F_{\mu\nu}$ and $G_{\mu\nu}$, the other two
 will be the equations of motion coming from the action. One can see
 the difference in the treatment of the magnetic and the electric
 part. The `magnetic' equations are imposed by construction, the
 `electric' equations are imposed dynamically. This description of
 electromagnetism is not duality invariant. 

The above action generates the correct Lorenz force for charged
particules provided Dirac's veto holds. This veto, stating that two
strings cannot cross each other, also implies the quantization of electric charge
described in the introduction.

\section{Double potential formalism}

At the outset, the electromagnetic duality is a symmetry of the equations of motion
but not of the action. As we saw, the formalism used to describe the
theory breaks the duality. More directly, one can write the action
without sources in term of the electric and magnetic fields:
\begin{equation}
S = \int d^4 x\, \half \left(E^2 - B^2 \right).
\end{equation}
This is clearly not invariant under the exchange $E'=-B$ and $B'=E$. 

An important question is whether one can construct an action which is
invariant under the duality. We will now describe the construction of such an
action. This construction is based on the Hamiltonian
formalism and, as such, is not manifestly invariant under the Lorentz
group.

\subsection{Canonical formulation of electromagnetism}
\label{sec:canonelec}

The usual Hamiltonian action without sources is given by
\begin{eqnarray}
\label{eq:ELECusualhamilaction}
S[E^i, A_i,A_0] &= &\int d^4x \, \left( - E^i \dot A_i - \cH_{EM} - A_0 \cG
\right),\\
\cH_{EM} & = & \half \left( E^2 + B^2 \right),\\
\cG & = & \d_i E^i,
\end{eqnarray}
where we have used the notation $B^i = \epsilon^{ijk} \d_j A_k$. The function $\cH_{EM}$
is the Hamiltonian density, $\cG$ is the Gauss constraint, the generator of
the gauge transformation, and associated to it we have $A_0$ which is
the Lagrange multiplier. The associated Poisson bracket can be read
easily from the kinetic term:
\begin{equation}
\left \{F,G \right\} = \int d^3x \left( \frac{\delta F}{\delta E^i(x)}
  \frac{\delta G}{\delta A_i(x)} - F \leftrightarrow G\right).
\end{equation}
One can check easily that this action generates the Maxwell equations.
Furthermore, it reduces to the usual Lagrangian action when the
auxiliary field $E^i$ is solved for.

\subsection{Vector decomposition}

In the rest
of the first part, we will assume that we have
boundary conditions on all our fields such that the Laplacian operator
$\Delta = \d_i \d^i$ is invertible; in other words, the Laplace
equation $\Delta \phi = \psi$ has one and only one solution for each
allowed value of the field $\psi$ (where $\phi$ and $\psi$ stands for any field under
consideration).  Using this property, any vector $V^i$ can be divided in two
parts: a longitudinal and a transverse part as follows:
\begin{eqnarray}
V^i &=& V^{Ti} + V^{Li},\\
V^{Li} & =& \d^i \Delta^{-1} \d_j V^j,
\end{eqnarray}
with the transverse part $V^T$ defined as the reminder $V^T=V-V^L$. We
have $\d_i V^{Ti} = 0$.

We can apply the Poincar\'e lemma to
introduce a potential for the transverse part of $V$: 
\begin{equation}
V^{TI} = \cO W^i = \epsilon^{ijk} \d_j W_k.
\end{equation}
 Remark that the curl operator $\cO$ kills
the longitudinal part of $W^i$. In that sense, all the interesting information
is stored in the transverse part of $W$ wich is unique and given by  $W^{Ti} =
-\cO \Delta^{-1} \cO V^{Ti}$. From the
three components of a general vector in 3D, two are stored in
the transverse part and one is in the longitudinal part.

The elements of the decomposition are orthogonal under 
integration if boundary terms can be neglected,
\begin{equation}
  \int d^3 x\, V^i W_i = \int d^3 x 
  \left( V^{Ti} W^T_i + V^{Li} W^L_i\right).
\end{equation} 
and the operator $\cO $ is
self-adjoint, e.g., 
\begin{eqnarray}
  \int d^3 x\, \left( \cO V\right)^iW_i=
\int d^3 x\, V^i\left( \cO W\right)_{i}.
\end{eqnarray}
\subsection{Reduced phase space and duality invariant action}

The Gauss constraint implies that $E^i$ must be transverse.
If we solve the constraint by putting $E^L$ to zero, we must also fix the gauge freedom
associated to it. This gauge freedom is $\delta A_i = \d_i \epsilon$
and it can be used to fix the longitudinal part of $A$ to
zero. The resulting action is
\begin{eqnarray}
S[E^{Ti}, A^T_i] &= &\int d^4x \, \left( - E^{Ti} \dot A^T_i - \cH
\right),\\
\cH & = & \half \left( E^{T2} + B^2 \right).
\end{eqnarray}
This is the reduced phase-space formulation of electromagnetism: all
the constraints and gauge freedom have been removed. The only degrees
of freedom left are the 2 physical degrees of freedom describing the
photon. They are stored in the canonical pairs $(A^T_i, -E^{Ti})$.

\vspace{5mm}

 As we saw
above,  we can then introduce a new potential $Z^{Tk}$ to describe the
transverse part of the electric field:
\begin{equation}
E^i = \cO Z^{Ti} = \epsilon^{ijk} \d_j Z^T_k.
\end{equation}
Putting this in the action gives
\begin{eqnarray}
S[A^T_i, Z^T_i] & = & \int d^4x\, \left\{-\cO Z^{Ti} \dot A^T_i
  - \half \left( \left(\cO Z^{T}\right)^2 + \left(\cO A^T\right)^2 \right) \right\}.
\end{eqnarray}
This action is already duality invariant \cite{Deser:1976ve} but the symmetry can be made
more manifest if we rename our fields as $A^{Ta}_i = (A^T_i, Z^T_i)$ and
$B^{ai} = (B^i, E^i)$ with
$a=1,2$ \cite{Schwarz:1994cr}. After some integrations by part, the action becomes
\begin{eqnarray}
\label{eq:EMfixedgauge}
S[A^{Ta}_i] & = & \int d^4x\, \left\{\epsilon_{ab}\half \cO
  A^{Tai} \dot A^{Tb}_i  - \cH \right\},\\
\cH & = & \half B^{ai} B_{ai},\\
B^{ai} & = & \cO A^{Tai},
\end{eqnarray}
where $\epsilon_{ab}$ is antisymmetric with $\epsilon_{12} = 1$
($a,b,...$ are moved up and down using the Kronecker symbol $\delta_{ab}$). The
electromagnetic duality is the $SO(2)$ rotation acting on the duality
index:
\begin{equation}
\delta_D A^{Ta}_i = \epsilon^{ab}A^T_{bi}.
\end{equation}
The fact that the action is invariant is
obvious because every term is a scalar under this
rotation. The duality generator is
\begin{equation}
D = -\half \int d^3x A^{Tai} \cO A^T_{ai}
\end{equation}
which is simply an $SO(2)$ Chern-Simons action \cite{Deser:1997fk}.

\subsection{Poincar\'e transformations}

To obtain a duality invariant action, we had to use the Hamiltonian
formalism and thus loose the manifest Lorentz invariant of the action. To be
precise, we still have manifest invariance under the spatial part of the group
but not for the boosts. This is a common fact in the Hamiltonian
formalism: the 3+1 split breaks manifest Lorentz invariance because time has become a
special coordinate. On the other hand, going to the Hamiltonian
formulation doesn't change the symmetries of the action (see e.g. \cite{Henneaux:1992fk}). The boosts
should be hidden somewhere. Tracing them back from the Lagrangian
formulation, the Poincar\'e generator associated to the usual
electromagnetic action (\ref{eq:ELECusualhamilaction}) are given by
\begin{eqnarray}
\label{eq:EMpoincaregen1}
\cQ_G(\omega, a) & = & \half \omega_{\mu\nu} J^{\mu\nu} - a_\mu P^\mu
\end{eqnarray}
where $P^0$ is the time translation given by the Hamiltonian, $P^i$
are the spatial translations, $J^{\mu\nu}$ the spatial rotations and
the boosts. They can be constructed from the symmetric
energy-momentum tensor as follows:
\begin{equation}
\label{eq:EMpoincaregen4}
T^{00}  =H_{EM}, \qquad T^{i0} =  \epsilon^{ijk} E_j B_k,
\end{equation}
\begin{equation}
\label{eq:EMpoincaregen5}
P^\mu = \int d^3x \, T^{\mu 0}, \qquad J^{\mu\nu} = -\int d^3x \,
\left( x^\mu T^{\nu 0} - x^\nu T^{\mu 0}\right).
\end{equation}

From these generators, one can deduce the generators of the reduced
phase space theory by evaluating them on the constraint surface. The
generators of the double potential formalism are then obtained by
introducing the new potential. The final result is given by
(\ref{eq:EMpoincaregen1}) with
\begin{equation}
\label{eq:EMpoincaregen2}
T^{00}  =H, \qquad T^{i0} = -\half \epsilon^{ijk} \epsilon_{ab} B^a_j
B^b_k = \epsilon^{ijk} E_j B_k,
\end{equation}
\begin{equation}
\label{eq:EMpoincaregen3}
P^\mu = \int d^3x \, T^{\mu 0}, \qquad J^{\mu\nu} = -\int d^3x \,
\left( x^\mu T^{\nu 0} - x^\nu T^{\mu 0}\right),
\end{equation}
and the associated transformations
\begin{eqnarray}
\label{eq:EMpoincarefixgauge}
\delta_\xi A^{Ta}_i & =& -\epsilon^{ab} B_{bi} \xi^0 - \epsilon_{ijk}
\xi^j B^{ak},\\
\xi^\mu &= & a^\mu + \omega^\mu_{\phantom \mu \nu} x^\nu \qquad
\text{with} \qquad \omega_{\mu\nu} = -\omega_{\nu\mu}.
\end{eqnarray}

One can also show by direct computation that these generators form a
representation of the Poincar\'e algebra under the Poisson bracket
associated with the action:
\begin{equation}
\left \{ F, G\right \} = \int d^3 x \, \left(\half \frac{\delta F}{\delta
    A^a_i(x)}\epsilon^{ab} \epsilon_{ijk} \d^j \Delta^{-1} \frac{\delta G}{A^b_k(x)} \right).
\end{equation}

\section{Double potential formalism with external sources}

The action presented in the previous section concerns only the reduced phase space in absence of
sources, only the dynamical degrees of freedom of electromagnetism. It
obviously doesn't say anything about the sources. The next logical step
would be to add them. To build
an action that is invariant under the duality in presence of sources,
one has two possibilities to describe the
interactions. The first one is to treat both electric and magnetic
sources as topological information by introducing strings for
both sides \cite{Deser:1998kx}.
What we will present here is the other natural possibility: treat
both sides as dynamical information. This was done in
\cite{Barnich:2008ve} and we will mainly review the part of their
results that will be interesting in the spin 2 case.

\vspace{5mm}

Let's first have a look at the Hamiltonian formulation in presence of
electric source $j_e^\mu$ (this source must be conserved $\d_\mu
j^\mu_e = 0$). We have
\begin{eqnarray}
\label{eq:EMusualhamil}
S[A_i, E^i, A_0] & = & \int d^4 x \, \left( -E^i \dot A_i - \cH - A_0 (
  \d_i E^i  - j_e^0)\right),\\
\cH & = & \half \left( E^2 + B^2 - A_i j^i_e\right).
\end{eqnarray}
The source term modifies the Gauss constraint. The longitudinal part of $E$
is no longer zero and part of the information about the sources is
stored in the gauge sector. To add magnetic sources, we need an
equivalent mechanism for the dual part. 

\vspace{5mm}

The solution proposed by the authors of
\cite{Barnich:2008ve} is to double the gauge freedom. To do so, they
introduce the following parametrisation for the magnetic and electric
fields $B^{ai} = (B^i,E^i)$:
\begin{equation}
\label{eq:EMnewB}
B^{ai} = \epsilon^{ijk} \d_j A^a_k + \d^i C^a.
\end{equation}
The fields $A^a_i$ and $C^a$ are respectively potentials for the
transverse part and the longitudinal parts of $B^a$. The
electromagnetic action they propose is the following:
\begin{equation}
\label{eq:EMdualsources}
S[A^a_i,C^a, A^a_0] = S_{EM}[A^a_i,C^a, A^a_0] +S_I[A^a_i, A^a_0;j^{a\mu}] 
\end{equation}
where
\begin{eqnarray}
 S_{EM}[A^a_i,C^a, A^a_0]  &=& \int d^4x \, \left(\half \epsilon_{ab}
   (B^{ai} + \d^iC^a)\dot A_i^b - A^a_0 \cG_a - 
  \cH\right),\\
\cH & =& \half B^{ai}B_{ai}\\
\cG_a & = & \epsilon_{ab}\d_i B^{bi}
\end{eqnarray}
is the substitute for the usual Maxwell action and
\begin{equation}
S_I[A^a_i, A^a_0;j^{a\mu}] = \int d^4x \, \epsilon_{ab} A^a_\mu j^{b\mu}
\end{equation}
is the ``interaction'' action. The external magnetic and electric
currents are labeled as $j^{a\mu} = (j_m^\mu,j_e^\mu)$ and are
conserved, $\d_\mu j^{a\mu}=0$. This action is manifestly invariant
under duality rotations that acts on the $a$ indices
\begin{equation}
\delta_D A^a_\mu = \epsilon^{ab}A_{b\mu}, \qquad \delta_D C^a =
\epsilon^{ab} C_b, \qquad \delta_D j^{a\mu} = \epsilon^{ab} j^\mu_b.
\end{equation}

We will start by studying the electromagnetic core of this action,
namely $S_{EM}$.

\subsection{Canonical structure}

After some integration by parts, one can write the kinetic term as
\begin{equation}
\int d^4x \, \half \epsilon_{ab}
   (B^{ai} + \d^iC^a)\dot A_i^b = \int d^4x \, \left(- \cO A^{2Ti}
     \dot A^{T1}_i + \d^i C^1 \dot A^{L2}_i -\d^i C^2 \dot A^{L1}_i\right)
\end{equation}
The canonically conjugate pairs are identified as
\begin{equation}
\left(A^{T1}_i(x), - \cO A^{2Ti} (y)\right), \, \left( A^{L2}_i (x),
  \d^i C^1 (y)\right), \, \left( A^{L1}_i (x),
  -\d^i C^2 (y)\right).
\end{equation}

The usual canonical pairs of electromagnetism can be chosen in term
of the new variables as
\begin{equation}
\left(A^{T1}_i(x), - \cO A^{2Ti} (y)\right), \, \left( A^{L2}_i (x),
  \d^i C^1 (y)\right).
\end{equation}
The new canonical pair 
\begin{equation}
\left( A^{L1}_i (x),
  -\d^i C^2 (y)\right)
\end{equation}
is just the magnetic dual of the pure gauge pair $\left( A^{L2}_i (x),
  \d^i C^1 (y)\right)$ of the usual formulation.

\vspace{5mm}

Those results imply
\begin{equation}
\left\{A^{ai}(x), B^{bj} (y) \right\} = -\epsilon^{ab} \delta^{ij}
\delta^3 (x,y) , \quad \left\{B^{ai}(x), B^{bj} (y) \right\} = -\epsilon^{ab} \epsilon^{ijk}
\d_k\delta^3 (x,y).
\end{equation}

\subsection{Gauge structure and degree of freedom count}
\label{sec:elecgaugestruct}

In this action, there are 2 first-class constraints:
\begin{equation}
\cG_a = \epsilon_{ab}\d_i B^{bi} = \epsilon_{ab}\Delta C^b \approx 0.
\end{equation}
For $a=1$ we have the usual Gauss constraint $\d_i E^i = 0$. On the
other hand, $a=2$ gives its dual: the magnetic Gauss constraint $\d_i
B^i = 0$. In the usual formulation, this Maxwell equation is imposed
by the formalism. We see that in this case, it has become a dynamical
equation.

The gauge transformations generated by $\int d^3x \,
\lambda^a \cG_a$ can be easily computed:
\begin{equation}
\delta_\lambda A^a_\mu = \d_\mu \lambda^a, \qquad \delta_\lambda C^a = 0.
\end{equation}
Exactly as expected, we have the usual gauge transformation
parametrized by $\lambda^1$ and a new dual one parametrized by $\lambda^2$.
The gauge sector of the theory has been doubled in a duality invariant
way.

The new constraint $\cG_2$ can be gauge fixed through the condition
$A^{L2}_i=0$. This partially gauge fixed theory corresponds to
the usual electromagnetic theory in Hamiltonian form as described in
section \bref{sec:canonelec}. More precisely, the observables of a Hamiltonian field theory with
constraints are defined as equivalence classes of functionals that
have weakly vanishing Dirac brackets with the constraints and where
two functionals are considered as equivalent if they agree on the
surface defined by the constraints (see e.g.~\cite{Henneaux:1992fk}).
The new constraint together with the gauge fixing condition form
second class constraints. The Dirac bracket algebra of observables of
this (partially) gauge fixed formulation is isomorphic to the Poisson
bracket algebra of observables of the extended formulation on the one
hand, and to the Poisson bracket algebra of observables of 
usual electromagnetism on the other hand.

\vspace{5mm}

We start with 4 degrees of freedom per spacetime point. The 2
constraints, being first class, remove 2 of those degrees of
freedom. In the end, the fully gauge fixed theory contains 2 physical
degrees of freedom per spacetime point described by the transverse
vector potential $A^{T}_i=A^{T1}_i$ and its canonically conjugate variable
$-E^{Ti}=- \cO A^{2Ti}$, as it should.

\subsection{Duality generator}

The duality generator is the $SO(2)$ Chern-Simons term suitably
extended to the longitudinal potentials:
\begin{equation}
D = -\half \int d^3x \left(B^{ai} + \d^i C^a \right) A_{ai}.
\end{equation}
This generator commutes with the Hamiltonian and the other Poincar\'e
generators given below but it is only weakly gauge invariant:
\begin{equation}
\left\{ \cG_a, D\right\} =\epsilon_{ab} \cG^b.
\end{equation}

\subsection{Poincar\'e transformations}
\label{sec:ELECpoincaretrans}

Consider now a symmetry generator of the usual Hamiltonian action of
electromagnetism (\ref{eq:ELECusualhamilaction}). It is defined
by an observable $K[A,E]=K[A^1,B^2]$ whose representative is weakly
conserved in time,
\begin{equation}
  \frac{\d}{\d t} K+\{K,H_{EM}\}\approx 0.
\end{equation}
Since the new Hamiltonian differs from the usual one by terms
proportional to the new constraint, 
\begin{equation}
H=H_{EM}+\int d^3x\, \cG_2 k, \qquad k=\half C^1,
\end{equation}
and since $K$, when expressed in terms of the new variables, does not
depend on $A^{L2}$, so that $\{K,\int d^3x\, \cG_2 k\}\approx 0$ in the extended theory, it follows that $K$ is
also weakly conserved and thus a symmetry generator of the extended
theory,
\begin{equation}
  \frac{\d}{\d t} K+\{K,H\}\approx 0.
\end{equation}

Consider then the Poincar\'e generators $Q_G(\omega,a)$ of
electromagnetism (\ref{eq:EMpoincaregen1}). When
expressed in terms of the new variables, they are representatives for
the Poincar\'e generators of the extended theory. Indeed, we just have
shown that they are symmetry generators, while we have argued in
Section~\bref{sec:elecgaugestruct} that their Poisson algebra is
isomorphic when restricted to their respective constraint surfaces.

As expected, the extended theory is invariant under
Poincar\'e transformations but the generators obtained above are not
invariant under the duality. Generators being defined up to terms
proportional to the constraints, one can try to change the representative of
the Poincar\'e generators to make them invariant under the
duality. It is possible: one solution is to keep expressions
(\ref{eq:EMpoincaregen1}-\ref{eq:EMpoincaregen3}) but with the
new definition of $B^{ai}$. The associated transformations are now:
\begin{eqnarray}
\delta_\xi C^a&=&0,\\
\delta_\xi A^{a}_i & =& \d_i \lambda^a_\xi-\epsilon^{ab} B_{bi} \xi^0 - \epsilon_{ijk}
\xi^j B^{ak},\\
\delta_\xi A^a_0&=& \d_0 \lambda^a_\xi + \epsilon^{ab} B_{bi} \xi^i,
\end{eqnarray}
where
\begin{equation}
\lambda^a_\xi = -\epsilon^{ab} \Delta^{-1} \d_i \left( B^i_b
  \xi^0\right) + \Delta^{-1} \d_i \left(  \epsilon^{ijk} B^a_j \xi_k\right).
\end{equation}

\subsection{Equations of motion and dyon solution}
The equations of motion can be easily computed and shown to be
equivalent to the usual Maxwell equations. The variation of
$A^a_0$ gives the two Gauss laws:
\begin{equation}
\d_i B^{ai} = \Delta C^a = j^{a0}.
\end{equation}
The variation of $C^a$ gives
\begin{equation}
\epsilon_{ab}\Delta A^b_0 = \epsilon_{ab} \d^i \dot A^{b}_i - \Delta C_a.
\end{equation}
We are working with boundary conditions such that $\Delta$ is
invertible. With this assumption, both $A^a_0$ and $C^a$ are auxiliary
fields in the sense that their equations of motion can be solved to
express them in term of the other fields without the need for initial conditions.
The variation of $A^a_i$ gives the remaining Maxwell's equations
\begin{equation}
\label{eq:ELECEOM189}
-\epsilon_{ab} \dot B^{bi} + \epsilon^{ijk} \d_j B_{ak} =
\epsilon_{ab} j^{bi}.
\end{equation}

\vspace{5mm}

The simplest non-trivial source is a dyon sitting at the origin with
charges $Q^a=(P,Q)$:
\begin{equation}
j^{a\mu} = 4 \pi Q^a \delta^\mu_0 \delta^3(x).
\end{equation}
The Maxwell's equation in the above form are now solved by
\begin{equation}
A^a_\mu = -\frac{\epsilon^{ab} Q_b}{r} \delta^0_\mu, \qquad C^a = -\frac{Q^a}{r}.
\end{equation}
In this set-up, the string-like singularity has been completely removed.

\subsection{Surface charges}

The purpose of this extended double potential formalism is to have both sectors of the theory on
the same footing and to be able to compute both the electric and the
magnetic charge as dynamical conserved quantities. Doubling of the
gauge freedom gives exactly that: two gauge freedoms imply two surface
charges. In the original paper, the authors used the
techniques developed in \cite{Barnich:2002fk,Barnich:2003fk,Barnich:2008uq} to compute those
charges and by coupling this theory
to gravitation and were able to derive the first law for R-N
black holes charged with both electric and magnetic 
sources. Unfortunately, in the spin 2 case, the theory will not be
local as a Hamiltonian gauge theory and those techniques will not be
available. We thus revert to the original Hamiltonian method of
\cite{Regge:1974kx,Benguria:1977fk}. We refer the reader to appendix 
\ref{app:regge-teit-revis} for a quick summary adapted to the problem
at hand where there is no need to discuss fall-off conditions.

The analysis of appendix 
\ref{app:regge-teit-revis} is not directly applicable to our case
since we do not have Darboux coordinates and the Poisson brackets of
the fundamental variables are non local. In the spin 1 case, everything is
under control and well behaved in the end. But, because it will not be the
case for the spin 2, we will spend some time here to present the
method used in the following chapter to deal with the surface charges
of linearized gravity in the double potential formalism. 

The
non-locality brings two problems: $\delta_{\epsilon_s} z^A=0$ may not
have non trivial solutions (see spin 2), and $\Delta^{-1}$ applied to
localized sources will spread them throughout space. In view of this,
the idea is to redo the analysis of appendix 
\ref{app:regge-teit-revis} while keeping the sources
explicitly throughout the argument.

In the presence of the sources, the constraints are
\begin{equation}
\cG_a = \epsilon_{ab} \d_i B^{bi} - \epsilon_{ab} j^{b0}.
\end{equation}
Instead of (\ref{eq:app29c}), we can write
\begin{equation}
\lambda^a \cG_a = -\d_i \lambda^a\epsilon_{ab} B^{bi} -
\lambda^a\epsilon_{ab} j^{b0} - \d_i \tilde k^i_\lambda[z^A]
\end{equation}
with
\begin{equation}
\tilde k^i_\lambda[z^A]=-\lambda^a\epsilon_{ab} B^{bi}.
\end{equation}
Consider now gauge parameters $\lambda^a(x)$ satisfying
\begin{equation}
\d_i \lambda_s^a = 0 \quad \Leftrightarrow \quad \lambda_s^a = cst.
\end{equation}
It follows that the surface charges
\begin{equation}
\label{eq:ELECsufracechargesspin1}
\cQ_{\lambda_s} = \oint_S d^3x_i \tilde k^i_{\lambda_s}[z_s],
\end{equation}
only depend on the homology class of $S$ outside of the sources. The
equation of motion (\ref{eq:ELECEOM189}) implies
\begin{equation}
\d_0 k^i_{\lambda_s}[z_s] = \lambda^a \epsilon_{ab}j^{bi} -
\epsilon^{ijk} \d_j \left(\lambda^a B_{ak}\right).
\end{equation}
The surface charges (\ref{eq:ELECsufracechargesspin1}) are thus
time-independent outside of the sources. Breaking them into their two
components, we obtain both the electric and the magnetic charges:
\begin{equation}
\cQ = \oint_S d^3x_i E^i, \quad \cP = \oint_S d^3x_i B^i.
\end{equation}

The result is what we expected and is what was obtained in
\cite{Barnich:2008ve} using the
covariant approach to surface charges.

\chapter{Double potential formalism for spin 2}
\label{chap:spin2}

In this chapter, we will introduce a double potential formalism for
linearized gravity in presence of both type of gravitational sources:
the ``electric'' source which is the mass or the energy and the
``magnetic'' source. The strategy is the same as the one presented
in the previous chapter for electromagnetism. The two questions are
closely related but the spin 2 case is more complicated.

We will start by a small review on gravitational electromagnetic
duality. After this introduction, the results of
Henneaux-Teitelboim \cite{Henneaux:2005qf} for the invariance of the action in absence of
sources will be presented. Finally, the last section will be devoted to
our construction of an extended double potential formalism for
linearized gravity \cite{Barnich:2009vn}.

\section{Electromagnetic duality for gravity}
\label{sec:EMspin2}

The first hint at the existence of the electromagnetic duality in
gravitation came with the discovery of the Taub-Nut solution to the
Einstein equations in 4 dimensions. The Taub-Nut metric \cite{Newman:1963vn,Bunster:2006fk} is given by
\begin{equation}
ds^2 = -V(r) \left[dt + 2N (1- \cos \theta) d\phi \right]^2 + V(r)^{-1}
dr^2 + (r^2 + N^2) (d\theta^2 + \sin^2 \theta d\phi^2),
\end{equation}
with
\begin{equation}
V(r) = 1-\frac{2 (N^2 + M r)}{(r^2 + N^2)} = \frac{r^2 -2Mr - N^2}{R^2+N^2},
\end{equation}
where $M$ is the usual mass and $N$ a second parameter. The study of
this solution showed that this $N$ plays the role of an
electromagnetic dual to $M$. A closer look at the $g_{t\phi}$ component of the metric
shows a string-like singularity quite similar to the one appearing in
the vector potential for the magnetic monopole in spherical
coordinates (\ref{eq:monopspheric}). This string-like singularity is
in this case known as the Misner string. Unfortunately, the attempt to
interpret it as a non-physical object, as an artefact coming from the
metric description of gravity, led to problems \cite{Misner:1963ys}.
The non-linear nature of gravity implies that for this singularity to
be spurious, we need the time coordinate to be periodic. In that case,
the Taub-NUT acquires some pathological behavior like time-like closed
curves. 

A lot of work has been done in an attempt to better understand the
duality in the full non-linear theory but it is still a difficult
question. On the other hand, the linear version, describing a spin 2
particule in a Minkowski background, is nicer. The theory is described
by the Pauli-Fiertz action:
\begin{equation}
S_{PF}[h_{\mu\nu}] = -\frac{1}{4} \int d^4x \left( \d^\rho h^{\mu\nu}
  \d_\rho h_{\mu\nu} - 2 \d_\mu h^{\mu\nu} \d_\rho h^\rho_\nu + 2
  \d^\mu h^\rho_\rho \d^\nu h_{\mu\nu} - \d^\mu h^\rho_\rho \d_\mu h^\sigma_\sigma\right)
\end{equation}
where $h_{\mu\nu}$ is symmetric. This action can be obtained from the
full Einstein action by linearizing around flat space, $h_{\mu\nu}$ being
the deviation of the full metric from the Minkowski metric. The
associated equations of motions can be written  as
\begin{equation}
\label{eq:GRAVEOM}
R_{\mu\nu} = R^\rho_{\phantom \rho \mu \rho \nu} = 0
\end{equation}
where $R_{\mu\nu\rho\sigma}=- R_{\nu\mu\rho\sigma} =
-R_{\mu\nu\sigma\rho}$ is the linearized Riemann tensor defined
as
\begin{equation}
R_{\mu\nu\rho\sigma} = \d_{[\mu} h_{\nu][\rho,\sigma]}.
\end{equation}
By construction, this tensor satisfies to the following identities
\begin{eqnarray}
\label{eq:GRAVcyclic}
R_{\mu[\nu\rho\sigma]} & = & 0\\
\label{eq:GRAVbianchi}
R_{\mu\nu[\rho\sigma,\alpha]} & =& 0.
\end{eqnarray}
It follows that $R_{\mu\nu\rho\sigma}$ is symmetric for the exchange
of the pairs $(\mu\nu)$ and $(\rho\sigma)$:
\begin{equation}
R_{\mu\nu\rho\sigma}=R_{\rho\sigma\mu\nu},
\end{equation}
and, using (\ref{eq:GRAVbianchi}) and (\ref{eq:GRAVEOM}), one can easily prove
\begin{equation}
\label{eq:GRAVdualbianchi}
\d^\mu R_{\mu\nu\rho\sigma} = 0.
\end{equation}
It turns out that equations (\ref{eq:GRAVcyclic}) and
(\ref{eq:GRAVbianchi}) are the conditions for the existence of a
tensor potential $h_{\mu\nu}$ for a general tensor
$R_{\mu\nu\rho\sigma} =  - R_{\nu\mu\rho\sigma} =
-R_{\mu\nu\sigma\rho}$.

\vspace{5mm}

Equations (\ref{eq:GRAVEOM}),
(\ref{eq:GRAVcyclic})-(\ref{eq:GRAVbianchi}) and (\ref{eq:GRAVdualbianchi}) are the equivalent of Maxwell's equation for
linearized gravity.
In term of the dual tensor $S_{\mu\nu\rho\sigma} =  - S_{\nu\mu\rho\sigma} =
-S_{\mu\nu\sigma\rho}$ defined as
\begin{equation}
S_{\mu\nu\rho\sigma} = -\half \epsilon_{\mu\nu\alpha\beta}
R^{\phantom{\rho\sigma} \alpha\beta}_{\rho\sigma}
\end{equation}
they will keep the same form. Indeed, the structural equations
(\ref{eq:GRAVcyclic}) and (\ref{eq:GRAVbianchi}) become
\begin{eqnarray}
S_{\mu\nu}&= &S^\rho_{\phantom \rho \mu \rho \nu} = 0\\
\d^\mu S_{\mu\nu\rho\sigma} &=& 0.
\end{eqnarray}
On the other hand, the equations of motion (\ref{eq:GRAVEOM}) and
(\ref{eq:GRAVdualbianchi}) become
\begin{eqnarray}
S_{\mu[\nu\rho\sigma]} & = & 0,\\
S_{\mu\nu[\rho\sigma,\alpha]} & =& 0.
\end{eqnarray}
The electromagnetic duality for the spin 2 is defined as the $SO(2)$ rotation
between $R_{\mu\nu\rho\sigma}$ and its dual $S_{\mu\nu\rho\sigma}$:
\begin{equation}
\delta_D R_{\mu\nu\rho\sigma} = S_{\mu\nu\rho\sigma}, \qquad \delta_D S_{\mu\nu\rho\sigma} = -R_{\mu\nu\rho\sigma}.
\end{equation}
It obviously leaves the set of equations invariant. 

\vspace{5mm}

A linearized Riemann tensor satisfying to 
(\ref{eq:GRAVcyclic}) and (\ref{eq:GRAVEOM}) can be completely parametrized by two 3
dimensional symmetric tensors $E_{ij}$ and $B_{ij}$ defined as
\begin{eqnarray}
E_{ij} = R_{0i0j}, \qquad B_{ij} = -\half \epsilon_{ikl} R_{0j}^{\phantom{0j}kl}.
\end{eqnarray}
The are called the electric and magnetic part of the Weyl tensor which
in this case is equal to the Riemann tensor. The duality
transformations are equivalent to
\begin{equation}
\label{eq:GRAVdualitypartriemann}
\delta_D E_{ij} = B_{ij}, \qquad \delta_D B_{ij} =-E_{ij}.
\end{equation}

\vspace{5mm}

The case we described until here is a free spin 2 field without any
external source. Adding the usual ``electric'' source can be done
easily by adding an interacting term to the Pauli-Fierz
action:
\begin{equation}
S[h_{\mu\nu};T^{\mu\nu}] = \frac{1}{16\pi G} S_{PF} + \half \int d^4x \, h_{\mu\nu}T^{\mu\nu}.
\end{equation}
This source is an energy-momentum tensor, it is symmetric and
conserved:
\begin{equation}
T^{\mu\nu} = T^{\nu\mu}, \qquad\d_\mu T^{\mu\nu} = 0.
\end{equation}
The new equations of motion are
\begin{equation}
\label{eq:GRAVEOMsources}
G^{\mu\nu} = 8 \pi G  T^{\mu\nu},
\end{equation}
where $G^{\mu\nu}=R^{\mu\nu} - \half \eta^{\mu\nu} R^\rho_\rho$ is the
linearized Einstein tensor. This also
implies a modification of the identity (\ref{eq:GRAVdualbianchi}):
\begin{equation}
\d_\mu R^{\mu \rho \gamma \delta} = 8 \pi G \left(\d^\gamma \bar
  T^{\rho \delta} - \d^\delta \bar T^{\rho \gamma} \right)
\end{equation}
where for a tensor $K^{\mu\nu}$, we have defined $\bar K^{\mu\nu} = K^{\mu\nu} - \half \eta^{\mu\nu}
K^\rho_\rho$. The structural equations (\ref{eq:GRAVcyclic}) and (\ref{eq:GRAVbianchi}) don't change in
presence of electric sources.
The equations are no longer invariant under the duality: for instance,
the duality will send
equation (\ref{eq:GRAVEOMsources}) to
\begin{equation}
S^{\mu\nu} -\half \eta^{\mu\nu} S^\rho_\rho = 8 \pi G  T^{\mu\nu} \qquad
\iff \qquad R_{\mu[\nu\rho\sigma]} = -\frac{8}{3} \pi G
\epsilon_{\nu\rho\sigma \delta} \bar T^\delta_ \mu.
\end{equation}

The idea to restore the duality is exactly the same in the spin 2
case as what was done by Dirac for electromagnetism: adding dual sources. The ``magnetic'' sources are represented by
a dual, magnetic, energy-momentum tensor $\Theta^{\mu\nu}$. As for
$T^{\mu\nu}$, it must be conserved: $\d_\mu \Theta^{\mu\nu}=0$. The
new equations are \cite{Bunster:2006fk}:
\begin{eqnarray}
\label{eq:GRAVEOMgene1}
G^{\mu\nu} &=& 8 \pi G  T^{\mu\nu},\\
\label{eq:GRAVEOMgene2}
R_{\mu[\nu\rho\sigma]}& =& -\frac{8}{3} \pi G
\epsilon_{\nu\rho\sigma \delta} \bar \Theta^\delta_ \mu,\\
\label{eq:GRAVEOMgene3}
R_{\mu\nu[\rho\sigma,\alpha]} & =& \frac{8}{3} \pi G
\epsilon_{\rho\sigma\alpha \beta} \left(\d_\nu \bar
  \Theta^\beta_\mu - \d_\mu \bar \Theta^\beta_\nu \right),\\
\d_\mu R^{\mu \rho \gamma \delta} &=& 8 \pi G\left(\d^\gamma \bar
  T^{\rho \delta} - \d^\delta \bar T^{\rho \gamma} \right).
\label{eq:GRAVEOMgene4}
\end{eqnarray}
The equations are symmetric under the generalized duality:
\begin{eqnarray}
\delta_D R_{\mu\nu\rho\sigma} &=& S_{\mu\nu\rho\sigma},\\
\delta_D S_{\mu\nu\rho\sigma} &=& -R_{\mu\nu\rho\sigma},\\
\delta_D T^{\mu\nu} & =&  \Theta^{\mu\nu},\\
\delta_D \Theta^{\mu\nu} & = & -T^{\mu\nu}.
\end{eqnarray}

Again, the same problem appears in this case. One can write duality
invariant equations but writing an action is more difficult. As for
electromagnetism, the necessary conditions to introduce the potential
$h_{\mu\nu}$, namely equations (\ref{eq:GRAVcyclic}) and
(\ref{eq:GRAVbianchi}), are no longer valid in presence of magnetic
sources. In \cite{Bunster:2006fk}, the authors proposed a solution by
introducing string-like terms. Following the ideas of Dirac, they managed to write an action
invariant under Poincar\'e containing both electric and
magnetic external sources. They also derived a quantification condition for the spin 2
conserved charges $P^\gamma$ and $Q^\gamma$ (the electric and magnetic
4-momentum of the spin 2 field):
\begin{equation}
\frac{4 G P_\gamma Q^\gamma}{\hbar} \in \ZZ.
\end{equation}

\section{Duality symmetric action without sources}

The strategy in this case is the same as the one used for
electromagnetism: write the Hamiltonian action, completely fix
the gauge to go to the reduced phase space and introduce new potentials. Doing this, the authors of
\cite{Henneaux:2005qf} managed to write a duality invariant action for
linearized gravity. In the following, we will present their results in
a different notation.

\subsection{Canonical formulation of Pauli-Fierz theory}
\label{sec:canon-form-pauli}

The Hamiltonian formulation of general relativity linearized around
flat spacetime is
\begin{eqnarray}
  S_{PF}[h_{mn},\pi^{mn},n_m,n]=\int dt\Big[\int
  d^3x\, \big(\pi^{mn}\dot h_{mn}-n^m\cH_m-n\cH \big)-H_{PF}\Big], 
\end{eqnarray}
with
\begin{multline}
H_{PF}[h_{mn},\pi^{mn}]=\int
d^3x\big(\pi^{mn}\pi_{mn}-\half\pi^2+\frac{1}{4}\d^rh^{mn}\d_rh_{mn}-\\
-\half\d_m h^{mn}\d^r h_{rn}+\half\d^m h\d^n h_{mn}-\frac{1}{4}\d^m
h\d_m h\big),\label{H}
\end{multline}
and
\begin{eqnarray}
\cH_m=-2\d^n\pi_{mn},\quad
\cH_\perp=\Delta h-\d^m\d^n h_{mn}.\label{eq:GRAVconstraintsPF}
\end{eqnarray}
Again, indices are lowered and raised with the flat space metric
$\delta_{mn}$ and its inverse, $h={h^m}_m$, $\pi={\pi^m}_m$ and
$\Delta =\d_m\d^m$ is the Laplacian in flat space. The linearized $4$
metric is reconstructed using $h_{00}=-2n$ and $h_{0i}= n_i$.

\subsection{Decomposition of symmetric rank two tensors}
\label{sec:decomp-symm-rank}

Symmetric rank two tensors $\phi_{mn}$ decompose as
\cite{deser:1967zr,R.-Arnowitt:1962uq}
\begin{eqnarray}
\phi_{mn} & = & \phi^{TT}_{mn} + \phi^T_{mn} + \phi^L_{mn},\\
\phi^L_{mn} & = & \partial_m \psi_n + \partial_n \psi_m,\\
\phi^T_{mn} & = & \half \left( \delta_{mn} \Delta-\partial_m \partial_n
\right) 
\psi^T.\label{eq:GRAVdecomp}
\end{eqnarray}
Here $\phi^{TT}_{mn}$ is the transverse-traceless part, containing two
independent components. The tensor $\phi^T_{mn}$ contains the trace
of the transverse part of $\phi_{mn}$ and only one independent
component. The last three components are the longitudinal part
contained in $\phi^L_{mn}$. In terms of the original tensor
$\phi_{mn}$ the potentials for the longitudinal part and the trace are
given by 
\begin{eqnarray}
\psi_{m} & = & \Delta^{-1} \left( \d^n \phi_{mn}
  - \frac{1}{2} \Delta^{-1} \d_m\d^k\d^l\phi_{kl} \right),\\
\psi^T & = & \Delta^{-1} \left( \phi - 
\Delta^{-1} \d^m\d^n\phi_{mn}\right),
\end{eqnarray}
while the transverse traceless part is then defined as the remainder, 
\begin{eqnarray}
\phi^{TT}_{mn} & = & \phi_{mn} - \phi^T_{mn} - \phi^L_{mn}.
\end{eqnarray}
This implies 
\begin{align}
  \label{eq:26}
\Delta^2
\phi^{TT}_{mn}&=\Delta^2\phi_{mn}-\Delta\d_m\d^k\phi_{kn}-\Delta\d_n\d^k\phi_{km}
\nonumber\\ & -\half\Delta(\delta_{mn}\Delta-\d_m\d_n)\phi +
\half(\delta_{mn}\Delta+\d_m\d_n)\d^k\d^l\phi_{kl},\\
\int d^3x\, \phi^{mn}\Delta^2 \phi^{TT}_{mn} & = \int d^3x\,
\Big(\Delta\phi^{mn}\Delta\phi_{mn}+2 \d_m\phi^{mn}\Delta \d^k\phi_{kn}
-\half (\Delta\phi)^2\nonumber\\ & +\d_m\d_n \phi^{mn}\Delta \phi+\half
\d_m\d_n\phi^{mn}\d_k\d_l\phi^{kl}\Big).
\end{align}

Alternatively, one can introduce the local operator $\mathcal{P}^{TT}$
\begin{equation}
  \left(\mathcal{P}^{TT} \phi \right)_{mn} =
  \half\big[\epsilon_{mpq} \partial^p (\Delta {\phi^{q}}_n
  -\partial_n \partial_r \phi^{qr}) + \epsilon_{npq} \partial^p (\Delta
  {\phi^{q}}_m  - \partial_m \partial_r 
  \phi^{qr})\big], \label{eq:GRAVdefPTT}
\end{equation}
which projects out the longitudinal and trace parts and onto a
transverse-traceless tensor,
\begin{eqnarray}
\left(\mathcal{P}^{TT} \phi \right)_{mn}  =  \mathcal{P}^{TT} \left(
  \phi^{TT} \right)_{mn}  =
\left( \mathcal{P}^{TT} \phi \right)^{TT}_{mn}. 
\end{eqnarray}
In addition, 
\begin{eqnarray}
  \label{eq:GRAVpropP}
  \left(\mathcal{P}^{TT} ( \mathcal{P}^{TT}\phi)
  \right)_{mn}=-\Delta^3 \phi^{TT}_{mn}.
\end{eqnarray}
As a consequence, the transverse-traceless tensor $\phi^{TT}_{mn}$ can
be written as $\mathcal{P}^{TT}$ acting on a suitable potential
$\psi^{TT}_{mn}$,
\begin{equation}
\phi^{TT}_{mn} =\left( \mathcal{P}^{TT} \psi^{TT}\right)_{mn},
\quad
\psi^{TT}_{mn} = - \Delta^{-3}\left(\mathcal{P}^{TT} \phi
\right)_{mn}.  
\end{equation}

When acting on a
transverse-traceless tensor, the last two terms of $\mathcal{P}^{TT}$ can be dropped. In
this case, it is related to the generalized curl 
\cite{Deser:2005ve,Deser:2005qf}, 
\begin{eqnarray}
  \left(\cO\phi\right)_{mn}=\half(\epsilon_{mpq} \partial^p {\phi^{q}}_n + 
\epsilon_{npq} \partial^p  {\phi^{q}}_m),\\
\left(\mathcal{P}^{TT} \phi^{TT} \right)_{mn}=\Delta
\left(\cO\phi^{TT}\right)_{mn}. 
\end{eqnarray}
Remark that the generalized curl acting on a transverse-traceless
tensor will produce a transverse-traceless tensor.
A second operator that projects out the longitudinal and trace
parts and onto a transverse-traceless tensor is $\mathcal{Q}^{TT}$,
\begin{eqnarray}
  \left(\mathcal{Q}^{TT} \phi \right)_{mn}  =  
\epsilon_{mpq}\epsilon_{nrs} \partial^p\partial^r \Delta\phi^{qs}
-\half(\delta_{mn}\Delta -\d_m\d_n)(\Delta\phi-\d^p\d^r\phi_{pr}). 
\end{eqnarray}
In this case, 
\begin{eqnarray}
  \label{eq:9}
  \left(\mathcal{Q}^{TT} ( \mathcal{Q}^{TT}\phi)
  \right)_{mn}=\Delta^4 \phi^{TT}_{mn},
\end{eqnarray}
so that the transverse-traceless tensor $\phi^{TT}_{mn}$ can
be written as $\mathcal{Q}^{TT}$ acting on another potential
$\chi^{TT}_{mn}$,
\begin{eqnarray}
  \phi^{TT}_{mn} =\left( \mathcal{Q}^{TT} \chi^{TT}\right)_{mn},
\quad
  \chi^{TT}_{mn} = \Delta^{-4}\left(\mathcal{Q}^{TT} \phi
\right)_{mn}. 
\end{eqnarray}
In turn this operator is related to the way the constraints $\cH_m=0$
are solved by expressing the momenta $\pi^{mn}$ in terms of
superpotentials in \cite{Henneaux:2005qf}.  When acting on a
transverse-traceless tensor, the last term can again be dropped and
it is related to the square of the generalized curl,
\begin{eqnarray}
  \left(\mathcal{Q}^{TT} \phi^{TT} \right)_{mn}=\Delta 
\left(\cO(\cO\phi^{TT})\right)_{mn}=-\Delta^2 \phi^{TT}_{mn}.
\end{eqnarray}

The elements of the decomposition are orthogonal under 
integration if boundary terms can be neglected,
\begin{equation}
  \int d^3 x\, \phi^{mn} \varphi_{mn} = \int d^3 x 
  \left( \phi^{TTmn} \varphi^{TT}_{mn} +\phi^{Lmn} \varphi^L_{mn} + 
\phi^{Tmn} \varphi^T_{mn} \right).
\end{equation} 
and the operators $\mathcal{P}^{TT},\mathcal{Q}^{TT},\cO $ are
self-adjoint, e.g., 
\begin{eqnarray}
  \int d^3 x\, \left( \mathcal{P}^{TT} \phi\right)^{mn}\varphi_{mn}=
\int d^3 x\, \phi^{mn}\left( \mathcal{P}^{TT} \varphi\right)_{mn}.
\end{eqnarray}

\subsection{Duality invariant action without sources}
\label{sec:pauli-fierz-terms}

 Because of
the orthogonality of the decomposition, the canonically conjugate
pairs can be directly read off from the kinetic term and are given by
\begin{eqnarray}
\big (h^{TT}_{mn}( x),\,\pi^{kl}_{TT}(\vec
  y)\big),\quad \big(h^L_{mn}(\vec
  x),\, \pi^{kl}_L( y)\big),\quad 
\big(h^{T}_{mn}( x),\,\pi_{T}^{kl}( y)\big).
\end{eqnarray}
The first class constraints $\cH_m=0=\cH$ are equivalent to
$\pi^{kl}_L=0=h^T_{mn}$. They can be gauge fixed through the
conditions $h^L_{mn}=0=\pi_{T}^{kl}$. The reduced theory only depends
on $2$ degrees of freedom (per spacetime point), the
transverse-traceless components $(h^{TT}_{mn}(\vec
x),\,\pi^{kl}_{TT}( y))$ and the reduced Hamiltonian simplifies
to 
\begin{eqnarray}
  H^R=\int d^3x\, \Big(
  \pi^{mn}_{TT}\pi_{mn}^{TT}+\frac{1}{4}\d_rh^{TT}_{mn}\d^rh_{TT}^{mn}\Big). 
\end{eqnarray}

Using the same strategy than in the electromagnetic case, the next
step is to introduce new potentials. 
In this case, the authors of \cite{Henneaux:2005qf} introduced two new potentials:
\begin{equation}
\pi^{mn}_{TT} = - \Delta H^{Dmn}_{TT} \quad \text{and} \quad
h^{TT}_{mn} = 2 \left(\cO H^{TT}\right)_{mn}.
\end{equation}
Plugging this into the reduced action and doing some
integrations by parts brings it to
\begin{eqnarray}
  S[H_{TT},H^D_{TT}]&=&\int dt\Big[\int
  d^3x\, \left[-2\Delta \left( \cO H^D_{TT}\right)^{mn}\dot
  H^{TT}_{mn}\right]-H\Big], \\
H&=&\int d^3x\, \Big(
  \Delta^2 {H^D_{TT}}^{mn} H^{DTT}_{mn}+\Delta^2 H_{TT}^{mn} H^{TT}_{mn}\Big).
\end{eqnarray}
As before, one can introduce a notation better suited to the duality
$H^a_{TT}=(H_{TT},H^D_{TT})$. Using this, the action takes a
duality invariant form:
\begin{eqnarray}
\label{eq:GRAVactdualinv}
  S[H^a_{TT}]&=&\int dt\Big[\int
  d^3x\, \epsilon_{ab}\Delta \left( \cO H^a_{TT}\right)^{mn}\dot
  H^{bTT}_{mn}-H\Big], \\
H&=&\int d^3x\, \Delta^2 H_{aTT}^{mn} H^{aTT}_{mn}.
\end{eqnarray}
The duality is the $SO(2)$ rotation generated by
\begin{equation}
\label{eq:GRAVdualpot}
\delta_D H^{TTa}_{mn} = \epsilon^{ab}H^{TT}_{bmn}.
\end{equation}
This transformation obviously leaves the action
(\ref{eq:GRAVactdualinv}) invariant. The associated duality generator
is
\begin{equation}
D = - \int d^3x \cP^{TT} (H^a)_{mn} H^{mn}_a.
\end{equation}
In \cite{Henneaux:2005qf}, it was cast in the form of a Chern-Simons term.

The fact that this transformation is indeed the electromagnetic
duality is less clear for the spin 2. In the electromagnetism case,
the two potentials where easily related to the electric and the
magnetic fields. It was then easy to interpret the $SO(2)$ rotation on
the $a$ index as the duality introduced at the level of the equations
of motion. Following \cite{Henneaux:2005qf}, we see that the electric
and magnetic part of the Weyl tensor are given in term of the new
potentials as:
\begin{eqnarray}
E_{mn} & =& 2\epsilon_{mpq} \d^p \Delta H^q_{\phantom{q}n},\\
B_{mn} & =& 2\epsilon_{mpq} \d^p \Delta {H^D}^q_{\phantom{q}n}.
\end{eqnarray}
We see that the duality defined in (\ref{eq:GRAVdualpot}) induce the
duality (\ref{eq:GRAVdualitypartriemann}) at the level of the Riemann tensor.

\subsection{Poincar\'e generators}

As for electromagnetism, the invariance of the action under the
Poincar\'e transformations is no longer manifest but is of course
still present. The strategy here is the same: introduce the new potentials in the reduced phase space
generators which are given by the generators of Pauli-Fierz evaluated on the
constraint surface.

To build the Poincar\'e generators of Pauli-Fierz, we will use the
fact that it is the linearization of general relativity. In this section, we assume that the canonical variables vanish
sufficiently fast at the boundary so that integrations by parts can be
used even if the gauge parameters do not vanish at the
boundary.

In the Hamiltonian formulation of general relativity
\cite{R.-Arnowitt:1962uq}, the canonically conjugate variables are the
spatial $3$ metric $g_{ij}$ and the extrinsic curvature
$\pi^{ij}$. The constraints are explicitly given by
\begin{equation}
  \mathcal{H}_\perp  =  
\frac{1}{\sqrt{g}}\left( \pi^{mn}\pi_{mn} - \frac{1}{2}\pi^2\right)-
\sqrt{g} R\,,\qquad 
  \mathcal{H}_i  =  -2 \nabla_j \pi^{j}_{\phantom{j}i}. 
\end{equation}
The associated generators of gauge transformation $H[\xi]=\int d^3x \,
\left( \mathcal{H}_\perp \xi^\perp + \mathcal{H}_i \xi^i\right)$
satisfy the so-called surface deformation algebra
\cite{Teitelboim:1973kx,Hojman:1976vn},
\begin{eqnarray}
\label{eq:GRAVHalgebra}
\left\{ H[\xi] , H[\eta]\right\} & = & H[[\xi,\eta]_{SD}],\\
{[\xi,\eta]}^\perp_{SD} & = & \xi^i \partial_i \eta^\perp 
- \eta^i \partial_i \xi^\perp,\\
{[\xi, \eta ]}^i_{SD} & = & g^{ij}\left( \xi^\perp \partial_j
  \eta^\perp - \eta^\perp \partial_j \xi^\perp\right) 
+ \xi^j \partial_j \eta^i - \eta^j \partial_j \xi^i.
\end{eqnarray}
When the parameters $f,g$ of gauge transformations depend on the
canonical variables, \eqref{eq:GRAVHalgebra} is replaced by
\cite{Brown:1986zr}
\begin{align}
  \label{eq:GRAVPoincarealgebra}
  \left\{ H[f] , H[g]\right\} & =  H[k],\\
k& =[f,g]_{SD}+\delta_g f-\delta_f g -m, \\
m^\perp & =\int d^3x^\prime\, \Big[\{f^\perp, g^\perp(x^\prime)\}
\cH_{\perp}(x^\prime) + \{f^\perp, g^j(x^\prime)\}
\cH_{j}(x^\prime)\Big],\\
m^i & =\int d^3x^\prime\, \Big[\{f^i, g^\perp(x^\prime)\}
\cH_{\perp}(x^\prime) + \{f^i, g^j(x^\prime)\}
\cH_{j}(x^\prime)\Big],
\end{align}
where 
\begin{align}
  \label{eq:GRAVgaugefull1}
  \delta_\xi g_{ij} & = \nabla_i \xi_j+\nabla_j\xi_i+2
  D_{ijkl}\pi^{kl}\xi^\perp,\\
D_{ijkl}&=\frac{1}{2\sqrt g}(g_{ik}g_{jl}+g_{jk}g_{il}-g_{ij}g_{kl}),\\
\delta_\xi \pi^{ij} & = -\xi^\perp\sqrt g(R^{ij}-\half
g^{ij}R)+\frac{\xi^\perp}{2\sqrt
  g}g^{ij}(\pi^{kl}\pi_{kl}-\half\pi^2)\nonumber\\ & -2\frac{\xi^\perp}{\sqrt g}
(\pi^{im}\pi_m^j -\half \pi^{ij}\pi)+\sqrt
g(\nabla^j\nabla^i\xi^\perp-g^{ij}\nabla_m \nabla^m \xi^\perp)\nonumber\\ &
+\nabla_m (\pi^{ij} \xi^m)-\nabla_m
\xi^i\pi^{mj}-\nabla_m\xi^j\pi^{mi}. \label{eq:GRAVgaugefull2}
\end{align}

Let $g_{ij}=\delta_{ij}+h_{ij}$ and consider the canonical change of
variables from $g_{ij},\pi^{kl}$ to $z^A=(h_{ij},\pi^{kl})$. We will
expand in terms of the homogeneity in the new variables and use the
flat metric $\delta_{ij}$ to raise and lower indices in the remainder
of this appendix. Furthermore, Greek indices take values from $0$ to
$3$ with $\mu=(\perp,i)$. Indices are lowered and raised with
$\eta_{\mu\nu}={\rm diag} (-1,1,1,1)$ and its inverse. Let
$\tilde\omega_{\mu\nu}=-\tilde\omega_{\nu\mu}$.

To lowest order, i.e., when $g_{ij}=\delta_{ij} $, the vector fields 
\begin{equation}
\label{eq:GRAVvect}
\xi_P(\tilde\omega,\tilde a)^\mu=
-\tilde\omega^\mu_{\phantom{\mu}i}x^i+\tilde a^\mu,
\end{equation}
with bracket the surface deformation bracket form a representation of
the Poincar\'e algebra \cite{Regge:1974kx},
\begin{equation}
\label{eq:GRAVvectalgebra}
[\xi_P(\tilde \omega_1,\tilde a_1), \xi_P(\tilde \omega_2,\tilde 
a_2)]^{(0)}_{SD}=
\xi_P([\tilde \omega_1,\tilde\omega_2], \tilde\omega_1 \tilde a_2 -
\tilde\omega_2 \tilde a_1).
\end{equation}

For the gauge generators, we find $H[\xi]=H^{(1)}[\xi] + H^{(2)}[\xi]
+H^{(3)}[\xi] + \cdots$, where 
\begin{align}
  H^{(1)}[\xi] & = \int d^3x \, \left( -2 \partial^j \pi_{ij}\xi^i +
    ( \partial^i\partial^j h_{ij} - \Delta
    h) \xi^\perp \right)\\
  & = \int d^3x \, \left( \cH_m^{(1)}\xi^m+\cH^{(1)}_\perp
    \xi^\perp\right)
\end{align}
are the gauge generators associated to the constraints \eqref{eq:GRAVconstraintsPF}
of the Pauli-Fierz theory. 
Because 
\begin{equation}
H[[\xi,\eta]_{SD}]=H^{(1)}[[\xi,\eta]_{SD}^{(0)}]+H^{(2)}[[\xi,\eta]_{SD}^{(0)}]
+H^{(1)}[[\xi,\eta]_{SD}^{(1)}]+O(z^3),\label{eq:GRAVpoincareexprhs}
\end{equation}
we have to lowest non trivial order
\begin{equation}
  \label{eq:GRAVlowestnontriv}
 \left\{H^{(1)}[\xi],H^{(2)}[\eta]\right\}=H^{(1)}[[\xi,\eta]_{SD}^{(0)}].
\end{equation}
This means that $H^{(2)}[\eta]$ are observables, i.e., weakly gauge
invariant functionals.

One can use integrations by parts to show that
$H^{(1)}[\xi_P]=0$. It then follows that 
\begin{equation}
\left\{ H[\xi_P] , H[\eta_P]\right\} = \left\{ H^{(2)} [\xi_P], 
H^{(2)}[\eta_P]\right\} + O(z^3).
\end{equation}
For vectors $\xi_P(\tilde\omega,\tilde a),\eta_P(\tilde \theta,\tilde
b)$ of the form \eqref{eq:GRAVvect}, the first term on the RHS of
\eqref{eq:GRAVpoincareexprhs} vanishes on account of \eqref{eq:GRAVvectalgebra}. To
lowest non trivial order, \eqref{eq:GRAVHalgebra} then implies
\begin{equation}
  \label{eq:GRAVHalgexp}
  \left\{ H^{(2)} [\xi_P], 
H^{(2)}[\eta_P]\right\}=H^{(2)}[[\xi_P,\eta_P]_{SD}^{(0)}]
+H^{(1)}[[\xi_P,\eta_P]_{SD}^{(1)}]. 
\end{equation}
The generators $H^{(2)} [\xi_P]$ equipped with the Poisson bracket
thus form a representation of the Poincar\'e algebra when the
constraints of the Pauli-Fierz theory are satisfied. Explicitly, the
term proportional to the constraints is
\begin{equation}
  H^{(1)}[[\xi,\eta]_{SD}^{(1)}]=-2\int d^3x\, \d^j\pi_{ji} h^{ik}(
  \xi^\perp_P\theta^\perp_{\phantom{\perp}k} -\eta^\perp_P
\omega^\perp_{\phantom{\perp}k}),
\end{equation}
while
\begin{align}
\mathcal{H}^{(2)}_i & = - 
2 \partial_j \left(\pi^{jk} h_{ik} \right)+ \pi^{jk}\partial_i
h_{jk}\\
\mathcal{H}^{(2)}_\perp & =  \pi^{ij} \pi_{ij} - 
\frac{1}{2} \pi^2\nonumber\\
&+\frac{1}{4} \partial_k h_{ij} \partial^k h^{ij} -
\frac{1}{2} \partial_k h^{ki} \partial^j h_{ij} +
\frac{1}{2} \partial_i h \partial_j h^{ij} 
- \frac{1}{4} \partial_i h \partial^i h\nonumber\\
&+\partial_l\left( \frac{1}{2} h \partial^l h 
- h^{ij}\partial^l h_{ij} - \frac{1}{2} h \partial_i h^{il} 
- h^{il}\partial_i h + \frac{3}{2} h^{lj}\partial^i h_{ij}
+\frac{1}{2} h_{ij} \partial^i h^{jl}\right).
\end{align}

Isolating terms proportional to the constraints, we find
\begin{align}
  H^{(2)}[\xi]&=\int d^3x\, \left(\cH_m h^{mi}\xi_i+\frac{1}{2} \cH h
    \xi^\perp\right)  + \bar H^{(2)}[\xi],\\
\bar \cH^{(2)}_i&=-\pi^{jk}(\d_jh_{ki}+\d_k h_{ji}-\d_ih_{jk}),\\
\bar \cH^{(2)}_\perp&=\pi^{ij} \pi_{ij} - 
\frac{1}{2} \pi^2\nonumber\\
&+\frac{1}{4} \partial_k h_{ij} \partial^k h^{ij} -
\frac{1}{2} \partial_k h^{ki} \partial^j h_{ij} +  \frac{1}{4}\partial_i
h \partial^i h 
\nonumber\\
&+\partial_l\left(
- h^{ij}\partial^l h_{ij} 
- h^{il}\partial_i h + \frac{3}{2} h^{lj}\partial^i h_{ij}
+\frac{1}{2} h_{ij} \partial^i h^{jl}\right),
\end{align}
with $\bar H^{(2)}[\xi]=\int d^3x\,\left( \bar \cH^{(2)}_i\xi^i+\bar
  \cH^{(2)}_\perp \xi^\perp\right)$.  On account of \eqref{eq:GRAVlowestnontriv} and
the analog of \eqref{eq:GRAVPoincarealgebra} for $H^{(1)}[f]$, it follows that
\begin{equation}
  \left\{ \bar H^{(2)} [\xi_P], 
\bar H^{(2)}[\eta_P]\right\}\approx \bar H^{(2)}[[\xi_P,\eta_P]_{SD}^{(0)}],
\end{equation}
where $\approx$ means an equality up to terms proportional to the
constraints $\cH_m,\cH_\perp$ of Pauli-Fierz theory. Note that the
functionals $H^{(2)}[\xi_P]$ and $\bar H^{(2)}[\xi_P]$ generate
transformations of the canonical variables that are equivalent because
they differ at most by a gauge transformations of the Pauli-Fierz
theory when restricted to the constraint surface.

The generators for global Poincar\'e transformations of Pauli-Fierz
theory can then be identified as 
\begin{eqnarray}
  \label{eq:GRAVidentpoincare}
&&  Q_G(\omega,a)=\half\omega_{\mu\nu}J^{\mu\nu}_G-a_\mu P^\mu_G
=\bar H^{(2)}[\xi_P(\tilde \omega,\tilde
  a)]\nonumber\\
&& \tilde\omega_{\mu\nu}=\omega_{\mu\nu},\quad \tilde a_\perp=
a_\perp, \quad \tilde a_i= a_i+\omega_{\perp i}x^0.
\end{eqnarray}
Indeed, differentiating \eqref{eq:GRAVHalgexp} with respect to $b_\perp$ gives
\begin{eqnarray}
  \label{eq:103}
  \{H,Q_G(\omega,a)\}=\ddl{}{t}Q_G(\omega,a)+2\int d^3x\, \d^j\pi_{ji} h^{ik}
\omega_{\perp k}.
\end{eqnarray}
When combined with \eqref{eq:GRAVHalgexp} and \eqref{eq:GRAVidentpoincare}, this shows that,
on the constraint surface, the generators $Q_G(\omega,a)$ are
conserved and satisfy the Poincar\'e algebra.

Finally, we can further simplify the explicit expression for $\bar
H^{(2)}[\xi_P]$ by using linearity of $\xi_P$ in $x^i$ and
integrations by parts to show that
\begin{align}
\int d^3x\,\bar
 \cH^{(2)}_\perp \xi^\perp_P=\int d^3x\, \Big[&\pi^{ij} \pi_{ij} - 
\frac{1}{2} \pi^2+\frac{1}{4} \partial_k h_{ij} \partial^k h^{ij} -
\frac{1}{2} \partial_k h^{ki} \partial^j h_{ij}   \nonumber\\
&+\frac{1}{4}\partial_i
h \partial^i h +\partial_l\big(
h \partial_i h^{il} +  h^{lj}\partial^i h_{ij}\big)\Big]\xi^\perp_P.
\end{align}

The expansion of the gauge transformations \eqref{eq:GRAVgaugefull1},
\eqref{eq:GRAVgaugefull2} gives to first order:
\begin{gather}
  \delta^{(0)}_\xi h_{ij}= \d_i\xi_j+\d_j\xi_j,\qquad \delta^{(0)}_\xi
  \pi^{ij}=(\d^i\d^j-\delta^{ij} \Delta)\xi^\perp,\\
\delta^{(1)}_\xi h_{ij}=\xi^k\d_k h_{ij}+\d_i \xi^k h_{kj}+\d_j\xi^k
h_{ik} +2 \pi_{ij}\xi^\perp-\delta_{ij}\pi\xi^\perp,\label{eq:49a}\\
\delta^{(1)}_\xi \pi^{ij}=\half h(\d^i\d^j
-\delta^{ij}\Delta)\xi^\perp
-h^{im}\d_m\d^j\xi^\perp-h^{jm}\d_m\d^i\xi^\perp+h^{ij}
\Delta\xi^\perp\nonumber\\+\delta^{ij}h^{mn}
\d_m\d_n\xi^\perp+\d_m(\pi^{ij}\xi^m)-\pi^{mj}\d_m\xi^i
-\pi^{mi}\d_m\xi^j
\nonumber\\ +\half \d_k\xi^\perp\Big[
-\d^j h^{ki}-\d^i h^{kj} +\d^k h^{ij}+\delta^{ij}(2\d_l h^{kl}-\d^k
h)\Big]
\nonumber\\+\half \xi^\perp\Big[\d^i\d^j h
+\Delta h^{ij} -\d_k \d^i h^{jk} -\d_k \d^j
h^{ik} -\delta^{ij} (\Delta h-\d_k\d_l h^{kl})\Big]. \label{eq:49b}
\end{gather}

\vspace{5mm}

From these generators, one can deduce the generators of the reduced
phase space theory by evaluating them on the constraints
surface. This gives
\begin{eqnarray}
 Q_G(\omega,a) & = & \half\omega_{\mu\nu}J^{\mu\nu}_G-a_\mu P^\mu_G = \int d^3x\,\left( \bar \cH^{(2)}_i\xi^i+\bar
  \cH^{(2)}_\perp \xi^\perp\right)\\
\bar \cH^{(2)}_i&=&-\pi^{TTjk}(\d_jh^{TT}_{ki}+\d_k h^{TT}_{ji}-\d_ih^{TT}_{jk}),\\
\bar \cH^{(2)}_\perp&=&\pi^{TTij} \pi^{TT}_{ij} +\frac{1}{4} \partial_k h^{TT}_{ij} \partial^k h^{TTij}.
\end{eqnarray}
Finally, the Poincar\'e generators for the double potential formalism
are obtained by introducing the new potentials in the above expression for $Q_G$.

\section{Duality symmetric action with external sources}

Having briefly recalled the double formalism potential for gravity, we
will now introduce our work on the extension of this formalism to
include external sources, as was done in the spin 1 case. The main
difference is that the spin 2 action was built piece by piece due to
its complexity.

The analysis starts with a
degree of freedom count that shows that the phase space of duality
invariant spin $2$ fields with doubled gauge invariance can be taken
to consist of $2$ symmetric tensors, $2$ vectors and $2$ scalars in
$3$ dimensions. We then define the metric, extrinsic curvature and
their duals in terms of the phase space variables and propose our
duality invariant action principle with enhanced gauge invariance. We
proceed by identifying the canonically conjugate pairs and discuss the
gauge structure, Hamiltonian and duality generators of the theory.  In
the absence of sources, we then show how the generators for global
Poincar\'e transformations can be extended to the
duality invariant theory.

The next step is the introduction of the external sources. The equations of motion are
then solved in the simplest case of a point-particle dyon sitting at
the origin. They are Coulomb-like without string singularities. By
identifying the Riemann tensor in terms of the canonical variables and
computing it for this case, we then show that this solution indeed
describes the linearized Taub-NUT solution. 

Finally, we discuss the surface charges
of the theory and show that they include electric and magnetic
energy-momentum and angular momentum. Because of the non-locality of
the Poisson structure, we proceed indirectly and show that the
expressions obtained by generalizing the surface charges of
Pauli-Fierz theory in a duality invariant way fulfill the standard
properties.  Finally, we investigate how the surface charges transform
under a global Poincar\'e transformation of the sources.

\subsection{Degree of freedom count}
\label{sec:degree-freedom-count}

In order to be able to couple to sources of both electric and magnetic
type in a duality invariant way, we want to keep all components and
double the gauge invariance of the theory. With 2 degrees of freedom,
$\#\, {\rm dof}=2$, and 8 first class constraints, $\#\, {\rm fcc}=8$,
we thus need $10$ canonical pairs, $\#\, {\rm cp}=10$, according to
the degree of freedom count \cite{Henneaux:1990bh}
\begin{eqnarray}
  2*( \#\, {\rm cp})=2*(\#\, {\rm dof})+2*(\#\, {\rm fcc}). 
\end{eqnarray}
This can be done by taking 2 symmetric tensors, 2 vectors and 2
scalars as fundamental canonical variables, 
\begin{eqnarray}
  \label{eq:18}
  z^A=(H^a_{mn},A^a_m,C^a).
\end{eqnarray} 

\subsection{Change of variables and duality rotations}
\label{sec:change-vari-dula}

For $a=1,2$, consider $h^{a}_{mn}=(h_{mn},h^D_{mn})$ and
$\pi^{mn}_a=(\pi^{mn}_D,\pi^{mn})$ and the definitions
\begin{align}
  \label{eq:16}
  h^a_{mn} &=  \epsilon_{mpq}\d^p {H^{aq}}_{n}+
 \epsilon_{npq}\d^p
 {H^{aq}}_{m}+\d_m A^a_n +\d_n A^a_m +\half
(\delta_{mn}\Delta-\d_m\d_n)C^a\nonumber\\
&=2\Delta^{-1}\left(\cP^{TT} H^a\right)_{mn}+ 
\partial_m \big(\Delta^{-1}\epsilon_{npq} \partial^p  \partial_r
H^{aqr}+A^a_n\big) \nonumber\\ &\hspace*{2cm} + \partial_n\big(
\Delta^{-1} \epsilon_{mpq}    \partial^p \partial_r H^{aqr} +A^a_m\big)
+ \half (\delta_{mn}\Delta-\d_m\d_n)C^a,
\\
\pi^{a}_{mn} &=
\epsilon_{mpq}\epsilon_{nrs} \partial^p\partial^r 
H^{aqs} -\d_m \d^r H^a_{rn} - \d_n \d^r H^a_{rm} -(\delta_{mn}
\Delta-\d_m\d_n) H^a +\delta_{mn} \d^k\d^l H_{kl}^a\nonumber\\ &=
\Delta^{-1}\left(\cQ^{TT} H^a\right)_{mn}
-\d_m \d^r H^a_{rn} - \d_n \d^r H^a_{rm}  
-\half(\delta_{mn}\Delta -\d_m\d_n) H^a\nonumber\\ &\hspace*{2cm}
+\half \Delta^{-1}(\delta_{mn}\Delta +
\d_m\d_n )\d^p\d^rH^a_{pr}
\nonumber\\
 &= -\Delta H^a_{mn}.\label{eq:16b}
\end{align}
The relations for $h_{mn}[H^1,A^1,C^1]$ and $\pi^{mn}[H^2]$ are the
local change of coordinates from the standard canonical variables of
linearized gravity to the new variables. They are invertible and, as
usual, the inverse is not local. The relations for
$h^{2}_{mn}=h^D_{mn}$, $\pi^{mn}_1=\pi^{mn}_D$ serve to denote
convenient combinations of the new variables in terms of which
expressions below will simplify. As indicated by the notation, the
infinitesimal duality rotations among the fundamental variables are
\begin{equation}
  \label{eq:20}
  \delta_D H^a_{mn}=\epsilon^{ab} H_{bmn},\ 
\delta_D A^a_{m}=\epsilon^{ab} A_{bm},\ \delta_D C^a=\epsilon^{ab}
C_{b}.
\end{equation}
Since $h^a_{mn},\pi^{mn}_a$ are
linear combinations of the fundamental variables, they are rotated in
exactly the same way. We can thus consider $h^2_{mn}=h^D_{mn}$,
$\pi^{mn}_1=\pi^{mn}_D$ as the dual spatial metric and the dual
extrinsic curvature in the linearized theory.

\subsection{Action principle and locality}
\label{sec:action-principle}

The duality invariant local action principe that we propose is of the
form
\begin{eqnarray}
  \label{eq:17}
  S_G[z^A,u^\alpha]=\int d^4x\, (a_A[z]\dot z^A  - u^\alpha
  \gamma_\alpha[z])  -\int dt\, H[z],  
\end{eqnarray}
where $u^\alpha$ denote the 8 Lagrange multiplies and $\gamma_\alpha$
the constraints. 

Let us stress here that we use the assumption that the flat space
Laplacian $\Delta$ is invertible in order to show equivalence with the
usual Hamiltonian or covariant formulation of Pauli-Fierz theory and
also to disentangle the canonical structure. The action principle
\eqref{eq:17} itself and the associated equations of motion will be
local both in space and in time independently of this assumption. The
theory itself is not local as a Hamiltonian gauge theory (see
e.g.~\cite{Henneaux:1992fk}, chapter 12) because the Poisson brackets
among the canonical variables will not be local.

\subsection{Canonical structure}
\label{sec:canonical-structure-1}

The explicit expression for the kinetic term is
\begin{multline}
  a_A\dot z^A =
 \epsilon_{ab} H^{amn}\Big(
  \big(\cP^{TT}\dot{H}^b\big)_{mn}+  \d_m\Delta \dot A^b_n
+\d_n\Delta \dot A^b_m+\\+
\half  (\delta^{mn} \Delta-\partial^m \partial^n) \Delta \dot
C^b\Big).
\label{eq:19a}
\end{multline}
The canonically conjugate pairs are identified by writing the
integrated kinetic term as
\begin{multline}
  \int d^4x\, a_A\dot z^A= \int d^4 x \, \Big( -2\Delta\left( \cO
    H^{2}_{TT}\right)^{mn} \dot
  H^{1TT}_{mn}+2\Delta \d_m H^{2mn}_L\dot A^1_n
\\-2\Delta \d_m \dot H^{1 mn}_L A^2_n
 -\half\Delta(\Delta H^{2}_T-\d_{p}\d_q H^{2pq}_{T})\dot
  C^1+\half\Delta (\Delta H^{1}_T-\d_{p}\d_q H^{1pq}_{T})\dot C^2 \Big).
 \end{multline}
This means that the usual canonical pairs of linearized gravity
can be chosen in terms of the new variables as
\begin{eqnarray}
  \Big( H^{1 TT}_{mn}( x),\,-2
\Delta\left(\cO H^{2}_{TT}\right)^{kl}(
  y)\Big),\,
\Big(C^1( x),\,-\half\Delta(\Delta
 H^{2}_T-\d_{p}\d_q H^{2pq}_{T})
  \,  (y)\Big)\,,
\nonumber \\
\Big(A^1_m ( x),\,
2 \Delta \d_r H^{2rn}_L ( y)\Big),
\end{eqnarray}
The $4$ additional canonical pairs are 
\begin{equation}
 \Big(A^{2}_m(x),\,
 -2 \Delta \d_r H^{1rn}_L  ( y)\Big),
 \Big(  C^2( x),
  \,\half \Delta(\Delta H^1_T-\d_{p}\d_q H^{1pq}_{T}) ( y)\Big).
\end{equation}
In particular, it follows that
\begin{equation}
\boxed{
  \{h^a_{mn}(x),\pi^{bkl}(y)\}=\epsilon^{ab}\half(\delta^k_m\delta^l_n+
\delta^k_m\delta^l_n)\delta^3(x,y).
}
\end{equation}

\subsection{Gauge structure}
\label{sec:gauge-structure}

The constraints $\gamma_\alpha\equiv( \cH_{am},\cH_{a\perp})$ are chosen as
\begin{empheq}[box=\fbox]{align}
  \cH_{am}&=-2\epsilon_{ab}\d^n
  \pi^b_{mn}=2\epsilon_{ab}\Delta\d^nH^b_{mn},\\
\cH_{a\perp}&=\Delta h_a-\d_m\d_n h_a^{mn}= \Delta^2 C_a. \label{eq:24a}
\end{empheq}
The constraints $\cH_{1m},\cH_{1\perp}$ are those of the standard
Hamiltonian formulation of Pauli-Fierz theory expressed in terms of
the new variables. The constraints are first class and abelian
\begin{eqnarray}
  \label{eq:12c}
  \{\gamma_\alpha,\gamma_\beta\}=0.
\end{eqnarray}

The new constraints $\gamma_\Delta^N=0$ are
$\cH_{2m}=0=\cH_{2\perp}$. They are equivalent to $\d^rH^1_{rm}=0=C^2$
and are gauge fixed through the conditions $A^2_m=0=H^{1T}_{mn}$.
This does not affect $\pi^{2kl}$, while $h^1_{mn}$ is changed by a
gauge transformation. The partially gauge fixed theory corresponds to
the usual Pauli-Fierz theory in Hamiltonian form as described in
section \bref{sec:canon-form-pauli}. 

In the same way, the original constraints $\cH_{1m}=0=\cH_{1\perp}$
are equivalent to $\d^rH^2_{rm}=0=C^{1}$ and are gauge fixed through
$A^2_m=0=H^{2T}_{mn}$, leading to the completely reduced theory in
terms of the $2$ transverse-traceless physical degrees of freedom.

If $\varepsilon^\alpha=(\xi^{am},\xi^{a\perp})$
collectively denote the gauge parameters, the gauge symmetries are
canonically generated by the smeared constraints,
\begin{eqnarray}
  \label{eq:27}
 \delta_\varepsilon z^A=\{z^A,\Gamma[\varepsilon]\},\quad
  \Gamma[\varepsilon]=\int d^3x\, \gamma_\alpha\epsilon^\alpha,
\end{eqnarray}
so that 
\begin{equation}
  \label{eq:GRAVgaugetransf}
\delta_\varepsilon
H^a_{mn}=-\Delta^{-1}\epsilon^{ab}(\delta_{mn}\Delta-\d_m\d_n)\xi^\perp_b,\quad 
\delta_\varepsilon A^a_m= \xi^a_m,\quad \delta_\xi C^a=0,
\end{equation}
which implies in particular 
\begin{equation}
\delta_\varepsilon h^{a}_{mn}=\partial_m\xi^a_n+\partial_n\xi^a_m,\qquad 
  \delta_\varepsilon
  \pi^{a}_{mn}=\epsilon^{ab}(\delta^{mn}\Delta-\d^m\d^n)\xi^\perp_b. 
\end{equation}

Note that a way to get local gauge transformations for the fundamental
variables is to multiply the constraints by $\Delta$, which is allowed
when the flat space Laplacian is invertible. This amounts to
introducing suitable potentials for the gauge parameters and Lagrange
multipliers.

\subsection{Duality generator}
\label{sec:duality-generator}

The canonical generator for the infinitesimal duality rotations
\eqref{eq:20} is 
\begin{align}
  \label{eq:1}
  D&=\int d^3x\, \Big( -(\cP^{TT} H^a)_{mn} 
H_a^{mn}+2\Delta \d_r H^{rm}_{a} A^a_{m}\nonumber\\ & \hspace{2cm} -
\half \Delta(\Delta H^a-\d^m\d^n H_{mn}^a) C_a\Big).
\end{align}
On the constraint surface, it reduces to
\begin{empheq}[box=\fbox]{align}
D&\approx -\int d^3x\,\cP^{TT}(H^a)_{mn} 
H_a^{mn},
\end{empheq}
which is also the duality generator of the
non-extended double potential formalism (see section \ref{sec:pauli-fierz-terms}).

This generator is only weakly gauge invariant,
\begin{equation}
  \{\cH_{am},D\}=\epsilon_{ab} \cH^b_m\qquad \{\cH_{a\perp},D\}=
\epsilon_{ab} \cH^b_\perp.
\end{equation}

\subsection{Hamiltonian}
\label{sec:hamiltonian}

In terms of the new variables \eqref{eq:16}-\eqref{eq:16b}, the
Pauli-Fierz Hamiltonian reads
\begin{multline}
H_{PF}=\int d^3x\, \Big(H^{amn}\Delta^2
H_{amn}^{TT}-2\Delta\d^r H^2_{rn}\d_s H^{2sn}-\\-\d^r\d^s H_{rs}^2\Delta
H^2-\half (\d^r\d^s H^2_{rs})^2+\frac{1}{8} \Delta C^1 \Delta^2 C^1\Big),
\end{multline}
where one can use \eqref{eq:26} to expand the first term as a local
functional of $H^a_{mn}$. 

The local Hamiltonian $H=\int d^3x\, h$ of the manifestly duality
invariant action principle \eqref{eq:17} is
\begin{align}
  H &=\int d^3x\, \Big(H^{amn}\Delta^2 
H_{amn}^{TT}-2\Delta\d^r H^a_{rn}\d_s H^{sn}_a-\nonumber\\ 
& \hspace{2cm}-\d^r\d^s H_{rs}^a\Delta
H_a-\half \d^r\d^s H^a_{rs}\d_k\d_l H^{kl}_a 
+\frac{1}{8}\Delta C^a \Delta^2 C_a\Big),
\end{align}
which simplifies to
\begin{empheq}[box=\fbox]{align}
H&=\int d^3x\, \Big(\Delta H^{a}_{mn}\Delta H^{mn}_a-\half \Delta
H^a\Delta H_a +\frac{1}{8}\Delta C^a \Delta^2 C_a\Big).
\end{empheq}
It is equivalent to the Pauli-Fierz Hamiltonian since it reduces to
the latter when the additional constraints $\d^r H^1_{rm}=0=C^2$
hold. Note that the terms proportional to $\d^r H^a_{rm}$ and $C^a$
may be dropped since they vanish on the constraint surface, $H \approx
\int d^3x\, H^{amn}\Delta^2 H_{amn}^{TT}$.

The Hamiltonian is gauge invariant on the constraint surface,
\begin{equation}
  \label{eq:GRAVep19}
  \{H,\Gamma[\xi]\}=\int d^3x\, \cH^a_m \d^m \xi^\perp_a.
\end{equation}

In order for the action \eqref{eq:17} to be gauge invariant,
it follows from \eqref{eq:GRAVep19} that the Lagrange multipliers
$u^\alpha$ need to transform as
\begin{equation}
  \label{eq:21}
  \delta_\xi u^{am}=\dot \xi^{am}-\d^m
   \xi^{a\perp}, \qquad \delta_\xi
  u^{a\perp}=\dot \xi^{a\perp}. 
\end{equation}

\subsection{Poincar\'e generators}
\label{sec:poincare-generators-1}

The same argumentation we used in section (\ref{sec:ELECpoincaretrans}) to
show that the symmetry generators of the usual Hamiltonian action of
electromagnetism are also symmetry generators of the extended double
potential formalism is still valid. The symmetry generators of
Pauli-Fierz are symmetry generators of the extended theory. It follows that the Poincar\'e generators $Q_G(\omega,a)$ of Pauli-Fierz
theory as described in section, when
expressed in terms of the new variables, are representatives for
the Poincar\'e generators of the extended theory.

As before, the generators obtained that way are not invariant under
the duality. We now want to show that one
can find representatives for the Poincar\'e generators that are
duality invariant, 
\begin{equation}
\{Q^D_G(\omega,a),D\}=0\label{eq:60},
\end{equation}
by adding terms proportional to the new constraints.

The first step in the proof consists in showing that the reduced phase
space generators, i.e., the generators $Q_G(\omega,a)$ for which all
variables except for the physical $H^a_{TT}$ have been set to zero,
are duality invariant. All other contributions to $Q_G(\omega,a)$ are
then shown to be proportional to the constraints of Pauli-Fierz
theory. Both these steps follow from straightforward but slightly
tedious computations. For the generators of rotations and boosts for
instance the computation is more involved because the explicit $x^i$
dependence has to be taken into account when performing integrations
by parts.

In terms of the new variables, the terms proportional to the
constraints are bilinear in $(h^1,A^2)$, $(\pi^2,A^2)$, $(h^1,C^1)$
and $(\pi^2,C^1)$. The duality invariant generators $Q^D_G(\omega,a)$
are then obtained by adding the same terms with the substitution
$h^1\to h^2$, $A^2\to -A^1$, $\pi^2\to -\pi^1$ and $C^1\to C^2$, while
keeping unchanged the terms involving only the physical variables
$H^a_{TT}$.

As a consequence, the duality invariant Poincar\'e transformations of
$h^1,\pi^2$ are unchanged on the extended constraint surface. They are
given by \eqref{eq:49a}-\eqref{eq:49b} where $\xi^\perp=-{\omega^0}_\nu
x^\nu+a^0$ and $\xi^i=-{\omega^i}_\nu x^\nu +a^i$. Because of
\eqref{eq:60}, those for of $h^2,-\pi^1$ are obtained, on the
contraint surface, by applying a duality rotation to the right
hand-sides of \eqref{eq:49a}-\eqref{eq:49b}.

\subsection{Interacting variational principle}
\label{sec:coupl-cons-electr}

We define
\begin{equation}
  h_{0m}^a=n^a_m=h_{m0}^a,\qquad h_{00}^a=-2n^a,
\end{equation}
and consider the action  
\begin{empheq}[box=\fbox]{align}
S_T[z^A,u^\alpha;T^{a\mu\nu}]=\frac{1}{16\pi
  G}S_G+S^J\label{eq:GRAVinteractaction},
\end{empheq}
with $S_G$ given in \eqref{eq:17} and the gauge invariant
interaction term
\begin{equation}
  \label{eq:35}
  S^J=\int d^4x\, \half h^a_{\mu\nu} T^{\mu\nu}_a,\qquad
  \d_\mu T^{\mu\nu}_a=0,
\end{equation}
where $T^{\mu\nu}_a\equiv (T^{\mu\nu},\Theta^{\mu\nu})$ are external,
conserved electric and magnetic energy-momentum tensors.

\subsection{Equations of motion}
\label{sec:GRAVEOMextdualpot}

Our goal in this section is to show that our interacting variational
principle (\ref{eq:GRAVinteractaction}) generates the duality
invariant equations of motion
(\ref{eq:GRAVEOMgene1})-(\ref{eq:GRAVEOMgene2}) introduced in section
\ref{sec:EMspin2}. To do so, we need to give the expression of the
full Riemann tensor and its dual in terms of our fields $(H^a_{mn},
A^a_m,C^a)$. We will start by introducing a new decomposition of the
Riemann tensor and its dual in term of electric and magnetic 
part. This decomposition is a generalization of the decomposition of a
Weyl tensor in a electric and magnetic part. After that, we will derive
the expression of this new parametrisation in term of our canonical variables.

\footnotesize
We will use a duality invariant notation: $R^a_{\mu\nu\rho\sigma} =
(R_{\mu\nu\rho\sigma},S_{\mu\nu\rho\sigma})$. The Ricci tensors and
Einstein tensors are defined as
\begin{equation}
R^a_{\mu\nu}={R^{a\alpha}}_{\mu\alpha\nu},\quad
G^a_{\mu\nu}={\nG^{a\alpha}}_{\mu\alpha\nu}=R^a_{\nu\mu}-\half
\eta_{\mu\nu}R^a.
\end{equation}
A general Riemann tensor $ R_{\mu\nu\rho\sigma}  =  - R_{\nu\mu\rho\sigma} =
  -R_{\mu\nu\sigma\rho}$ has 36 independent components.

The equations of motion
(\ref{eq:GRAVEOMgene1}) and (\ref{eq:GRAVEOMgene2}) becomes
\begin{equation}
  \label{eq:GRAVgeneEOM}
  G^{\mu\nu}_a=8\pi G\, T^{\mu\nu}_a\iff R^a_{\mu\nu\rho\sigma} + 
R^a_{\mu\sigma\nu\rho} +R^a_{\mu\rho\sigma\nu} = 8\pi G\, \epsilon^{ab}
  \lc_{\delta\nu\rho\sigma}\, 
  \overline{T}^\delta_{b\,\mu}. 
\end{equation}
They imply in particular that, on-shell, the tensors $R^a_{\mu\nu}, G^a_{\mu\nu}$ are symmetric \cite{Bunster:2006fk}.
Furthermore, the Bianchi ``identities'' (\ref{eq:GRAVEOMgene3}) and (\ref{eq:GRAVEOMgene4}) read
\begin{multline}
\partial_\epsilon R^a_{\gamma\delta\alpha\beta} +\partial_\beta 
R^a_{\gamma\delta\epsilon\alpha} +\partial_\alpha 
R^a_{\gamma\delta\beta\epsilon} =  8\pi G\, \epsilon^{ab}
\lc_{\epsilon\alpha\beta\rho} 
\left( \partial_\gamma \overline{T}^\rho_{b\,\delta} - 
\partial_\delta \overline{T}^{\rho}_{b\,\gamma}\right)\\ \iff
  \partial_\mu R_a^{\gamma\delta\rho\mu}  =  8\pi G\, \left(\partial^\delta
\overline{T}_a^{\rho\gamma}-
\partial^\gamma\overline{T}_a^{\rho\delta} \right),\label{eq:GRAVgenebianch}
\end{multline}
while the contracted Bianchi identities are
\begin{equation}
  \d_\nu G_a^{\mu\nu}=0. 
\end{equation}

Let 
\begin{multline}
  \label{eq:GRAVdefK}
  {K^{\lambda\tau}}_{\mu\nu\rho\sigma}
  [R^a_{\lambda\tau}]=\frac{1}{2} \left[ \eta_{\mu\rho} R^a_{\nu\sigma} +
    \eta_{\nu\sigma} R^a_{\mu\rho} - \eta_{\mu \sigma} R^a_{\nu\rho} -
    \eta_{\nu\rho} R^a_{\mu\sigma}\right] -\\- \frac{R^a}{6} \left[
    \eta_{\mu\rho} \eta_{\nu\sigma} -  \eta_{\mu \sigma}
    \eta_{\nu\rho}\right]. 
\end{multline}
Defining
\begin{equation}
  \label{eq:84}
  \tilde R^a_{\mu\nu\rho\sigma}=R^a_{\mu\nu\rho\sigma}-\half
\epsilon^{ab}\lc_{\rho\sigma\alpha\beta}
{{K^{\lambda\tau}}_{\mu\nu}}^{\alpha\beta}[R_{b\lambda\tau}],
\end{equation}
the tensor $\tilde R^a_{\mu\nu\rho\sigma}$ is skew in the first and
last pairs of indices, satisfies the cyclic identity because
$\lc^{\gamma\nu\rho\sigma}
R^a_{\mu\nu\rho\sigma}=\lc^{\gamma\nu\rho\sigma}\half
\epsilon^{ab}\lc_{\rho\sigma\alpha\beta}
{{K^{\lambda\tau}}_{\mu\nu}}^{\alpha\beta}[R_{b\lambda\tau}]$ and, as
a consequence, is also symmetric in the exchange of the first and last
pair of indices, $\tilde R^a_{\mu\nu\rho\sigma}= \tilde
R^a_{\rho\sigma\mu\nu}$. The associated Ricci tensors $\tilde
R^a_{\nu\sigma}=
R^a_{\nu\sigma}-\half\epsilon^{ab}\lc_{\nu\sigma\mu\alpha}R^{\mu\alpha}_b$
is then symmetric, $\tilde R^a_{\nu\sigma}= \tilde R^a_{\sigma\nu}$.
It follows that $\tilde R^a_{\nu\sigma}=R^a_{(\nu\sigma)}$ and
$R^a_{[\nu\sigma]}=\half\epsilon^{ab}\lc_{\nu\sigma\mu\alpha}R^{\mu\alpha}_b$.
The Weyl tensors are then defined as usual in terms of $\tilde
R^a_{\mu\nu\rho\sigma}$,
\begin{equation}
C^a_{\mu\nu\rho\sigma}=\tilde R^a_{\mu\nu\rho\sigma}-
{K^{\lambda\tau}}_{\mu\nu\rho\sigma}[\tilde R^a_{\lambda\tau}],
\end{equation}
and satisfy all standard symmetry properties: skew-symmetry in the first
and last pairs of indices, tracelessness (because
$\tilde R^a_{\nu\sigma}={K^{\lambda\tau\,
    \mu}}_{\nu\mu\sigma}[\tilde R^a_{\lambda\tau}]$), the cyclic identity
(because
$\epsilon^{\gamma\nu\rho\sigma}{K^{\lambda\tau}}_{\mu\nu\rho\sigma}
[\tilde R^a_{\lambda\tau}]=0$), which implies also symmetry in the
exchange of the first and last pair of indices,
\begin{gather}
  \label{eq:80}
  C^a_{\mu\nu\rho\sigma}  =  - C^a_{\nu\mu\rho\sigma} =
  -C^a_{\mu\nu\sigma\rho},\\
C^{\mu a}_{\nu\mu\sigma}=0,\quad
\epsilon^{\gamma\nu\rho\sigma}C^a_{\mu\nu\rho\sigma}=0, \quad 
C^a_{\mu\nu\rho\sigma}  =   C^a_{\rho\sigma\mu\nu}. 
\end{gather}
As before, the 10 independent components of the Weyl tensor can be
parametrized by the electric and magnetic components
$E^a_{mn}\equiv(E_{mn},B_{mn})$, symmetric and traceless tensors
defined by 
\begin{equation}
  \label{eq:GRAVdefelecWeyl}
  E^a_{mn}=C^a_{0m0n}=\half \epsilon_{njk}\epsilon^{ab}
  C_{b\,0m}^{\phantom{\mu\nu}jk}.
\end{equation}

Putting all definitions together, the relation between the Riemann and
Weyl tensors is
\begin{align}
  \label{eq:GRAVgeneWeyl}
R^a_{\mu\nu\rho\sigma}&=C^a_{\mu\nu\rho\sigma}+
{K^{\lambda\tau}}_{\mu\nu\rho\sigma}[R^a_{\lambda\tau}]+\half
\epsilon^{ab}\lc_{\rho\sigma\alpha\beta}
{{K^{\lambda\tau}}_{\mu\nu}}^{\alpha\beta}[R_{b(\lambda\tau)}]\\
&=C^a_{\mu\nu\rho\sigma}+
{K^{\lambda\tau}}_{\mu\nu\rho\sigma}[R^a_{(\lambda\tau)}]+\half
\epsilon^{ab}\lc_{\rho\sigma\alpha\beta}
{{K^{\lambda\tau}}_{\mu\nu}}^{\alpha\beta}[R_{b\lambda\tau}].
\end{align}
In particular, it follows that the $36$ independent components of the
Riemman tensor $R^1_{\mu\nu\rho\sigma}$ can be parameterized by the
$10$ independent components of the Weyl tensor
$C^1_{\mu\nu\rho\sigma}$, the $16$ components of the Ricci
tensor $R^1_{\lambda\tau}$, and the $10$ components of
  $R^2_{(\lambda\tau)}$. 

If we define
\begin{equation}
  \label{eq:GRAVdefnewvar}
  \cE^a_{mn}=R^a_{0(m|0|n)}, \quad \cF^{am}=\half \epsilon^{mjk} 
R^a_{0[j|0|k]},\quad \cR^a_{mn}=R^a_{(mn)}+\cE^a_{mn}
\end{equation}
the parameterization consisting in choosing the symmetric tensors
$\cE^a_{mn},\cR^a_{mn}$ (24 components), $\cF^a_m$, (6 components),
and $R^1_{[\mu\nu]}(={}^*R^2_{[\mu\nu]})$ (6 components) is more
useful for our purpose.  That all tensors can be reconstructed from
these variables follows from the fact that
\begin{gather}
  R^a_{0m}=-2\epsilon^{ab}\cF_{bm},\quad
  R^a_{00}=\cE^{a},\quad R^a_{(mn)}=\cR^a_{mn}-\cE^a_{mn}. 
\end{gather}
This means that the symmetric part of the Ricci tensors can be
reconstructed from the variables. Since the antisymmetric parts belong
to the variables, so can the complete Ricci tensors $R^a_{\mu\nu}$.
Using now \eqref{eq:GRAVgeneWeyl} and definitions
\eqref{eq:GRAVdefelecWeyl}, \eqref{eq:GRAVdefnewvar}, \eqref{eq:GRAVdefK}, we find
\begin{equation}
  E^a_{mn}=\half(\cE^a_{mn}+\cR^a_{mn})-\frac{\delta_{mn}}{6}(\cE^{a}+ 
\cR^{a}).
\end{equation}
It follows that the Weyl tensors and then, using again \eqref{eq:GRAVgeneWeyl},
the Riemann tensors can  be reconstructed. 

In terms of the new parameterization, the equations of motion
\eqref{eq:GRAVgeneEOM} read $R^a_{[\mu\nu]}=0$ and
\begin{gather}
  \label{eq:GRAVnewEOM1}
  -2\epsilon^{ab}\cF_{bm}=8\pi G T^{a}_{0m}, \\
  \half \cR^{a}=8\pi G T^a_{00},\\
  \cR^a_{mn}-\cE^a_{mn}+\delta_{mn}(\cE^{a}-\half 
  \cR^{a})=8\pi G T^a_{mn}. \label{eq:GRAVnewEOM3}
\end{gather}
Using these equations of motion, the Bianchi identities
\eqref{eq:GRAVgenebianch} are equivalent to 
\begin{equation}
  \label{eq:GRAVnewEOM4}
  \d^k(\epsilon_{ikm}\cF^{am}+\cR^a_{ik})=\half\d_i\cR^a,
\end{equation}
\begin{equation}
  \label{eq:GRAVnewEOM5}
  2\epsilon^{ab}\d_0\cF_{bm}=\d^n
  (\cE^a_{mn}+\epsilon_{mnk} \cF^{ak})-\d_m\cE^a, 
\end{equation}
\begin{multline}
  \label{eq:GRAVnewEOM6}
  \d_0\cR^a_{ik}=\half\epsilon^{ab}\big[
   \epsilon_{kjl}\d^j{\cE_{b i}}^l+
\epsilon_{ijl}\d^j{\cE_{b k}}^l-2\delta_{ik}\d_j\cF^j_b-\d_i
\cF_{bk}-\d_k\cF_{bi}
\big]\iff\\
\epsilon^{ab}\d_0(\cR^{ik}_{b}-\half\delta^{ik}
\cR_b)=-\half
\big[\epsilon^{klm}\d_l{\cE^a_{m}}^i+\epsilon^{ilm}
\d_l{\cE^a_{m}}^k+2\delta^{ik}\d^j
\cF^a_j-\d^i\cF^{ak}-\d^k\cF^{ai}\big]. 
\end{multline}

\vspace{5mm}

We will now express the Riemann tensor in terms of the canonical
variables in such a way that the covariant equations
\eqref{eq:GRAVnewEOM1}-\eqref{eq:GRAVnewEOM6} coincide with the Hamiltonian equations
deriving from \eqref{eq:GRAVinteractaction}. 

From the constraints with sources, we find
\begin{gather}
  \label{eq:70}
  \cR^a=\d^m\d^n h^a_{mn}-\Delta h^a=-\Delta^2 C^a,\\
\cF^a_{m}=\half\Delta\d^n H^a_{mn}.  \label{eq:70a}
\end{gather}
Assuming $\Delta$ to be invertible, which we do in the rest of this
section, $\cR^a$ and $C^a$, respectively $\cF^a_m$ and $\d^n
H^a_{mn}$ determine each other. By taking the divergence, the Bianchi
identity \eqref{eq:GRAVnewEOM4} implies that
\begin{equation*}
\d^m\d^n\cR^a_{mn}=-\half\Delta^3 C^a.
\end{equation*}
Similarly, the Bianchi identity \eqref{eq:GRAVnewEOM5} implies in particular
that $\Delta\cE^a-\d^m\d^n\cE^a_{mn}=\epsilon^{ab}\d_0\Delta\d^m\d^n
H_{bmn}$. When combined with \eqref{eq:GRAVnewEOM3}, the equations of motion
following from variation with respect to $C^a$ read
\begin{equation*}
  \frac{1}{2}\Delta^3 C^a+ \epsilon^{ab} \Delta\d_0
  (\Delta H_b-\d^m\d^n
  H_{bmn})+ 2\Delta^2 n^a=
\Delta\cE^a-\d^m\d^n(\cR^a_{mn}-\cE^a_{mn}).
\end{equation*}
When combined with the previous relations, they imply that 
\begin{gather*}
  \label{eq:85}
  \cE^a= -\half \epsilon^{ab}\d_0\Delta H_b+ \Delta n^a,\\
  \d^m\d^n\cE^a_{mn}= -\half\epsilon^{ab}\d_0\Delta(\Delta
  H_b-2\d^m\d^n H_{bmn})+ \Delta^2 n^a.
\end{gather*}
The rest of the Bianchi identities \eqref{eq:GRAVnewEOM4}, \eqref{eq:GRAVnewEOM5} are
taken into account by applying a curl. This gives
$\epsilon^{rsi}\d_s \d^k\cR^a_{ik}=\half \Delta (\Delta\d^k
H^{ar}_k-\d^r \d^m\d^n H^a_{mn})$ and $\epsilon^{rsi}\d_s
\d^k\cE^a_{ik}=\epsilon^{rsi}2\epsilon^{ab}\d_0\d_s\cF_{bi}-
\d^r\d^k \cF^a_k+\Delta \cF^{ar}$. Yet another curl gives
$\d_k\d^m\d^n \cR^a_{mn}-\Delta \d^n\cR^a_{kn}=\half
\epsilon_{klr}\d^l \Delta^2 \d^n H^{ar}_n$ and $\d_k\d^m\d^n
\cE^a_{mn}-\Delta \d^n\cE^a_{kn}=2\epsilon^{ab}\d_0(\d_k \d^n \cF_{bn}
-\Delta \cF_{b k})+ \epsilon_{klr}\d^l \Delta
\cF^{ar}$. Using the previous relations we then get
\begin{gather*}
  \d^n\cR^a_{kn}=-\half \d_k \Delta^2 C^a-\half
  \epsilon_{klr}\d^l \Delta \d^n H^{ar}_n,\\
  \d^n\cE^a_{kn}= \epsilon^{ab}\d_0\Delta(-\frac{1}{2}\d_k
  H_b+\d^n H_{bkn})+ \d_k \Delta n^a-\half \epsilon_{klr}\d^l
  \Delta \d^n H^{ar}_n.
\end{gather*}

The equations of motion following from variation with respect to
$A^a_m$ are then identically satisfied. 

Defining $\cD^a_{mn}=\cR^a_{mn}-\cE^a_{mn}$ and using definition
\eqref{eq:GRAVdefPTT} of $\cP^{TT}$ combined with \eqref{eq:GRAVnewEOM3}, the equations
of motion following from variation with respect to $H^a_{mn}$ read
\begin{multline}
\epsilon_{ab}\d_0\Big[2 \left( \cP^{TT} H^b\right)_{mn} + \d_m \Delta
A_n^b +\d_n \Delta A_m^b +\half(\delta_{mn}\Delta-\d_m\d_n)C^b
\Big]-\epsilon_{ab} \Delta (\d_m n^b_n+\d_n n^b_m)-\\-2\Delta^2 H^a_{mn}
+\delta_{mn} \Delta^2 H^a
=-{\epsilon_{mpq}}\d^p\cD^q_{an}-{\epsilon_{npq}}\d^p\cD^q_{am}. \label{eq:GRAVeqH}
\end{multline}
Taking into account definition \eqref{eq:GRAVdefPTT} and previous relations, we
can extract
\begin{multline}
  -\Delta^{-1}\left(\cP^{TT} \cD^a\right)_{mn}=
\half \epsilon_{ab}\d_0\Big[2\left( \cP^{TT} H^b\right)_{mn}
-\epsilon_{mpq}\d_n \d^p \d^r H^{bq}_r-\epsilon_{npq}\d_m \d^p \d^r
H^{bq}_r +\\ + \d_m \Delta
A_n^b + \d_n \Delta A_m^b +\frac{1}{2}(\delta_{mn}\Delta-\d_m\d_n)C^b
\Big]-\epsilon_{ab} \Delta (\d_m n^b_n+\d_n n^b_m)-\\-\Delta^2 H^a_{mn}
+\half \delta_{mn} \Delta^2 H^a.
\end{multline}
In order to extract the remaining information from \eqref{eq:GRAVeqH}, 
we first apply $\delta^{mn}\Delta-\d^m\d^n$ to get 
\begin{equation}
  \label{eq:98}
  \epsilon_{ab}\d_0 \Delta C^b+2 \d^m\d^n H_{amn}=0,
\end{equation}
and then a divergence $\d^m$ giving
\begin{equation}
  \label{eq:114}
  \epsilon_{ab}\d_0 (\Delta A^b_n-\epsilon_{npq}\d^p\d^k H^{b
    q}_k)=\epsilon_{ab} \Delta n_n^b+2 \Delta \d^k H^k_{an}-\half \d_n
  \Delta H^a -\d_n \d^k\d^l H^a_{kl}. 
\end{equation}
We can now inject the latter relations into \eqref{eq:GRAVeqH} and use
\eqref{eq:GRAVpropP}, \eqref{eq:GRAVdecomp} to get
\begin{align}
  \label{eq:93}
  \cD_{mn}^{aTT}&=-\epsilon^{ab}  \d_0 \Delta
  H^{TT}_{bmn}-\left(\cP^{TT} H^a\right)_{mn}, \\
\cD_{mn}^{a}&=-\epsilon^{ab}\d_0\Delta 
\Big[ H_{bmn}-\half \delta_{mn} H_b\Big]  -\left(\cP^{TT}
  H^a\right)_{mn}-\d_m\d_n n^a-\nonumber \\ & \hspace{4cm}
-\frac{1}{4}(\delta_{mn}\Delta+\d_m\d_n)\Delta C^a.  \label{eq:GRAVEOMintrem2}
\end{align}

Injecting into the second form of the last Bianchi identity
\eqref{eq:GRAVnewEOM6} and using previous relations gives
\begin{multline}
  \label{eq:GRAVEOMintrem1}
  \epsilon_{ab}\d_0 \cR^b_{ij}=-\left( \cO \cR_a\right)_{ij}+\Delta^2 H^{TT}_{aij}
 +\frac{1}{4} \Delta \d_i \d^k H_{akj} +\frac{1}{4} \Delta \d_j \d^k
  H_{aki} -\half \d_i\d_j\d^k\d^l H_{akl}\\-\half
  \epsilon_{ab}\d_0\Big[\epsilon_{iqn} \d^q\Delta H^{bn}_j +
  \epsilon_{jqn} \d^q \Delta
  H^{bn}_i+\half(\delta_{ij}\Delta+\d_i\d_j)\Delta C^b
  \Big].
\end{multline}
Identifying the terms with time derivatives gives
\begin{gather}
  \label{eq:96}
  \cR^a_{ij}=-\half\Big[\epsilon_{iqn} \d^q\Delta
  H^{an}_j + \epsilon_{jqn} \d^q
  \Delta H^{an}_i+\half
  (\delta_{ij} \Delta+\d_i\d_j)\Delta C^a \Big]\nonumber\\
  =\half\Big[\d_i\d^k h^a_{kj}+\d_j\d^k h^a_{ki}-\d_i\d_j h^a-\Delta
  h^a_{ij}- \epsilon_{ikl}\d^k\d^p\d_j H^{al}_p- \epsilon_{jkl}
  \d^k\d^p\d_i H^{al}_p\Big].
\end{gather}
The terms without time derivatives in \eqref{eq:GRAVEOMintrem1} then cancel
identically.  Together with \eqref{eq:GRAVEOMintrem2} this then finally gives
\begin{gather}
  \label{eq:97}
  \cE_a^{ij}=\epsilon_{ab}\d_0\Delta \Big[ H^{bij}-\half\delta^{ij}
  H^b\Big]+\d^i\d^j n_a- \half \epsilon^{ikl}\d_k\d^j
  \d^p H_{alp}-\half \epsilon^{jkl}\d_k\d^i \d^p H_{alp}\nonumber\\
  =-\epsilon_{ab}\d_0(\pi^{bij}-\half \delta^{ij}\pi^b)
+\d^i\d^j n_a- \half \epsilon^{ikl}\d_k\d^j \d^p H_{alp}
-\half \epsilon^{jkl}\d_k\d^i \d^p H_{alp}.
\end{gather}

\normalsize

\subsection{Linearized Taub-NUT solution}
\label{sec:point-part-grav}

We start by considering the sources corresponding to a point-particle
gravitational dyon with electric mass $M$ and magnetic mass $N$ at
rest at the origin of the coordinate system, for which
\begin{equation}
  \label{eq:GRAVTaubNutsource}
  T^{\mu\nu}_a(x)=\delta^\mu_0\delta^\nu_0 M_a\delta^{(3)}(x^i), \qquad
  M_a=(M,N). 
\end{equation}
In this case, only the constraints \eqref{eq:24a} are affected by the
interaction and become
\begin{equation}
  \label{eq:63}
  \cH_{a\perp}=-16\pi GM_a\delta^{(3)}(x). 
\end{equation}
They are solved by 
\begin{equation}
  \label{eq:64}
  \Delta C^a= G M^a(\frac{4}{r}), 
\end{equation}
where $r=\sqrt{x^i x_i}$. 
It is then straightforward to check that all equations of motions are
solved by
\begin{gather}
  \label{eq:66}
C^a=GM^a( 2 r),\quad  
n^a=G M^a(-\frac{1}{r}), \quad
A^a_m=n^{am}=H^a_{mn}=0,\nonumber \\
  h^a_{mn}=GM^a(\delta_{mn}+\frac{x_mx_n}{r^3}),\quad
  \pi_a^{mn}=0.
\end{gather}

The usual Schwarzschild form is obtained by adding a pure gauge
solution with parameter $\xi^{am}=GM^a(-\half \frac{x^m}{r})$,
$\xi^{a\perp}=0$. The solution then reads
\begin{gather}
  C^a=GM^a(2r),\
  n^a=G M^a(-\frac{1}{r}),\
  A^a_m=GM^a(-\half \frac{x_m}{r}),\ n^{am}=H^a_{mn}=0,\nonumber
  \\
  h^a_{mn}=GM^a(\frac{2 x_mx_n}{r^3}),\quad \pi_a^{mn}=0.
\end{gather}

To show that this solution describes the linearized Taub-NUT solution,
we need to compute its Riemann tensor using the
relations given in the previous section. Following for instance \cite{Cohen-Tannoudji:1989dq} section $A_1.2$ and
using a regularization in Fourier space, we find 
\begin{align}
  \label{eq:99}
  \cR^a_{ij}&=GM^a\big[\frac{16\pi}{3}\delta_{ij}
  \delta^3(x)+\frac{\eta(r)}{r^3}(\delta_{ij}-\frac{3x_ix_j}{r^2})\big],\\
  \cE^a_{ij}&= GM^a\big[\frac{4\pi}{3}\delta_{ij}
  \delta^3(x)+\frac{\eta(r)}{r^3}(\delta_{ij}-\frac{3x_ix_j}{r^2})\big],
\end{align}
where $\eta(r)$ is a regularizing function that suppresses the
divergence at the origin and is $1$ away from the origin. We then find
\begin{gather}
  \label{eq:101}
  R^a_{00}=GM^a 4\pi \delta^3(x),\quad R^a_{ij}=GM^a 4\pi
  \delta_{ij} \delta^3(x),\\
E^a_{ij}=GM^a \frac{\eta(x)}{r^3}(\delta_{ij}-\frac{3x_ix_j}{r^2}),
\end{gather}
and all other components of $R^a_{\mu\nu}$ vanishing. For the Riemann
tensor, this implies
\begin{gather}
  \label{eq:102}
  R^a_{0i0j}=GM^a\big[\frac{4\pi}{3}\delta_{ij}
  \delta^3(x)+\frac{\eta(x)}{r^3}(\delta_{ij}-\frac{3x_ix_j}{r^2})\big],\\
R^a_{0ijk}=-\epsilon^{ab}{\epsilon_{jk}}^lGM^a\big[\frac{4\pi}{3}\delta_{il}
  \delta^3(x)+\frac{\eta(x)}{r^3}(\delta_{il}-\frac{3x_ix_l}{r^2})\big],
\end{gather}
with all other components obtained through the on-shell symmetries of
the Riemann tensor. This is the usual Riemann tensor for the
linearized Taub-NUT solution.

As in the electromagnetism case, this formalism resolves the string
singularity of the linearized Taub-NUT solution present in the standard
Pauli-Fierz formulation. In spherical coordinates, the latter can for
instance be described by
\begin{equation}
  h_{rr}=\frac{2GM}{r}=h_{00},\qquad h_{0\varphi}=-2N(1-\cos\theta),
\label{eq:74}
\end{equation}
and all other components vanishing, with a string-singularity along
the negative $z$-axis.

\subsection{Electric and magnetic energy-momentum and angular momentum surface charges }
\label{sec:electr-magn-energy}

As for the spin 1 case, the analysis of appendix  is not directly
applicable since we do not have Darboux coordinates and the Poisson
brackets of the fundamental variables are non-local. Another problem
is that the gauge transformations \eqref{eq:GRAVgaugetransf} do not allow for non trivial solutions
to $\delta_{\varepsilon_s }z^A=0$. As before, we will use the idea of
the appendix to derive expressions for the surface charges. We still
have to keep the sources explicitly throughout the argument because
of the presence of $\Delta^{-1}$.

In the presence of the sources, the constraints
$\gamma_\alpha^J=(\cH_{am}^J,\cH_{a\perp}^J)$ are determined
\begin{equation}
\cH_{am}^J=\cH_{am}-(16\pi G) T_{am}^0,\quad
\cH_{a\perp}^J=\cH_{a\perp} -(16\pi G) T_{a0}^0.\label{eq:38}
\end{equation}
Instead of \eqref{eq:app29c}, we can write
\begin{multline}
  \label{eq:GRAVPARTDERIVCONSTRAINTS}
  \gamma^J_\alpha\varepsilon^\alpha 
  =(\d^m\xi^{an}+\d^n\xi^{am})\epsilon_{ab} \pi^b_{mn}+
  (\delta^{mn}\Delta-\d^m\d^n)\xi^{a\perp} h_{amn} -\d_i
 \tilde k^{i}_{\varepsilon}[z]\\ -(16\pi G)
  (T_{am}^0\xi^{am}+T_{a0}^0\xi^{a\perp}) ,
\end{multline}
where
\begin{equation}
  \tilde k^{i}_{\varepsilon}[z] =2\xi^a_m\epsilon_{ab}\pi^{bmi}-
  \xi^{a\perp}(\delta^{mn}\d^i-\delta^{mi}\d^n) h_{amn}+
  h_{amn}(\delta^{mn}\d^i-\delta^{ni}\d^m)\xi^{a\perp}.
\end{equation}
Consider now gauge parameters $\epsilon^\alpha_s(x)$ 
satisfying the conditions
\begin{equation}
  \label{eq:GRAVgaugetransgene}
  \left\{\begin{array}{c}\d^m\xi^{an}_s+\d^n\xi^{am}_s=0=
      \d_0\xi^{am}_s-\d^m
      \xi^{a\perp}_s,  \\
      (\delta^{mn}\Delta-\d^m\d^n)\xi^{a\perp}_s=0=
 \d_0 \xi^{a\perp}_s, \end{array}\right.
\end{equation}
The general solution to conditions \eqref{eq:GRAVgaugetransgene} can be written as
\begin{equation}
  \label{eq:GRAVreducibilityparam}
  \xi^{a}_{\mu s}=-\omega^{a}_{[\mu\nu]}x^\nu+a^{a}_\mu,
\end{equation}
for some constants $a^{a}_\mu$,
$\omega^{a}_{[\mu\nu]}=-\omega^{a}_{[\nu\mu]}$.
It follows in particular that the surface charges 
\begin{equation}
\cQ_{\varepsilon_s}[z_s]=\frac{1}{16\pi
  G}\oint_{S} d^3x_i\,
 \tilde k^{i}_{\varepsilon_s}[z_s],\label{eq:GRAVsurfchargescomplete}
\end{equation}
do not depend on the homology class of $S$ outside of sources.

Assuming $\Delta$ invertible, the equations of motion associated to
$\cL_T=\frac{1}{16\pi G}\cL_H+\cL^J$ imply in particular that
\begin{multline}
  \label{eq:94}
  \d_0 h^a_{mn}=\d_m n_n^a+\d_nn_m^a-2\epsilon^{ab}\Delta
  H_{bmn}+\epsilon^{ab}\delta_{mn}\Delta H_b \\ 
  +(16\pi G) \epsilon^{ab}\Big(\Delta^{-1} \left(\cO
    T_b\right)_{mn}+\half \Delta^{-2}\d_m\epsilon_{npq}\d^p\d_k
  T^{kq}_b+\half \Delta^{-2}\d_n\epsilon_{mpq}\d^p\d_k T^{kq}_b\Big),
\end{multline}
\begin{multline}
 \epsilon_{ab} \d_0 \pi^b_{mn} =\left(\cP^{TT} H_a\right)_{mn}+(8\pi
  G) T_{amn}-\\-\half (\delta^{mn}\Delta-\d^m\d^n) (2n_a+\half \Delta C_a).
\end{multline}
By direct computation using the equations of motion, one then finds
\begin{equation}
  \label{eq:GRAVconservationCharge}
  \d_0 \tilde k^i_{\varepsilon_s}[z_s]=(16\pi G)(\xi^a_{\mu s} T_a^{\mu
    i})-\d_j k^{[ij]}_{\varepsilon_s}[z_s,u_s],
\end{equation}
with 
\begin{multline}
  \label{eq:107}
  k^{[ij]}_{\varepsilon_s}[z,u]=\Big(2n_a^i\d^j\xi^{a\perp}_s+\xi^{a\perp}_s\d^i
  n_a^j+\xi^{ai}_s\d^j(2n_a+\half \Delta C_a)+\xi^{a}_{s
    m}\epsilon^{mpq}\d_p\d^i H^j_{aq}\\+ \omega^{aj} \d^k
  H^i_{ak}+ \omega^{ai}\d^jH_a+2\omega^{ak}\d^i H^j_{ak}+16\pi
  G\epsilon^{ab}\epsilon^{imq}\Delta^{-1} T_{bq}^j\d_m\xi_{as}^\perp
  \\+8\pi G\epsilon^{ab}\epsilon^{mpq}\d_p\Delta^{-2} \d^i
  T^j_{bq}\d_m\xi_{as}^\perp -(i\longleftrightarrow j)\Big)
  \\+\epsilon^{ijk}\Big[\omega^a_k(2n_a+\half \Delta C_a) -
  \xi^{am}_s(\Delta H_{amk}-\d_m\d^rH_{ark})\\-16\pi
  G\epsilon^{ab}\Delta^{-1} \d^r T_{brk}\xi_{as}^\perp+8\pi
  G\epsilon^{ab}(\Delta^{-1}T^m_{bk}+\Delta^{-2}\d^m\d^rT_{brk})\d_m\xi^\perp_{as}
  \Big],
\end{multline}
where $\omega^a_{mn}=\omega^{ak}\epsilon_{kmn}$. The surfaces charges
\eqref{eq:GRAVsurfchargescomplete} are thus also time-independent outside of sources. 

Finally, the surface charges are gauge
invariant,
\begin{gather}
  \tilde k^i_{\varepsilon_s}[\delta_{\eta}z]=\d_j
  r^{[ij]}_{\varepsilon_s,\eta}, \\
r^{[ij]}_{\varepsilon_s,\eta}=\Big(2\xi^{aj}_{s}\d^i\eta^\perp_a+
2\eta_a^j\d^i\xi^{a\perp}_s+\xi^{a\perp}_s\d^j
  \eta_a^i-(i\longleftrightarrow j)\Big)
-2\epsilon^{ijk}\omega^a_k\eta^\perp_a. 
\label{eq:112}
\end{gather}

Defining 
\begin{equation}
Q_{\epsilon_s}[z]=\half\omega^a_{\mu\nu}J^{\mu\nu}_a-a^a_\mu
P^\mu_a,
\end{equation}
we get for the individual generators
\begin{empheq}[box=\fbox]{align}
(16\pi G)  P_a^\perp & =  -\oint_{S^\infty} d^3x_m \, \partial^m \Delta C_a = 
\oint_{S^\infty} d^3x_m \,\left( \partial_nh^{mn}_a - \partial^mh_a\right),\\
(16\pi G)  P_a^n & = 2\oint_{S^\infty} d^3x_m \, \epsilon_{ab} 
\Delta H^{bnm} =- 2\oint_{S^\infty} d^3x_m \, \epsilon_{ab} \pi^{bmn},\\
(16\pi G)  J_a^{kl} & =  2\oint_{S^\infty} d^3x_m \, 
\epsilon_{ab} \left( \Delta H^{bmk}x^l-\Delta H^{bml}x^k\right) \\ &
\quad = -2\oint_{S^\infty} d^3x_m \, \epsilon_{ab} 
\left( \pi^{bmk}x^l-\pi^{bml}x^k\right),\\
(16\pi G)  J^{\perp k}_a &=  \oint_{S^\infty} d^3x_m \, \left( \Delta C^a
    \delta^{mk} - \partial^m \Delta C_a x^k\right)\\ & \quad =
  \oint_{S^\infty} d^3x_m \, \left[ \left(\partial_n h^{mn}_a
      - \partial^m h_a \right) x^k -h^{mk}_a + h_a 
    \delta^{mk}\right].
\end{empheq}

The only non-vanishing surface charges of the dyon sitting at the
origin are  
\begin{equation}
 P^\perp_a= M_a.
\end{equation}
As expected, they measure the electric and magnetic mass of the dyon.

For later use, we combine $\tilde
k^i_\varepsilon,k^{[ij]}_{\varepsilon}$ into the $n-2$ forms
$k_\varepsilon[z,u]$ through the following expressions in Cartesian
coordinates, 
\begin{gather}
  k_{\varepsilon}=k^{[\mu\nu]}_\varepsilon d^{2}x_{\mu\nu},\qquad 
  k^{[0i]}_\varepsilon =\tilde k^i_\varepsilon,\\
  d^{n-k}x_{\mu_1\dots\mu_k}=\frac{1}{k!(n-k)!}
  \epsilon_{\mu_1\dots\mu_k\nu_{k+1}\dots
    \nu_n}dx^{\nu_{k+1}}\dots dx^{\nu_{n}}.\label{eq:135} \end{gather}
Equations \eqref{eq:GRAVconservationCharge} can then be summarized by 
\begin{equation}
  \label{eq:136}
  d k_{\varepsilon_s}\approx -(16\pi G) T_{\varepsilon_s},\quad
T_{\varepsilon_s}=T^{\mu}_{a\nu}\xi^{a\nu}_sd^{3}x_\mu,\quad
dT_{\varepsilon_s}=0,
\end{equation}
where closure of the $n-1$-forms $T_{\varepsilon_s}$ follows from the
conservation of the sources, the symmetry of the energy-momentum
tensor and \eqref{eq:GRAVreducibilityparam}.

\subsection{Poincar\'e transformations of surface charges}
\label{sec:poinc-transf-pres}

Suppose now that $z^A_s,u^\alpha_s$ solve the equations of motions for
the conserved sources $T^{\mu\nu}_a(x)$. Let $z^{\prime A}_s,u^{\prime
  \alpha}_s$ be the solution associated to new sources $T^{\prime
  \mu\nu}_a(x)$ related to $T^{\mu\nu}_a(x)$ through a (proper)
Poincar\'e transformation, $x^{\prime\mu}={\Lambda^\mu}_\nu
x^\nu+b^\mu$ with $|\Lambda|=1$, \begin{equation}
  \label{eq:130}
  T^{\prime\mu\nu}_a(x^\prime)={\Lambda^\mu}_\alpha{\Lambda^\nu}_\beta
  T^{\alpha\beta}_a(x).  
\end{equation}
For instance, starting from the conserved energy-momentum tensors
\eqref{eq:GRAVTaubNutsource} of a dyon sitting at the origin with world-line
$z^\mu=\delta^\mu_0 s$, one can obtain in this way the conserved
energy-momentum tensors of a dyon moving along a straight line,
$z^{\prime\mu}=u^\mu s+ a^\mu$ with $u^\mu,a^\mu$ constant, $u^\mu
u_\mu=-1$ and $s$ the proper time,
\begin{equation}
  \label{eq:113}
  T^{\prime\,\mu\nu}_a(x^\prime)=M_a u^\nu\int d\lambda \delta^{(4)}(x^\prime-
z^\prime(\lambda))  \frac{dz^{\prime\mu}}{d\lambda}=M_a\frac{u^\mu
      u^\nu}{u^0}\delta^{(3)}(x^{\prime i}-z^{\prime i}(x^0)).    
\end{equation}
when ${\Lambda^\mu}_0=u^\mu$.

Assume then that the $\xi^\mu_{a s}(x)$ transform like vectors
\begin{equation}
\xi^{\prime\nu}_{a
  s}(x^\prime)={\Lambda^\mu}_\alpha\xi^{\alpha}_{as}(x)
=-(\Lambda \omega_a\Lambda^{-1}
x^\prime)^\nu+(\Lambda \omega_a\Lambda^{-1}b+\Lambda a_a)^\nu,
\label{eq:137}
\end{equation}
which implies that the $T_{\varepsilon_s}$ are closed Poincar\'e invariant
$n-1$ forms, 
\begin{equation}
  \label{eq:133}
  T^{\prime}_{\varepsilon^\prime_s}(x^\prime,dx^\prime)=
  T_{\varepsilon_s}(x,dx). 
\end{equation}
We can then use the following variant of the tube lemma. If at fixed
time $t$, $T^{0\nu}_a(x)\xi^{a\nu}_s(x)$ has compact support and there
exists a tube, i.e., a space-time volume $\cW$ connecting the hypersurfaces
$\Omega: x^0=t$ and $\Omega^\prime: t=x^{\prime 0}={\Lambda^0}_\nu
x^\nu+ b^0$ such
that $T_{\varepsilon_s}$ is entirely contained in $\cW$, it follows
from Stokes' theorem that 
\begin{equation}
  \label{eq:138}
  \int_\Omega T_{\varepsilon_s}=\int_{\Omega^\prime}
  T^\prime_{\varepsilon^\prime_s}. 
\end{equation}
If we now compute the surface charges for a large enough sphere $S$ at
fixed $t$ containing both $T^{0\nu}_a(x)\xi^{a\nu}_s(x)$ and
$T^{\prime 0\nu}_a(x)\xi^{a\nu}_s(x)$, it finally follows from
\eqref{eq:GRAVPARTDERIVCONSTRAINTS} that the surface charges evaluated for the new
solutions $z^{\prime A}$ are obtained from those of the old solutions
$z^A$ through 
\begin{equation}
  \label{eq:132} \cQ_{\varepsilon^\prime_s}[z^\prime_s]=\cQ_{
    \varepsilon}[z_s]. 
\end{equation}

\section{Conclusion}
In this chapter, we have developed an extended double potential
formalism for spin 2. This allowed us to write a manifestly duality
invariant action in presence of both electric and magnetic external sources. We
derived the expression of the surface charges in term of the
fundamental canonical fields obtaining in a duality invariant way both
the mass and the NUT charge. Those charges
can also be constructed using Lagrangian methods but, as such, are not duality
invariant (see e.g. \cite{Argurio:2009kx,Argurio:2009uq,Argurio:2010fk}).

In fact we have shown here that the standard
expressions for surface charges in Pauli-Fierz theory, when extended
in a duality invariant way, have all the expected properties. More
interesting would be to develop the theory of surface charges from
scratch in theories of the current type where the Poisson brackets
among the fundamental variables are not local to see if the ones we
have found exhaust all possibilities. From the preceding discussion we
see that pseudo-differential operators will play a crucial role for a
discussion of these generalized conservation laws, as they do in the
discussion of ordinary conservation laws for evolution equations of
the Korteweg-de Vries type for instance.

This association with the soliton theory is genuine. As we will
see in the next chapter, there is a close relation between
electromagnetic duality and integrable systems.

\chapter{Electromagnetic duality and integrability}
\label{ch:Integ}
A cornerstone of soliton theory is the discovery that the evolution
equations are Hamiltonian systems
\cite{Gardner:1971vn,Zakharov:1971ys}. In this context, the occurrence of
hierarchies of evolution equations sharing the same infinite set of
conservation laws can be understood as a consequence of the existence
of a second compatible Hamiltonian structure giving rise to the same
evolution equations \cite{Magri:1978zr,Gelfand:1979ly}.

The equations of motion associated to the theories for the known fundamental
forces of nature, electromagnetism, Yang-Mills theories and
gravitation, are variational and thus Hamiltonian. This is no
coincidence, since these theories are fundamentally quantum, at least
the first three of them, and only for variational theories
quantization is sufficiently well understood. 

In order to have Poincar\'e invariance, respectively diffeomorphism
invariance, manifestly realized, most modern investigations of these
equations are carried out in the Lagrangian framework. This could be
the reason why the bi-Hamiltonian structure underlying these equations
and discussed below has hitherto remained unnoticed.

At the heart of our analysis is an important exception to this
paradigm, namely the question whether the duality invariance of the
four dimensional Maxwell or linearized gravity equations
admits a canonical generator. This question has been answered to the
affirmative in the reduced phase space of these theories and
generalized to massless higher spin gauge fields
\cite{Deser:1976ve,Henneaux:2005qf,Deser:2005ve}.

We will show in this chapter that the reduced phase space formulation
of massless higher spin gauge fields is bi-Hamiltonian. We will start by a quick review of the theory of bi-Hamiltonian
systems. After that, we will first study the electromagnetic case and then go to
linearized gravity and  massless gauge fields of
spin higher than $2$ where the analysis of spin 1 can be carried over easily by
taking care of additional spatial indices.

\section{Bi-Hamiltonian systems}

We will now present the basic results of the
bi-Hamiltonian theory using the famous Korteweg-de Vries equation as an
exemple. We refer the reader to the book of Olver \cite{Olver:1986nx} for a more complete
presentation.

A bi-Hamiltonian system is remarkable in the sense that its evolution
equations can be written in Hamiltonian form in not just one but two
different ways. We are then interested in systems of the form
\begin{equation}
\label{eq:INTEGdefEOM}
\frac{\d z^A}{ \d t} = K^A_1 [z] = \left\{z^A, \cH_1 \right\}_1 = \left\{z^A, \cH_0 \right\}_0
\end{equation}
where $\{ , \}_1$ and $\{ , \}_0$ are two different Poisson brackets
associated to two different Hamiltonians $\cH_1$ and $\cH_2$. For
instance, the KdV
equation is given by
\begin{equation}
\frac{\d u (x)}{\d t} = \d_x^3 u + u \d_x u.
\end{equation}
It can be written as a Hamiltonian equation in two different ways:
\begin{eqnarray}
\left\{F, G \right\}_1 & = & \int dx \frac{\delta F}{\delta u(x)} \d_x
\frac{\delta G}{\delta u(x)} ,\\
\cH_1 & = & \int dx \left(-\half (\d_x u)^2 + \frac{1}{6} u^3\right),
\end{eqnarray}
and
\begin{eqnarray}
\left\{F, G \right\}_0 & = & \int dx \frac{\delta F}{\delta u(x)}
\left( \d_x^3 + \frac{2}{3} u \d_x u + \frac{1}{3} \d_x u\right)
\frac{\delta G}{\delta u(x)} ,\\
\cH_0 & = & \int dx \,\half u^2.
\end{eqnarray}

Being Hamiltonians, both $\cH_1$ and $\cH_2$ are conserved quantities
and, as such, generate symmetries through both Poisson brackets. The
following three transformations are then symmetries of the equations
of motion (\ref{eq:INTEGdefEOM}):
\begin{eqnarray}
\delta_0 z^A & = & K^A_0 [z] = \left\{z^A, \cH_0 \right\}_1,\\
\delta_1 z^A & = & K^A_1 [z] = \left\{z^A, \cH_1 \right\}_1 = \left\{z^A, \cH_0 \right\}_0,\\
\delta_2 z^A & = & K^A_2 [z] = \left\{z^A, \cH_1 \right\}_0.
\end{eqnarray}
Let's assume that $\delta_2 z^A$ is a Hamiltonian vector field for $\{
, \}_1$, i.e. there exists $\cH_2$ such that
\begin{equation}
\delta_2 z^A = \left\{z^A, \cH_2 \right\}_1.
\end{equation}
In that case, $\cH_2$ is a new conserved quantity and it generates a
new symmetry through the other canonical structure:
\begin{equation}
\delta_3 z^A  =  K^A_3 [z] = \left\{z^A, \cH_2 \right\}_0.
\end{equation}
If this continues, one can build an infinite tower of conserved quantities
$\cH_0, \cH_1, \cH_2, ...$ associated to an infinite tower of
symmetries whose characteristics are $K_0, K_1, K_2, ...$

A key point in this argument is the assumption that $\delta_2 z^A$ is
a Hamiltonian vector field for $\{ , \}_1$. In order to control this
property, we need to introduce the notion of
compatible Hamiltonian structures:
\begin{definition}
\quad  

\begin{itemize}
\item Two Hamiltonian structures  $\{ , \}_1$ and $\{ , \}_0$  are
  said to be compatible if for all $a,b \in \RR$, $a \{ , \}_0 + b \{
  , \}_1$ is a Hamiltonian structure. They form a Hamiltonian pair.
\item A system of evolution equations is a bi-Hamiltonian system if it
  can be written in the form (\ref{eq:INTEGdefEOM}) where $\{ , \}_1$
  and $\{ , \}_0$ form a Hamiltonian pair. 
\end{itemize}
\end{definition}

To any poisson bracket $\{ , \}$, we can associate a skew-adjoint
linear differential operator $J^{AB}$ such that:
\begin{equation}
\left\{F,G \right\} = \int d^nx \frac{\delta F}{\delta z^A(x)} J^{AB}
\left( \frac{\delta G}{\delta z^B(x)}\right) = - \int d^nx J^{BA}
\left( \frac{\delta F}{\delta z^A(x)}\right)\frac{\delta G}{\delta z^B(x)}.
\end{equation}
This operator is called a Hamiltonian operator. We will use the
notation $\cD^{AB}$ for the Hamiltonian operator associated to $\{ , \}_1$  and $\cE^{AB}$ the
one associated to $\{ , \}_0$. The two Hamiltonian operators
associated to the KdV equation are given by
\begin{equation}
\label{eq:INTEGkdvop}
\cD = \d_x,\quad  \cE = \d_x^3 + \frac{2}{3} u \d_x + \frac{1}{3} \d_x u.
\end{equation}

With those structures, we can state the main theorem on bi-Hamiltonian
systems.
\begin{theorem}
Let 
\begin{eqnarray}
\label{eq:INTEGevoleq}
\frac{\d z^A}{\d t} = K^A_1 [z] & = & \left\{z^A, \cH_1 \right\}_1 =
\cD^{AB} \frac{\delta \cH_1}{ \delta z^B} \\
 & = & \left\{z^A, \cH_0 \right\}_0 =
\cE^{AB} \frac{\delta \cH_0}{ \delta z^B} \nonumber
\end{eqnarray}
be a bi-Hamiltonian system of evolution equations. Assume that the
operator $\cD^{AB}$ is non-degenerate and define $\cR^A_B = \cE^{AC}
\cD^{-1}_{CB}$. Let $K^A_0[z] =\left\{z^A, \cH_1 \right\}_0$ and
assume that for all $n \in \NN$ we can recursively define
\begin{equation}
K^A_n[z] = \cR^A_B K^B_{n-1}[z] 
\end{equation} 
meaning that for each $n$, $K_{n-1}$ lies in the image of $\cD$.
Then, there exists a sequence of functionals $\cH_0, \cH_1, \cH_2,
...$ such that
\begin{itemize}
\item for all $n \ge 1$, the evolution equation
\begin{equation}
\label{eq:INTEGhierarchy}
\frac{\d z^A}{\d t} = K^A_n[z] = \left\{z^A, \cH_n \right\}_1 = \left\{z^A, \cH_{n-1} \right\}_0
\end{equation}
is a bi-Hamiltonian system;
\item the Hamiltonian functionals $\cH_n$ are all in involution with
  respect to either Poisson bracket:
\begin{equation}
\left\{\cH_m, \cH_n \right\}_1 = 0 = \left\{\cH_m, \cH_n \right\}_0
\quad n,m \ge 0 
\end{equation}
and hence provide an infinite collection of conservation laws for each
of the bi-Hamiltonian systems (\ref{eq:INTEGhierarchy}).
\end{itemize}
\end{theorem}

In the KdV case, the operator $\cR$ is given by
\begin{equation}
\cR = \d_x^2 + \frac{2}{3} u + \frac{1}{3} \d_x u \d_x^{-1},
\end{equation}
where $\d_x^{-1}$ is a formal operator acting only on functions that
are total derivative (if $Q=\d_x P$, then we set $P = \d^{-1}_x Q$.
We can remove the ambiguity of the additive constant by normalizing
$P\vert_{u=0,x=0}=0$). The hierarchy of evolution equations generated by the
theorem is the usual KdV hierarchy:
\begin{eqnarray}
K_0 & =& \d_x u, \\
K_1 & =& \d_x^3 u + u \d_x u,\\
K_2 & = & \d^5_x u + \frac{5}{3} u \d^3_x u + \frac{10}{3} \d_x u
\d_x^2 u + \frac{5}{6}u^2 \d_x u,\\
K_3 & = & ...
\end{eqnarray}
associated to an infinite amount of constants of motion:
\begin{eqnarray}
\cH_0 & =& \int dx\, \half u^2,\\
\cH_1 & =& \int dx \left(-\half (\d_x u)^2 + \frac{1}{6} u^3
\right),\\
\cH_2 & =& \int dx \left( \half (\d_x^2 u)^2 + \frac{5}{72} u^4 +
  \frac{5}{16} u^2 \d^2_xu\right),\\
\cH_3 & =&  ...
\end{eqnarray}

The operator $\cR^A_B$ defined in the theorem is a recursion operator
for the system in the sense that if $\delta_Q z^A = Q^A[z]$ is a
symmetry of the evolution equation (\ref{eq:INTEGevoleq}) and $Q^A$ is
in the image of $\cD$, then 
\begin{equation}
\delta_{\cR Q}z^A = \cR^A_B Q^B[z]
\end{equation}
is also a symmetry of the system. A repeated use of this operator
allows the creation of an infinite tower of symmetries from any given
symmetry. Let us point out that they are not necessarily
Hamiltonian even if the starting one is. In the case of Korteweg-de
Vries, this operator is known as the Lenard recursion
operator.

\section{Electromagnetism}
\label{sec:INTEGelectromagnetism}

As we saw in chapter \ref{ch:Elecmag}, the
reduced phase space action of electromagnetism can be written as 
\begin{gather}
  S^R[A^{Ta}_i]= \int dt \Big[\int d^3x\ \half \epsilon_{ab} (\cO A^{Ta})^i
 \d_0  A^{Tb}_i-  H_1\Big],\\
H_1=\half \int d^3x\ (\cO A^{Ta})_i(\cO A^T_{a})^i=-\half \int d^3x\
A^{Tai}\Delta A^T_{ai},
\end{gather}
where in the second expression for the Hamiltonian, we have used that
$\cO$ is ``self-adjoint''. The standard
Poisson bracket determined by the kinetic term is
\begin{equation}
  \label{eq:3}
  \{A^{Ta}_i(x),A^{Tbj}(y)\}_1=\epsilon^{ab}\Delta^{-1}
\epsilon^{jkl}\d^y_k\delta^{T(3)}_{il}(x-y)=\epsilon^{ab}\Delta^{-1}\big(
  \cO^y \delta^{T(3)}_{i}(x-y)\big)^j,
\end{equation}
where $\delta^{T(3)}_{ij}(x-y)$ is the transverse delta function, see
e.g.~\cite{Cohen-Tannoudji:1989dq} section $A_1.2$. In vacuum, Maxwell's
equations for the physical degrees of freedom read
\begin{equation}
  \label{eq:INTEG4}
  \d_0 A^{Tai}(x)=\{A^{Tai}(x),H_1\}_1=-\epsilon^{ab}(\cO A^T_{b})^i(x),
\end{equation}
while the generator for duality rotations is 
\begin{equation}
  \label{eq:INTEG5}
  H_0=-\half \int d^3x\  A^{ai}_T(\cO A_{a}^T)_i, \qquad
  \{H_0,H_1\}_1=0.  
\end{equation}

When presented in this way, the second Hamiltonian structure is
obvious and a lot simpler than the one induced from the covariant
action principle. Indeed, a natural Poisson bracket on reduced
phase space is simply
\begin{equation}
  \{A^{Ta}_i(x),A^{Tb}_j(y)\}_0=\epsilon^{ab}\delta^{T(3)}_{ij}(x-y),
\end{equation}
in terms of which the duality generator is the Hamiltonian for
Maxwell's equations,
\begin{equation}
  \{A^{Tai}(x),H_1\}_1=\{A^{Tai}(x),H_0\}_0.
\end{equation}
This is the main result of this chapter.

At this stage, one can pause and ask whether electromagnetism and its
quantization should not be based on this new Hamiltonian structure. A good
reason to favor the old, more complicated structure is that, by
construction, the Poincar\'e and conformal symmetries admit canonical
generators for the old structure, while not all of them do for the new
one. We plan to return to this question in detail elsewhere. 

The rest of the analysis is standard.  Following the previous section, the
associated recursion operator is defined by
\begin{equation}
  {\cR^{ai}}_{bj}=-\delta^a_b(\cO)^i_j.
\end{equation}
Consider, for $p\geq 1$, $K^{Tai}_{p}=(-)^p\epsilon^{ab}(\cO^{p}
A_b^T)^i$, or equivalently,
\begin{equation}
  K^{Tai}_{2n+1}(x)=(-)^{n+1}\epsilon^{ab}\Delta^n(\cO A_b^T)^i(x),\quad
  K^{Tai}_{2n+2}(x)=(-)^{n+1}\epsilon^{ab}\Delta^{n+1}A_b^{Ti}(x),
\end{equation}
for $n\geq 0$. 
The evolution equations of the hierarchy 
\begin{equation}
  \d_0 A^{Tai}(x)=K^{Tai}_p(x),\quad \forall p\geq 1,
\end{equation}
are also bi-Hamiltonian, 
\begin{equation}
  K^{Tai}_p(x)=\{A^{Tai}(x),H_p\}_1=\{A^{Tai}(x),H_{p-1}\}_0,
\end{equation}
where $H_{p-1}=\frac{(-)^p}{2}\int\ d^3x A^{Ta}_i(\cO^p A^T)_i$, 
\begin{equation}
  H_{2n}=\frac{(-)^{n+1}}{2} \int d^3x\  A^{Tai}\Delta^n (\cO
  A_{a}^T)_i,\
  H_{2n+1}=\frac{(-)^{n+1}}{2} \int d^3x\ A^{Ta}_i\Delta^{n+1} A^{T}_{ai}.
\end{equation}
with Hamiltonians that are in involution, 
\begin{equation}
  \{H_n,H_m\}_1=0=\{H_n,H_m\}_0,\quad \forall n,m\geq 0. 
\end{equation}

\section{Linearized gravity}
\label{sec:linearized-gravity}

Following the result presented in chapter \ref{chap:spin2}, 
the reduced phase space action of linearized gravity can be written as 
\begin{gather}
  \label{eq:19}
S^R[H^{TT a}_{mn}]=  \int dt\Big[ \int d^3x\  \epsilon_{ab} \Delta\left( \cO
  H^{TTa}\right)^{mn} 
\d_0 H^{TTb}_{mn}-H_1\Big],\\
H_1= \int d^3x\ \Big( H^{TTamn} \Delta^2  H^{TT}_{amn}).
\end{gather}
The standard Poisson bracket determined by the kinetic term is 
\begin{equation}
  \{H^{TT a}_{mn}(x),H^{TT b
    kl}(y)\}_1=\half\epsilon^{ab}\Delta^{-2}\left(\cO^y \delta^{(3)
      TT}_{mn}(x-y)\right)^{kl},
\end{equation}
where $\delta^{(3) TTkl}_{mn}(x-y)$ denotes the projector on the
transverse-traceless part of a symmetric rank two tensor. The duality
generator is 
\begin{equation}
  D=- \int d^3x\ H^{TT amn}\Delta\left(\cO H^{TT}_a\right)_{mn},\quad
  \{D,H_1\}_1=0. 
\end{equation}

The analogy with the spin $1$ case can be made perfect by the change
of variables,
\begin{equation}
  H^{1TT}_{mn}=\frac{1}{\sqrt{2}} \Delta^{-1}\left(\cO
    A^{2TT}\right)_{mn},\qquad H^{2TT}_{mn}=\frac{1}{\sqrt{2}} 
\Delta^{-1}\left(\cO  A^{1TT}\right)_{mn}. 
\end{equation}
in terms of which 
\begin{gather}
  H_1= \half \int d^3x\ \Big( \left(\cO A^{TTa}\right)^{mn}  \left(
   \cO A^{TT}_a\right)_{mn}\Big)=-\half \int d^3x\ \Big(
 A^{TTamn}\Delta  A^{TT}_{amn}\Big),\\
\{A^{TT a}_{mn}(x),A^{TT b
    kl}(y)\}_1=\epsilon^{ab}\Delta^{-1}\left(\cO^y \delta^{(3)
      TT}_{mn}(x-y)\right)^{kl},\\
H_0=-D= -\half \int d^3x\ A^{TT amn}\left(\cO A^{TT}_a\right)_{mn}. 
\end{gather}

All formulae of section~\bref{sec:INTEGelectromagnetism} below
equation~\eqref{eq:INTEG5} now generalize in a straightforward way to
massless spin $2$ fields by replacing $T$ (transverse) by $TT$
(transverse-traceless) and contracting over the additional spatial
index.

\section{Massless higher spin gauge fields}
\label{sec:higher-spin-gauge}

The extension of these results to massless higher spin gauge fields
\cite{Fronsdal:1978oq} (see also \cite{Wit:1980fv}) follows directly
from the observation that the Hamiltonian reduced phase space
formulation of these theories merely involves additional spatial
indices \cite{Deser:2005ve}, so that all above results generalize in a
straightforward way.

This can be seen for instance by starting from the approach inspired
from string field theory, where the Lagrangian action for massless
higher spin gauge fields is written as the mean value of the BRST
charge for a suitable first quantized particle model
\cite{Ouvry:1986kl,Bengtsson:1986tg,Henneaux:1987hc} (see also
\cite{Sagnotti:2004ij,Barnich:2005bs} for further developments). In
this framework the reduction of the action to the light-cone gauge
corresponds to the elimination of BRST quartets composed of ghost and
light-cone oscillators
(see~e.g.\cite{Siegel:1999dz,Barnich:2005fu}). In exactly the same
way, the ghost, temporal and longitudinal oscillators form quartets
that can be eliminated to yield the Lagrangian gauge fixed action for
a massless field of spin $s\geq 1$,
\begin{equation}
  S_L[\phi^{TT}_{i_1\dots i_s}]=-\half \int d^4x\ \d_\mu
\phi^{TT}_{i_1\dots i_s}\d^\mu \phi^{TTi_1\dots i_s}
\end{equation}
where the field $\phi^{TT}_{i_1\dots i_s}$ is real, completely symmetric,
traceless and transverse,
\begin{equation}
  \phi^{TT}_{i_1\dots i_s}=\phi^{TT}_{(i_1\dots i_s)},\  \phi^{TT i}_{i
    i_3\dots i_s}=0,\ \d^i  \phi^{TT}_{i i_2\dots i_s}=0. 
\end{equation}
The Hamiltonian formulation is direct, the momenta being
$\pi^{TT}_{i_1\dots i_s}=\d_0 \phi^{TT}_{i_1\dots i_s}$. 

Consider then the Fock space defined by $[a^i,
a^{\dagger}_j]=\delta^i_j$, $a_i|0\rangle=0$, the number operator
$N=a^\dagger_i a^i$, and the ``string field''
$\phi^{TT}_s(x)=\frac{1}{\sqrt{s!}}a^\dagger_{i_1}\dots
  a^\dagger_{i_s}|0\rangle \phi^{TT}_{i_1\dots i_s}(x)$ and the inner product
\begin{equation}
    \label{eq:29}
    \langle \phi^{TT}_s,\psi^{TT}_s\rangle=\int d^4x\ \langle
    \phi^{TT}_s,\psi^{TT}_s\rangle_F=\int d^4x\,   \phi^{TT}_{i_1\dots
      i_s} \psi^{TT i_1\dots i_s}. 
\end{equation}
With this inner product, the generalized curl \cite{Deser:2005ve} 
\begin{equation}
\cO =\frac{1}{N}\epsilon^{ijk}
a^\dagger_i\d_j a_k
\end{equation}
is again self-adjoint. Furthermore, it squares to $-\Delta$
inside the inner product involving transverse-traceless fields, 
\begin{multline}
  \cO^2=\frac{1}{N^2}\big[-\Delta N^2+(\partial\cdot
  a^\dagger)(\partial\cdot a) +(a^\dagger\cdot a^\dagger)\Delta(a\cdot
  a)+ 2(\partial\cdot a^\dagger) N (\partial\cdot a)\\-(\partial \cdot
  a^\dagger)^2 (a\cdot a) -(a^\dagger\cdot a^\dagger)(\partial\cdot
  a)\big],
\Longrightarrow  \langle \phi^{TT}_s,\cO^2\psi^{TT}_s\rangle=-
 \langle \phi^{TT}_s,\Delta\psi^{TT}_s\rangle.
\end{multline}

The change of variables making duality invariance transparent is
\begin{equation}
\phi^{TT}_{i_1\dots i_s}=A^{TT 1}_{i_1\dots i_s},\ \pi^{TT}_{i_1\dots
  i_s}=\left(\cO A^{TT2}\right)_{i_1\dots i_s}. 
\end{equation}
The first order reduced phase space variational principle becomes
\begin{gather}
  S^R[A^{TT a}]=\int dt\Big[\int d^3x\ \half \epsilon_{ab} \left(\cO
    A^{TTa}\right)^{i_1\dots i_s} \d_0A^{TT b}_{i_1\dots
    i_s}-H_1\Big],\\
H_1=-\half \int d^3x\ A^{TT a}_{i_1\dots i_s}\Delta A^{TTi_1\dots i_s}_a. 
\end{gather}
Again, all formulae of section~\bref{sec:INTEGelectromagnetism} below
equation~\eqref{eq:INTEG4}, including the one for the duality generator,
suitably generalize by contracting over the additional spatial
indices.

\section{Conclusion}

We have shown in this chapter that the reduced phase space formulation
of massless higher spin gauge fields is bi-Hamiltonian. The second
Poisson bracket on reduced phase space turned out to be more natural
than the one induced from the covariant variational principle, while
the generator for duality rotations played the role of the second
Hamiltonian. This result trivially generalizes to
Yang-Mills theory with an invariant, non degenerate metric, linearized
around a zero potential by decorating the expressions obtained in the
electromagnetic case with an additional Lie algebra index.

Several generalizations and extensions are suggested by this result. A
first exercise consists in studying the consequences 
for symmetries and conservation laws of both the Maxwell and the
higher spin equations and compare them to known results (see
e.g.~\cite{Fushchich:1992fk,Anco:2003uq,Pohjanpelto:2008kx} and references
therein). Another obvious question is to investigate more general
backgrounds. For instance, the generalization to massless spins
propagating on (anti-) de Sitter spaces instead of Minkowski spacetime
and the inclusion of fermionic gauge fields should be
straightforward. In Yang-Mills theories (anti) self-dual backgrounds could
be promising in view of their close connection to integrable systems.

The most important problem is however the inclusion of
interactions. When comparing to the Korteweg-de Vries equation for
instance, the present work corresponds to the bi-Hamiltonian structure
for the linearized equation. The question is then to find interactions
that preserve this structure.

\chapter{Gravitational features of the ${\rm \bf AdS_3/CFT_2}$ correspondence}
\label{chap:ads3}

In this chapter, we will work
with asymptotically $AdS_3$ space-times in Fefferman-Graham form. In
that case, the gauge is completely fixed as all subleading orders of
the asymptotic Killing vectors are uniquely determined. Equipped 
with a new modified Dirac-type Lie bracket taking into
account their metric dependence due to
the gauge fixing, those asymptotic Killing vectors form a
representation of the local conformal algebra to all orders in $r$. We
give a short description of the possible central extensions of
the 2 dimensions conformal algebra.

We then solve the Einstein's equations for metrics in the Fefferman-Graham
form and compute how the local conformal algebra in 2 dimensions is
realized on solution space. The last step is the covariant computation of the
surface charge algebra and the value of the central extension as
reviewed in appendix \ref{sec:covariant-approach}.

This chapter is based on results originally derived in \cite{Brown:1986zr} and developed
further in \cite{Fefferman:1985uq,Graham:1991kx,Henningson:2000vn,Banados:1998ys,Skenderis:2000zr,Graham:1999ly,Imbimbo:2000ve,Rooman:2000qf,Bautier:2000bh,Papadimitriou:2005dq}. The only original result is the introduction of the
modified Lie Bracket in section \ref{sec:asympt-symma}.

\section{Asymptotically ${\rm \bf AdS_3}$ spacetimes in
  Fefferman-Graham form}
\label{sec:poinc-gener-pauli}

The Fefferman-Graham form for the line element of a $3$
dimensional asymptotically anti-de Sitter spacetime is
\begin{equation}
ds^2 = \frac{l^2}{r^2} dr^2 + g_{AB}(r,x^C)\, dx^Adx^B,\label{eq:GF}
\end{equation}
with $g_{AB} = r^2\bar\gamma_{AB}(x^C)+ O(1)$, where $\bar\gamma_{AB}$
is a conformally flat $2$-dimensional metric.  For explicit
computations we will sometimes choose the parametrization
$\bar\gamma_{AB}=e^{2\varphi} \eta_{AB}$ with $\varphi(x^C)$ and
$\eta_{AB}$ the flat metric on the cylinder, $\eta_{AB}
dx^Adx^B=-d\tau^2+d\phi^2$, $\tau=\frac{1}{l}t$.

\section{Asymptotic symmetries}
\label{sec:asympt-symma}

The transformations leaving this form of the metric invariant are
generated by vector fields satisfying
\begin{equation}
\cL_\xi g_{rr}=0=\cL_\xi g_{rA},\label{eq:ads31}\quad
\cL_\xi g_{AB}=O(1),
\end{equation}
which implies
\begin{eqnarray}
\left\{\begin{array}{l}\xi^r =-\half\psi r, \\ \xi^A  =  Y^A +
    I^A,\quad 
 I^A=-\frac{l^2}{2}\d_B
  \psi\int_r^\infty \frac{dr^\prime}{r^\prime}
  g^{AB}=-\frac{l^2}{4r^2}
\bar\gamma^{AB}\d_B
  \psi+O(r^{-4}),
\end{array}\right.\label{eq:FGvect}
\end{eqnarray}
where $Y^A$ is a conformal Killing vector of $\bar\gamma_{AB}$, and
thus of $\eta_{AB}$, while 
$\psi=\bar D_A Y^A$ is the conformal factor. 

Indeed, the inverse to metric \eqref{eq:GF} is
\begin{equation*}
g^{\mu\nu} = \left(\begin{array}{cc}
 \frac{r^2}{l^2} & 0 \\
0 & g^{AB}
\end{array}\right)
\end{equation*}
where $g^{AB}g_{BC}=\delta^A_C$. From $\cL_\xi g_{rr}=0$, we find
$\xi^r= A r$ for some $A(x^C)$. From $\cL_\xi g_{rA} = 0$ we find
$\partial_r \xi^A=-g^{AB}\frac{l^2}{r}\partial_B A$ so that $\xi^A =
Y^A + I^A$ for some $Y^A(x^C)$ and where $I^A = {l^2}\partial_B A
\,\int_r^\infty dr' \, g^{AB}r'^{-1}$. Finally, $\cL_\xi
g_{AB}=O(1)$ requires $Y^A$ to be a conformal Killing
vector of $\bar\gamma_{AB}$ and $A=-\half \psi$.

Let $\hat Y^A=[Y_1,Y_2]^A$, $\hat
\psi=\bar D_A \hat Y^A$, denote by
$\delta^g_{\xi_1}\xi^\mu_2$ the change induced in $\xi^\mu_2(g)$ due
to the variation $\delta^g_{\xi_1}g_{\mu\nu}=\cL_{\xi_1}g_{\mu\nu}$
and define
\begin{equation}
  \label{eq:44a}
 [\xi_1,\xi_2]^\mu_M =[\xi_1,\xi_2]^\mu-\delta^g_{\xi_1}\xi^\mu_2+
\delta^g_{\xi_2}\xi^\mu_1. 
\end{equation}
For vectors $\xi_1,\xi_2$ given in \eqref{eq:FGvect}, we have
\begin{equation*}
  [\xi_1,\xi_2]^r_M=-\half \hat\psi r,\quad [\xi_1,\xi_2]^A_M=\hat
  Y^A+\hat I^A,
\end{equation*}
where $\hat I^A$ denotes $I^A$ with $\psi$ replaced by $\hat
\psi$. 

Indeed, for the $r$ component, we have $\delta^g_{\xi_1}\xi^r_2=0$ and
the result follows by direct computation of the Lie
bracket. Similarly, $\lim_{r\to\infty}[\xi_1,\xi_2]^A_M=\hat Y^A$.
Finally, using $\d_r\xi^r=\frac{1}{r}\xi^r$ and
$\d_r\xi^A=-\frac{l^2}{r^2}\d_B\xi^r g^{BA}$ a straightforward
computation shows that
$\d_r([\xi_1,\xi_2]^A_M)=\d_B([\xi_1,\xi_2]^r_M) g^{BA}$, which gives
the result. It thus follows that on an asymptotically anti-de Sitter
spacetime in the sense of Fefferman-Graham (solving or not Einstein's
equations with cosmological constant):

{\em The spacetime vectors \eqref{eq:FGvect}
  equipped with the bracket $[\cdot,\cdot]_M$ form a faithful
  representation of the conformal algebra.}

By conformal algebra, we
mean here the direct sum of 2 copies of the Witt algebra.
Furthermore, since $\delta^g_{\xi_1}\xi^r_2=0$,
$\delta^g_{\xi_1}\xi^A_2=O(r^{-4})$, it follows that these vectors
form a representation of the conformal algebra only up to terms of
order $O(r^{-4})$ when equipped with the standard Lie bracket.

\vspace{5mm}

More generally, one can also consider the transformations that leave
the Fefferman-Graham form of the metric invariant up to a Weyl
rescaling of the boundary metric $\bar\gamma_{AB}$. They are generated
by spacetime vectors such that
\begin{equation}
\cL_\xi g_{rr}=0=\cL_\xi g_{rA},\label{eq:ads33}\quad 
\cL_\xi g_{AB}=2\omega g_{AB}+O(1). 
\end{equation}
It is then straightforward to see that the general solution is given
by the vectors \eqref{eq:FGvect}, where $\psi$ is replaced by
$\tilde\psi=\psi-2\omega$. When equipped with the modified Lie bracket
$[\cdot,\cdot]_M$ these vectors now form a faithful representation of
the extension of the two dimensional conformal algebra defined by
elements $(Y,\omega)$ and the Lie bracket $(\hat
Y,\hat\omega)=[(Y_1,\omega_1),(Y_2,\omega_2)]$,
\begin{equation}
  \hat Y^A=Y^B_1\d_B Y_2^A-Y^B_2\d_B Y_1^A,\quad \hat\omega=0.
\end{equation}
with $\omega(x^C)$ arbitrary and $Y^A$ conformal Killing vectors of
$\bar\gamma_{AB}$ and thus also of $\eta_{AB}$. The asymptotic
symmetry algebra is then the direct sum of the abelian ideal of
elements of the form $(0,\omega)$ and of 2 copies of the Witt algebra.

Indeed, we have $\lim_{r\to\infty}(\frac{1}{r}[\xi_1,\xi_2]^r_M)=
-\half Y^A_1\d_A\tilde\psi_2+\d_C\omega_1 Y^C_2+(1\leftrightarrow
2)=-\half \hat \psi$ and $\d_r(\frac{1}{r}[\xi_1,\xi_2]^r_M)=0$, while
the proof for the $A$-component is unchanged. 

\section{Conformal algebra and central extension}
\label{sec:adsalgebra}

In terms of light-cone coordinates, 
$x^\pm=\tau\pm\phi$,
$2\d_\pm=\dover{}{\tau}\pm\dover{}{\phi}$, we have 
$\bar\gamma_{AB}dx^Adx^B=-e^{2\varphi}dx^+dx^-$, and if,
\begin{eqnarray}
Y^\pm(x^\pm)\d_\pm=\sum_{n\in \mathbf Z} c^n_{\pm} l^\pm_n,\quad
l^\pm_n=e^{i n x^\pm}\d_{\pm} , 
\end{eqnarray}
the algebra in terms of the basis vectors $l^\pm_n$ reads 
\begin{equation}
  \label{eq:ads3algebraabstract}
i[l^\pm_m,l^\pm_n]=(m-n)l^\pm_{m},\quad  i[l^\pm_m,l^\mp_n]=0.
\end{equation}
The definition of the basic vectors is different from
\cite{Barnich:2010fk}. This new definition is better suited to the
periodicity of $\phi$ and will allow us to do the flat limit
($l\rightarrow \infty$)  in the next chapter.

Up to equivalence, the most general central extension of
the conformal algebra in 2 dimensions is given by
\begin{eqnarray} 
\left\{\begin{array}{l}
i[l^\pm_m,l^\pm_n]=(m-n)l^\pm_{m}+\frac{c^\pm}{12}m(m+1)(m-1)\delta^0_{m+n}, \cr i[l^\pm_m,l^\mp_n]=0\,.
\label{eq:ads3algabstract}
\end{array}\right.
\end{eqnarray}
The proof follows from doing twice the one for the Witt algebra
$\mathfrak w$, see
e.g~\cite{D.Fuks:1986uq,L.Brink:1988kx,Azcarraga:1995vn}.

\section{Solution space}
\label{sec:solution-spacea}

Let us now start with an arbitrary metric of the form \eqref{eq:GF},
without any assumptions on the behavior in $r$ and let $k^A_B=\half
g^{AC}g_{CB,r}$. One can then define $K^A_B$ through the relation
$k^A_B=\frac{1}{r}\delta^A_B +\frac{1}{r^3} K^A_B$.
We have
\begin{equation*}
\begin{gathered}
\Gamma^r_{rr}   =   -\frac{1}{r},\quad 
\Gamma^r_{rA}  =  0,\quad \Gamma^A_{rr}  =  0,\\
\Gamma^r_{AB}  = -\frac{r^2}{l^2} k_{AB},\quad
\Gamma^A_{rB} =k^A_B,\quad 
\Gamma^A_{BC}  = {}^{(2)}\Gamma^A_{BC}, 
\end{gathered}
\end{equation*}
where $ {}^{(2)}\Gamma^A_{BC}$ denotes the Christoffel symbol
associated to the $2$-dimensional metric $g_{AB}$, which is used to
lower indices on $k^A_B$. If ${K^{TA}}_B$ denotes the traceless part of
$K^{A}_B$, the equations of motion are organized as follows
\begin{gather}
  \label{eq:80a}
  g^{AB}G_{AB}-\frac{2}{l^2} =0\Longleftrightarrow
  \d_r K=-r^{-3}(\half K^2+{K^{T}}^A_B{K^{T}}^B_A), \\\label{eq:80b}
G_{AB}-\half g_{AB}g^{CD}G_{CD}=0\Longleftrightarrow
\d_r {K^{T}}^A_B=-r^{-3} K {K^{T}}^A_B, \\\label{eq:80c}
G_{rA}\equiv r^{-3}( {}^{(2)}D_B K^B_A-\d_A K)=0,\\\label{eq:80d}
G_{rr}-\frac{1}{l^2}g_{rr}\equiv \half\big[r^{-6} (\half K^2
-{K^{T}}^A_B{K^{T}}^B_A )+2 r^{-4} K -\frac{l^2}{r^2}\, {}^{(2)}R\big]=0.
\end{gather}
Combining the Bianchi identities $2(\sqrt{-g}
G^\beta_\alpha)_{,\beta}+\sqrt{-g}
G_{\beta\gamma}{g^{\beta\gamma}}_{,\alpha}\equiv 0$ with the covariant
constancy of the metric, we get the identities 
\begin{multline}
  \label{eq:81}
  2(\frac{r}{l}\sqrt{|{}^{(2)}g|}G_{rA})_{,r}+2(\frac{l}{r}\sqrt{|{}^{(2)}g|}g^{BC}
  [G_{CA}-\frac{1}{l^2} g_{CA}])_{,B}
  +\\+\frac{l}{r}\sqrt{|{}^{(2)}g|}(G_{BC}-\frac{1}{l^2}
  g_{BC}){g^{BC}}_{,A}\equiv 0,
\end{multline}
\begin{multline}
  \label{eq:82}
\big(\frac{r}{l}\sqrt{|{}^{(2)}g|}
[G_{rr}-\frac{1}{r^2}]\big)_{,r} 
+(\frac{l}{r}\sqrt{|{}^{(2)}g|}g^{BA}
G_{Ar})_{,B}+\\+ \frac{1}{l}\sqrt{|{}^{(2)}g|}
[G_{rr}-\frac{1}{r^2}]-\frac{l}{r}\sqrt{|{}^{(2)}g|}(G_{AB}-\frac{1}{l^2}
g_{AB} )k^{AB}\equiv 0.
\end{multline}

To solve the equations of motion, we first contract \eqref{eq:80b}
with ${K^{T}}^B_A$, which gives
\begin{equation*}
    \d_r ({K^{T}}^A_B{K^{T}}^B_A)=-2r^{-3} K {K^{T}}^A_B{K^{T}}^B_A.
\end{equation*}
If we assume ${K^{T}}^A_B{K^{T}}^B_A=\half \cK^2$, we can take the sum and
difference with \eqref{eq:80a} to get 
\begin{equation*}
  \d_r(K+\cK)=-\frac{1}{2}r^{-3}(K+\cK)^2,\quad 
\d_r(K-\cK)=-\frac{1}{2}r^{-3}(K-\cK)^2, 
\end{equation*}
which can be solved in terms of $2$ integration ``constants'' $C(x^B),D(x^B)$
\begin{equation*}
  K=-\frac{1}{C+\half r^{-2}}-\frac{1}{D+\half r^{-2}},\quad 
{K^{T}}^A_B{K^{T}}^B_A=\frac{(D-C)^2}{2(C+\half r^{-2})^2(D+\half r^{-2})^2}.
\end{equation*}
When used in \eqref{eq:80b}, we find 
\begin{equation*}
  {K^{T}}^A_B={A^{T}}^A_B(\frac{1}{C+\half r^{-2}}-\frac{1}{D+\half r^{-2}}),\quad
  {A^{T}}^A_B{A^{T}}^B_A=\half,
\end{equation*}
and can now reconstruct the metric from the equation $\d_r g_{AB}=2
g_{AC} k^C_B$. Defining  $\Theta=\frac{1}{D}+\frac{1}{C}$,
$\Omega=\frac{1}{D}-\frac{1}{C}$, we get 
\begin{equation}
  \label{eq:88}
  g_{AB}=
r^2\bar\gamma_{AB}\big[1+\frac{1}{2r^2}
  \Theta
  +\frac{1}{16r^4}(\Theta^2+\Omega^2)\big]
  +A^T_{AB}\big[\Omega
  +\frac{1}{4r^2}\Theta\Omega\big],
\end{equation}
where $\bar\gamma_{AB}$ are additional integration constants,
restricted by the condition that $\bar\gamma_{AB}$ is symmetric, of
signature $-1$. The index on $A^{TA}_B$ is lowered with
$\bar\gamma_{AB}$, with $A^T_{AB}$ requested to be symmetric. It
follows that ${A^{T}}^A_B$ contains only $1$ additional independent
integration constant.  Writing
$g_{AB}=r^2\bar\gamma_{AB}+\gamma_{AB}$, with
$\gamma_{AB}=\hat\gamma_{AB}+ o(r^0)$, we have $K^A_B=-\hat
\gamma^A_B+o(r^0)$ where the index on $\hat\gamma^A_B$ has been lifted
with $\bar\gamma^{AB}$, the inverse of $\bar\gamma_{AB}$.

When \eqref{eq:80a} and \eqref{eq:80b} are satisfied, the Bianchi
identity \eqref{eq:81} implies that
$r\sqrt{|{}^{(2)}g|}G_{rA}$ does not depend on $r$. The equation of
motion \eqref{eq:80c} then reduces to the condition 
\begin{equation}
  \label{eq:87}
  \bar D_B \hat \gamma^B_A-\d_A \hat \gamma=0,
\end{equation}
where $\hat \gamma=\hat\gamma^A_A$.  When this condition holds in
addition to \eqref{eq:80a} and \eqref{eq:80b}, the remaining Bianchi
identity \eqref{eq:82} implies that $r^2\sqrt{|{}^{(2)}g|}
[G_{rr}-\frac{1}{r^2}]$ does not depend on $r$. The equation of motion
\eqref{eq:80d} then reduces to the condition
\begin{equation}
\hat\gamma=-\frac{l^2}{2}\bar R, 
\end{equation}
and also from the leading contribution to $K$ that
$\Theta=-\frac{l^2}{2}\bar R$.  The constraint \eqref{eq:87} then
becomes
\begin{equation}
\bar D_B \hat \gamma^{TB}_{\ A}=-\frac{l^2}{4}\d_A \bar
R.\label{eq:42}
\end{equation}

To solve this equation, one uses light-cone coordinates,
$x^\pm=\tau\pm\phi$, $2\d_\pm=\dover{}{\tau}\pm\dover{}{\phi}$ and the
explicit parameterization
$\bar\gamma_{AB}dx^Adx^B=-e^{2\varphi}dx^+dx^-$. This gives
\begin{equation}
  \label{eq:41}
  \hat \gamma=-4{l^2}e^{-2\varphi}\d_+\d_-\varphi \iff
\hat \gamma_{+-}=l^2\d_+\d_-\varphi,
\end{equation}
while the general solution to \eqref{eq:42} is 
\begin{equation}
  \label{eq:52}
  \hat\gamma_{\pm\pm}=l^2\big[\Xi_{\pm\pm}(x^\pm)
+\d^2_\pm\varphi-(\d_\pm\varphi)^2\big],
\end{equation}
with $\Xi_{\pm\pm}(x^\pm)$ 2 arbitrary functions of their
arguments. Using \eqref{eq:88},
one then gets
\begin{equation*}
  A^T_{\pm\pm}\Omega=\hat\gamma_{\pm\pm},\quad A^T_{+-}=0,\quad 
  \Omega^2=16e^{-4\varphi}\hat\gamma_{++}\hat\gamma_{--}.
\end{equation*}
In other words, one can choose $\varphi(x^+,x^-),\Xi_{\pm\pm}(x^\pm)$
as coordinates on solution space and, by expressing \eqref{eq:88} in
terms of these coordinates, we have shown that 

{\em The general solution to Einstein's equations with metrics in
  Fefferman-Graham form is given by
\begin{multline}
  g_{AB}dx^Adx^B = \Big(-e^{2\varphi} r^2 +2\hat \gamma_{+-}-
  r^{-2}e^{-2\varphi}(\hat \gamma^2_{+-}+\hat \gamma_{++}\hat
  \gamma_{--} )\Big) dx^+dx^- +\\+ \hat
  \gamma_{++}(1-r^{-2}e^{-2\varphi}\hat\gamma_{+-}) (dx^+)^2 +\hat
  \gamma_{--}(1-r^{-2}e^{-2\varphi}\hat\gamma_{+-}) (dx^-)^2,
\end{multline}
with $\hat \gamma_{AB}$ defined in equations \eqref{eq:41} and
\eqref{eq:52}. }

For instance, in these coordinates, the BTZ
black hole\cite{Banados:1992fk,Banados:1993uq} is determined by
$\varphi=0$ and
\begin{equation}
\Xi_{\pm\pm}=2G(M\pm \frac{J}{l}). 
\end{equation}

\section{Conformal properties of solution space}
\label{sec:conf-transf-solut}

By construction, the finite transformations generated by the spacetime
vectors \eqref{eq:FGvect} leave the Fefferman-Graham form invariant,
and furthermore transform solutions to solutions.

Using light-cone coordinates and the parametrization
$\bar\gamma_{AB}dx^Adx^B=-e^{2\varphi}dx^+dx^-$, we have 
\begin{gather*}
\left\{\begin{array}{l}
  \xi^r=-\half\psi r, 
\quad \psi=\d_+Y^++\d_-Y^-+2\d_+\varphi Y^++2\d_-\varphi Y^-, \\
  \xi^\pm=Y^\pm+\frac{l^2e^{-2\varphi}}{2r^2}\d_\mp\psi+O(r^{-4}),
\end{array}\right.
\end{gather*}
and get 
\begin{equation}
\begin{split}
  \cL_\xi g_{\pm\pm}  &\approx l^2\big[Y^\pm\d_\pm
  \Xi_{\pm\pm}+2\d_\pm Y^\pm \Xi_{\pm\pm}-\half \d^3_\pm
  Y^\pm\big]+O(r^{-2}),  \\
\cL_\xi g_{+-}   &\approx  O(r^{-2}).\label{eq:A12}
\end{split}
\end{equation}
It follows that the local conformal algebra acts on solution
space as
\begin{equation}  -\delta\Xi_{\pm\pm}=Y^\pm\d_\pm
  \Xi_{\pm\pm}+2\d_\pm Y^\pm \Xi_{\pm\pm}-\half \d^3_\pm
  Y^\pm, \label{eq:A15}
\end{equation} 
and with $\delta \varphi=0$. Note that the overall minus sign is
conventional and chosen so that $\delta \Xi_{\pm\pm}\equiv
\delta_Y\Xi_{\pm\pm}$ satisfies 
$[\delta_{Y_1},\delta_{Y_2}]\Xi_{\pm\pm}=\delta_{[Y_1,Y_2]}\Xi_{\pm\pm}$. 

More generally, when considering the extension of the
conformal algebra discussed at the end of section
\bref{sec:asympt-symma}, we find that
\begin{multline*}
  \cL_\xi g_{\pm\pm}  \approx l^2\big[Y^\pm\d_\pm
  \Xi_{\pm\pm}+2\d_\pm Y^\pm \Xi_{\pm\pm}-\half \d^3_\pm
  Y^\pm+\\+\d_\pm^2\omega-2\d_\pm\varphi\d_\pm
  \omega\big]+O(r^{-2}),  
\end{multline*}
\begin{equation*}
\cL_\xi g_{+-}   \approx 2\omega (-\frac{r^2}{2}
e^{2\varphi})+l^2\d_+\d_- \omega+ O(r^{-2}), 
\end{equation*}
and thus, that the extended algebra acts on solution space as in 
 \eqref{eq:A15} with in addition $-\delta\varphi=\omega$.

\section{Centrally extended surface charge algebra}
\label{sec:centr-extend-surf}

Let us take 
\begin{equation}
\varphi=0\label{eq:39}.
\end{equation}
in this section. In fact, starting from a Fefferman-Graham metric
\eqref{eq:GF} with $\bar\gamma_{AB}=e^{2\varphi}\eta_{AB}$ one
can obtain such a metric with vanishing $\varphi(x^C)$ through the
finite coordinate transformation generated by $\xi^r=-\varphi r$ and
$\xi^A=-l^2\d_B\varphi\int^\infty_r\frac{dr^\prime}{r^\prime}
g^{AB}(x,r^\prime)$ since $\cL_\xi g_{rr}=0=\cL_\xi g_{rA}$ and
$\cL_\xi g_{AB}=-2\varphi g_{AB}$.

The background metric is then 
\begin{equation}
d\bar s^2 = -r^2d\tau^2 + \frac{l^2}{r^2} dr^2 + r^2
d\phi^2.
\end{equation}
Furthermore,
\begin{equation*}
Y^+=Y^\tau+Y^\phi,\ Y^-=Y^\tau-Y^\phi,\ \Lambda=\hat
\gamma_{++}+\hat \gamma_{--},\ \Sigma=\hat
\gamma_{++}-\hat \gamma_{--},
\end{equation*}
and
$g_{AB}dx^Adx^B=-r^2d\tau^2+r^2d\phi^2+h_{AB}dx^Adx^B$
with 
\begin{equation}
\begin{gathered}
  h_{\tau\tau}\approx \Lambda(x) +O(r^{-2})\approx h_{\phi\phi},\quad
  h_{\tau\phi}\approx \Sigma(x)+O(r^{-2}),\\
  \d_\tau \Lambda=\d_\phi\Sigma,\quad
  \d_\tau\Sigma=\d_\phi\Lambda.\label{eq:sol}
\end{gathered}
\end{equation}

For the surface charges, we will use the covariant method of
\cite{Barnich:2002fk} whose results are summarized in appendix
\ref{app:surfacecharges}. One can prove the linearity of the charges 
(\ref{eq:applinearity}). We
will then just use equations (\ref{eq:applinearcharges}) and (\ref{eq:gravsurfacecharge})
with $n=3$ and the surface of integration
$\partial\Sigma$ is taken to be the circle at infinity. 
This gives
\begin{multline}
  \cQ_\xi[g-\bar g,\bar g]= \frac{l}{16 \pi G}\lim_{r \to \infty} \int_0^{2\pi}
  rd\phi \Big[\xi^r(\bar D^\tau h-\bar D_\sigma
  h^{\tau\sigma} +\bar D^r h^\tau_r-\bar D^\tau h^r_r)\\-\xi^\tau(\bar D^r h-\bar
  D_\sigma h^{r\sigma}-\bar D^r h^\tau_\tau+\bar D^\tau h^r_\tau)+\xi^\phi(\bar D^r
  h^{\tau}_\phi-\bar D^\tau h^{r}_\phi) +\frac{1}{2}h(\bar
  D^r\xi^\tau-\bar D^\tau\xi^r)\\+\half h^{r\sigma}(\bar D^\tau\xi_\sigma-\bar
  D_\sigma\xi^\tau)-\half  h^{\tau\sigma}(\bar D^r\xi_\sigma-\bar
  D_\sigma\xi^r)\Big].\label{eq:intch1}
\end{multline}
Using
\begin{equation}
\bar D^\tau h-\bar D_\sigma
  h^{\tau\sigma} +\bar D^r h^\tau_r-\bar D^\tau h^r_r
 =r^{-2}\bar \gamma^{AB}(\bar D_A h_{\tau B}-\bar D_\tau h_{AB})=r^{-4}(\d_\phi
 h_{\tau\phi}-\d_\tau h_{\phi\phi}),\nonumber
\end{equation} 
\begin{equation}
  \bar D^r h-\bar D_\sigma h^{r\sigma}-\bar D^r h^\tau_\tau+\bar D^\tau h^r_\tau=
 \frac{1}{l^2}(\d_r h_{\phi\phi}-\frac{1}{r} h_{\tau\tau}) ,\nonumber
\end{equation}
\begin{equation}
  \bar D^r h^\tau_\phi-\bar D^\tau h^r_\phi=
-\frac{1}{l^2}(\d_r h_{\tau \phi}-\frac{1}{r} h_{\tau \phi})\nonumber
\end{equation}
\begin{equation}
  \bar D^r\xi^\tau-\bar D^\tau\xi^r=\frac{2r}{l^2}
  Y^\tau-\frac{1}{r}\d_\tau\psi+O(r^{-3}),\nonumber
\end{equation}
\begin{equation}
\half h^{r\sigma}(\bar D^\tau\xi_\sigma-\bar
  D_\sigma\xi^\tau)-\half  h^{\tau\sigma}(\bar D^r\xi_\sigma-\bar
  D_\sigma\xi^r)=\frac{1}{rl^2}h_{\tau A} Y^A+\frac{1}{4r}
  h_{\tau}^A\d_A\psi+O(r^{-3}),\nonumber
\end{equation}
we find explicitly
\begin{multline}
\cQ_{\xi}[g-\bar g,\bar g]=\frac{1}{16\pi G l} \lim_{r \rightarrow
  \infty}\int_0^{2\pi}d \phi \, (2 Y^\tau h_{\phi\phi}+2
Y^{\phi}h_{\tau\phi})\\\approx 
\frac{1}{8\pi G l}\int_0^{2\pi}d \phi \, (
Y^\tau\Lambda+Y^\phi\Sigma)=
\frac{l}{8\pi G}\int_0^{2\pi}d \phi \, (Y^+\Xi_{++}+Y^-\Xi_{--}).
\end{multline}

The appendix
\ref{app:surfacecharges} suggests that these charges
form a representation of the conformal algebra, or more precisely,
that 
\begin{gather}
  \cQ_{\xi_1}[\cL_{\xi_2}
  g,\bar g] \approx \cQ_{[\xi_1,\xi_2]_M}[g-\bar g,\bar g]+
K_{\xi_1 , \xi_2},\label{bracket1}\\
  K_{\xi_1 , \xi_2} = \cQ_{\xi_1}[\cL_{\xi_2} \bar g,\bar g],\quad \left[
    \xi_1,\xi_2 \right ]_M = [\xi_1,\xi_2]+\delta^g_{\xi_1}\xi_2
  -\delta^g_{\xi_2}\xi_1.\label{bracket2}
\end{gather}
An asymptotic Killing vector of the form \eqref{eq:FGvect} depends on
the metric, $\xi=\xi[x,g]$ and
$\delta^g_{\xi_1}\xi_2=\xi_2[x,\cL_{\xi_1}g]$. {}From
$\delta^g_{\xi_1}\xi_2^\tau=O(r^{-4})$ and
$\delta^g_{\xi_1}\xi_2^\phi=O(r^{-4})$, it follows that only the Lie
bracket $[\xi_1,\xi_2]$ contributes on the right hand side,
$\cQ_{[\xi_1,\xi_2]_M}[g-\bar g,\bar g]=\cQ_{[\xi_1,\xi_2]}[g-\bar
g,\bar g]$. Using \eqref{eq:A12}, \eqref{eq:sol} and
integrations by parts in $\d_\phi$ and the conformal Killing equation
for $Y^A_1,Y^A_2$ to evaluate the left hand side, one indeed finds
\begin{equation}
\begin{gathered}
  \cQ_{\xi_1}[\cL_{\xi_2} g,\bar g]  \approx
  \cQ_{[\xi_1,\xi_2]}[g-\bar g,\bar g] + K_{\xi_1 , \xi_2} ,\\
  K_{\xi_1 , \xi_2} = \frac{l}{8\pi G } \int_0^{2\pi}d \phi \,
  (\d_\phi Y^\tau_1 \partial_\phi^2 Y^\phi_2 -
  \d_\phi Y^\tau_2 \partial_\phi^2 Y^\phi_1),
\end{gathered}
\end{equation}
where $K_{\xi_1 , \xi_2}$ is a form of the well-known Brown-Henneaux
central charge.

In addition, the covariant expression for the surface charges used
above coincides on-shell with those of the Hamiltonian formalism
\cite{Barnich:2002fk,Barnich:2008uq}. In this context, it follows from
the analysis of \cite{Brown:1986fk,Brown:1986zr,Regge:1974kx} that the
surface charge is, after the Fefferman-Graham gauge fixation, the
canonical generator of the conformal transformations in the Dirac
bracket.

\vspace{5mm}

We obtain the numerical value of the central charges $c^\pm$ by
evaluating $K_{\xi_1 , \xi_2}$ on the generators $l^\pm_m$:
\begin{equation}
K_{l^+_m , l^+_n} = i\frac{l}{8G} m^3, \quad K_{l^-_m , l^-_n} = i\frac{l}{8G} m^3, \quad K_{l^+_m , l^-_n} =0,
\end{equation}
which gives $c^+=c^-=\frac{3l}{2G}$. Remark that this extension is not
the one we gave in section \ref{sec:adsalgebra} but an
equivalent one. To transform it into the form (\ref{eq:ads3algabstract}), one has to
redefine $\cQ_{l^\pm_0}$ as $\cQ_{l^\pm_0} + \frac{c^\pm}{24}$. This change
corresponds to a change of background from the BTZ black hole with $M=J=0$ to $AdS_3$.

\chapter{${\rm \bf BMS_3/CFT_1}$ correspondence}
\label{chap:bms3}

In 3 dimensions, the asymptotic symmetry algebra of asymptotically
flat spacetimes at null infinity has been derived in \cite{Ashtekar:1997fk,Ashtekar:1997vn}. This algebra,
known as $\mathfrak{bms}_3$ is the semi-direct sum of the conformal
transformation of the circle with the abelian algebra of the function
on the circle. The algebra of the associated charges has been shown to
provide a centrally extended representation of $\mathfrak{bms}_3$ which has been
related by a contraction, similar to that from $\mathfrak{so}(2,2)$ to
$\mathfrak{iso}(2,1)$, to the centrally extended Poisson bracket
algebra of surface charges of asymptotically anti-de Sitter spacetimes
in $3$ dimensions \cite{Barnich:2007nx}.

In this chapter, we will apply the same technique we used on
asymptotically $AdS_3$ spacetimes to the case of asymptotically flat
spacetimes at null infinity. We start with asymptotically flat metrics
at null infinity in a form suggested from the analysis of
the 4 dimensional case by Sachs. This form is the analog of the
Fefferman-Graham form. In particular, the gauge is completely fixed in
the sense that all subleading orders of the asymptotic Killing vectors
are again completely determined. Equipped with thee modified Lie bracket
those vectors form a representation of the $\mathfrak{bms}_3$ algebra to all
orders in $r$. We then study the possible central extensions of the
$\mathfrak{bms}_3$. 

After that, we solve the flat equations of motion and compute the
representation of the $\mathfrak{bms}_3$ algebra on solution space. The last
section contains the derivation of the surface charge algebra using
the same techniques we used for the $AdS_3$ case.

\section{Asymptotically ${\rm \bf BMS_3}$ spacetimes}
We consider metrics of the form
\begin{equation}
  \label{eq:43d}
   ds^2=e^{2\beta}\frac{V}{r} du^2-2e^{2\beta}
   dudr+r^2e^{2\varphi}(d\phi-Udu)^2,
\end{equation}
or, equivalently, 
\begin{equation*}
g_{\mu\nu}=  \begin{pmatrix} e^{2\beta} V r^{-1}+r^2 e^{2\varphi} U^2 & 
  -e^{2\beta} & -r^2 e^{2\varphi} U \\
  -e^{2\beta} & 0 & 0 \\
-r^2 e^{2\varphi} U & 0 & r^2 e^{2\varphi}
\end{pmatrix}
\end{equation*}
with inverse given by 
\begin{equation*}
g^{\mu\nu}=  \begin{pmatrix} 0 & 
  -e^{-2\beta} & 0 \\
  -e^{-2\beta} & 
-\frac{V}{r} e^{-2\beta} & -U e^{-2\beta} \\
0 & -U e^{-2\beta} &  r^{-2} e^{-2\varphi}
\end{pmatrix}.
\end{equation*}
Here, $\varphi=\varphi(u,\phi)$. Three dimensional Minkowski space is
described by $\varphi=0=\beta=U$ and $V=-r$. The fall-off conditions
are taken as $\beta=O(r^{-1})$, $U=O(r^{-2})$ and $V=-2r^2\d_u
\varphi+O(r)$. In particular, $g_{uu}=-2r\d_u\varphi+O(1)$. 

\section{Asymptotic symmetries}
\label{sec:asympt-symm}

The transformations leaving this form of the metric invariant are
generated by vector fields such that
\begin{gather}
\cL_\xi g_{rr}=0=\cL_\xi g_{r\phi},\quad  \cL_{\xi}
g_{\phi\phi}=0,\label{eq:bms1}\\
\cL_\xi g_{ur}=O(r^{-1}),\quad \cL_\xi g_{u\phi}=O(1), \quad 
\cL_\xi g_{uu}=O(1).\label{eq:bms2}
\end{gather}
Equations \eqref{eq:bms1} imply that 
\begin{eqnarray}
  \left\{\begin{array}{l}\xi^u=f,\\
      \xi^\phi  =  Y+ I,\quad I = -e^{-2\varphi}
      \d_\phi f \,\int_r^\infty dr' \,
      {r'}^{-2} e^{2\beta}=-\frac{1}{r}e^{-2\varphi}\d_\phi f+O(r^{-2}),\\
      \xi^r =- r\big[\d_\phi \xi^\phi-\d_\phi f  U+
\xi^\phi\d_\phi \varphi +f \d_u
      \varphi \big], 
\end{array}\right.\label{eq:bms3vect}
\end{eqnarray}
with $\d_r f=0=\d_r Y$.  The first equation of \eqref{eq:bms2} then
implies that
\begin{equation}
\d_u f =f\d_u\varphi+ Y\d_\phi\varphi +\d_\phi Y\iff 
f=e^{\varphi}\big[T+\int_0^udu^\prime
e^{-\varphi}(\d_\phi Y+Y\d_\phi \varphi) \big],\label{eq:44}
\end{equation}
with $T=T(\phi)$, while the second requires $\d_u Y=0$ and thus
$Y=Y(\phi)$, which implies in turn that the last one is identically
satisfied. 

The Lie algebra $\mathfrak{bms}_3$ is determined by two arbitrary
functions $(Y,T)$ on the circle with bracket
$[(Y_1,T_1),(Y_2,T_2)]=(\hat Y,\hat T)$ determined by $\hat
Y=Y_1\d_\phi Y_2-(1\leftrightarrow 2)$ and $\hat T= Y_1\d_\phi T_2
+T_1\d_\phi Y_2-(1\leftrightarrow 2)$. Let $\cI=S^1\times {\mathbb R}$
with coordinates $u,\phi$ and consider the vector fields $\bar\xi=f
\dover{}{u}+ Y\dover{}{\phi}$ with $f$ as in \eqref{eq:44} and
$Y=Y(\phi)$. By direct computation, it follows that these vector
fields equipped with the commutator bracket provide a faithful
representation of $\mathfrak{bms}_3$. Furthermore : 

{\em The spacetime vectors \eqref{eq:bms3vect}, with $f$ given in 
\eqref{eq:44} and $Y=Y(\phi)$ form a faithful
representation of the $\mathfrak{bms}_3$ Lie algebra on an
asymptotically flat spacetime of the form \eqref{eq:43d} when equipped
with the modified bracket $[\cdot,\cdot]_M$. }

Indeed, for the $u$ component, there is no modification due to the
change in the metric and the result follows by direct computation.  As
a consequence, $\hat f=[\xi_1,\xi_2]^u_{(M)}$ corresponds to $f$ in
\eqref{eq:44} with $T$ replaced by $\hat T$ and $Y$ by $\hat Y$. By
evaluating $\cL_\xi g^{\mu\nu}$, we find
\begin{gather}
\left\{\begin{array}{l}
\delta_\xi\varphi=0,\\
  \delta_\xi\beta =\xi^\alpha\d_\alpha\beta+\half
  \big[\d_u f+\d_r\xi^r+\d_\phi f U],\\
\delta_\xi U =\xi^\alpha \partial_\alpha U + U \big[\d_u f+ \partial_\phi f
U-\d_\phi \xi^\phi\big]
- \partial_u \xi^\phi-\d_r \xi^\phi\frac{V}{r}
+\partial_\phi \xi^r\frac{e^{2(\beta-\varphi)}}{r^2}.
\end{array}
\right.
\end{gather}
It follows that 
\begin{gather}
\left\{\begin{array}{l}
  \delta^g_{\xi_1}\xi^\phi_2=-e^{-2\varphi}\d_\phi
  f_2\int^\infty_r\frac{dr^\prime}{{r^{\prime}}^2}
    e^{2\beta}2\delta_{\xi_1}\beta,\\
\delta^g_{\xi_1}\xi^r_2=-r\big[\d_\phi(\delta^g_{\xi_1}\xi^\phi_2)
+(\delta^g_{\xi_1}\xi^\phi_2)\d_\phi
\varphi-\d_\phi f_2 \delta_{\xi_1} U\big].
\end{array}
\right.
\end{gather}
We also have $\lim_{r\to\infty}[\xi_1,\xi_2]^\phi_M=\hat Y$. Using
$\d_r\xi^\phi=\frac{e^{2(\beta-\varphi)}}{r^2}\d_\phi f$,
\eqref{eq:44} and the expression of $\xi^r$ in \eqref{eq:bms3vect}, it
follows by a straightforward computation that
$\d_r([\xi_1,\xi_2]^\phi_M)=\frac{e^{2(\beta-\varphi)}}{r^2}\d_\phi\hat
f$, which gives the result for the $\phi$ component.  Finally, for the
$r$ component, we need the relation
\[\d_r(\frac{\xi^r}{r})=-\d_r\big(\d_\phi\xi^\phi
+\xi^\phi\d_\phi f\d_\phi\varphi-\d_\phi
fU\big).\] We then have
$\lim_{r\to\infty}\frac{[\xi_1,\xi_2]^r_M}{r}=-\d_\phi\hat Y-\hat Y
\d_\phi\varphi-\hat f\d_u\varphi$, while direct computation shows that
$\d_r(\frac{[\xi_1,\xi_2]^r_M}{r})=-\d_r\big(
\d_\phi([\xi_1,\xi_2]^\phi_M)
-\d_\phi([\xi_1,\xi_2]^u_M)U+[\xi_1,\xi_2]^\phi_M\d_\phi\varphi\big)$, 
which gives the result for the
$r$ component.

More generally, one can also consider the transformations that leave
the form of the metric \eqref{eq:43d} invariant up to a rescaling of
$\varphi$ by $\omega(u,\phi)$. They are generated by spacetime vectors
satisfying
\begin{gather}
\cL_\xi g_{rr}=0=\cL_\xi g_{r\phi},\quad  \cL_{\xi}
g_{\phi\phi}=2\omega g_{\phi\phi},\label{eq:bms1a}\\
\cL_\xi g_{ur}=O(r^{-1}),\quad \cL_\xi g_{u\phi}=O(1), \quad 
\cL_\xi g_{uu}=-2r\d_u\omega +O(1).\label{eq:bms2a}
\end{gather}
Equations \eqref{eq:bms1a}, \eqref{eq:bms2a} then imply that the
vectors are given by \eqref{eq:bms3vect}, \eqref{eq:44} with the
replacement $\d_\phi Y\to \d_\phi Y-\omega$.

With this replacement, the vector fields $\bar\xi=f \dover{}{u}+
Y\dover{}{\phi}$ on $\cI=S^1\times {\mathbb R}$ equipped with the
modified bracket provide a faithful representation of the extension of
$\mathfrak{bms}_3$ defined by elements $(Y,T,\omega)$ and bracket
$[(Y_1,T_1,\omega_1),(Y_2,T_2,\omega_2)]=(\hat Y,\hat T,\hat\omega)$,
with $\hat Y,\hat T$ as before and $\hat \omega=0$.

Indeed, the result is obvious for the $\phi$ component. Furthermore,
\[\delta^g_{\bar \xi_1} f_2= \omega_1 f_2+e^{\varphi}\int_0^udu^\prime
e^{-\varphi}[-\omega_1(\d_\phi Y_2-\omega_2+Y_2\d_\phi\varphi)+
Y_2\d_\phi \omega_1].\] At $u=0$, we get
$[\bar\xi_1,\bar\xi_2]^u_M|_{u=0}=e^\varphi|_{u=0} \hat T$, while
direct computation shows that $\d_u ([\bar\xi_1,\bar\xi_2]^u_M)=\hat
f\d_u\varphi +\hat Y\d_\phi\varphi+\d_\phi\hat Y$, as it
should. 

Following the same reasoning as before, one can then also show that
the spacetime vectors \eqref{eq:bms3vect} with the replacement
discussed above and equipped with the modified Lie bracket provide a
faithful representation of the extended $\mathfrak{bms}_3$ algebra. 

Indeed, we have $\xi=\bar\xi +I\dover{}{\phi}+\xi^r\dover{}{r}$.  Furthermore,
$[\xi_1,\xi_2]_M^u=[\bar\xi_1,\bar\xi_2]^u_M=\hat f$ as it should.  In
the extended case, the variations of $\beta,U$ are still given by
\eqref{eq:10}.  We then have
$\lim_{r\to\infty}[\xi_1,\xi_2]_M^\phi=\hat Y$ and find, after some
computations, $\d_r
([\xi_1,\xi_2]_M^\phi)=\frac{e^{2(\beta-\varphi)}}{r^2} \d_\phi\hat
f$, giving the result for the $\phi$ component. Finally, we have
$\lim_{r\to\infty}\frac{[\xi_1,\xi_2]^r_M}{r}=-\d_\phi\hat Y-\hat Y
\d_\phi\varphi-\hat f\d_u\varphi$, while direct computation shows that
$\d_r(\frac{[\xi_1,\xi_2]^r_M}{r})=-\d_r\big(
\d_\phi([\xi_1,\xi_2]^\phi_M)
-\d_\phi([\xi_1,\xi_2]^u_M)U+[\xi_1,\xi_2]^\phi_M\d_\phi\varphi
\big)$, which gives the result for the $r$ component.

\section{$\mathfrak{bms}_3$ algebra and central extensions}
\label{sec:bms3algebra}

The $\mathfrak{bms}_3$ algebra can also be viewed as the algebra of vector
fields on the circle acting on the functions of the circle and has
been originally derived in the context of a symmetry reduction of
four dimensional gravitational waves
\cite{Ashtekar:1997vn,Ashtekar:1997fk}.

More precisely, let $y=Y\dover{}{\phi} \in {\rm Vect}(S^1)$
be the vector fields on the circle and $T(d\phi)^{-\lambda}\in \cF_{\lambda}(S^1)$ tensor
densities of degree $\lambda$, which form a  module of the Lie
algebra ${\rm Vect} (S^1)$ for the  action 
\begin{equation}
\rho(y) t=
(YT^\prime-\lambda Y^\prime T)d\phi^{-\lambda}\,.
\end{equation}
The algebra $\mathfrak{bms}_3$ is the semi-direct sum of ${\rm Vect}
(S^1)$ with the abelian ideal $\cF_{1}(S^1)$, the bracket between
elements of ${\rm Vect}(S^1)$ and elements
$t=Td\phi^{-1}\in\cF_{1}(S^1)$ being induced by the module action,
$[y,t]=\rho(y)t$.

Consider the associated complexified Lie algebra and let
$z=e^{i\phi}$, $m,n,k ...\in \mathbb Z$. Expanding into modes,  
$y=a^n l_n $, $t=b^n t_n$,  where
\[
l_n=e^{in\phi}\dover{}{\phi}=iz^{n+1}\dover{}{z}
,\quad t_n= e^{in\phi}(d\phi)^{-1}=i z^{n+1}(dz)^{-1}\,,\] 
the commutation relations read explicitly
\begin{equation}
i[l_m,l_n]=(m-n) l_{m+n},\quad i[l_m, t_n]=(m-n) t_{m+n},\quad i[t_m, t_n]=0\,.
\label{eq:bms3abstractdef}
\end{equation}
The non-vanishing structure constants of $\mathfrak{bms}_3$ are thus
entirely determined by the structure constants $[l_m,l_n]=-i f^k_{mn}
l_k$, $f^k_{mn}=\delta^k_{m+n}(m-n)$ of the Witt subalgebra
$\mathfrak w$ defined by the linear span of the $l_n$.

Up to equivalence, the most general central extension of
$\mathfrak{bms}_3$ is given by
\begin{eqnarray} 
\left\{\begin{array}{l}
i[l_m,l_n]=(m-n)
  l_{m+n}+\frac{c_1}{12}m(m+1)(m-1)\delta^0_{m+n}, \cr i[l_m,
t_n]=(m-n) t_{m+n}+\frac{c_2}{12}m(m+1)(m-1)\delta^0_{m+n},\cr i [t_m,
t_n]=0\,.
\label{eq:bms3algabstract}
\end{array}\right.
\end{eqnarray}
The proof follows again by generalizing the one for the Witt algebra
$\mathfrak w$, see
e.g~\cite{D.Fuks:1986uq,L.Brink:1988kx,Azcarraga:1995vn}:
\begin{proof}
In order to get rid of the overall $i$ in (\ref{eq:bms3abstractdef}), we redefine
the generators as $l^\prime_m=il_m$. Inequivalent central extensions
of $\mathfrak{bms}_3$ are classified by the cohomology space
$H^2(\mathfrak{bms}_3)$. More explicitly, the Chevally-Eilenberg
differential is given by
\begin{equation}
  \gamma=-\half C^mC^{k-m}(2m-k)\dover{}{C^k}-C^m\xi^{k-m}(2m-k) 
\dover{}{\xi^k}\,,
\end{equation}
in the space $\Lambda(C,\xi)$ of polynomials in the anticommuting
``ghost'' variables $C^m,\xi^m$. The grading is given by the
eigenvalues of the ghost number operator,
$N_{C,\xi}=C^m\dover{}{C^m}+\xi^m\dover{}{\xi^m}$, the differential
$\gamma$ being homogeneous of degree $1$ and
$H^2(\mathfrak{bms}_3)\simeq H^2(\gamma,\Lambda(C,\xi))$.
Furthermore, when counting only the ghosts $\xi^m$ associated with
supertranslations, $N_\xi=\xi^m\dover{}{\xi^m}$, the differential
$\gamma$ is homogeneous of degree $0$, so that the cohomology
decomposes into components of definite $N_\xi$ degree. The cocycle
condition then becomes
\begin{eqnarray}
\gamma(\omega^0_{m,n}C^mC^n)=0,\quad 
\gamma(\omega^1_{m,n}C^m\xi^n)=0,\quad 
\gamma(\omega^2_{m,n}\xi^m\xi^n)=0, 
\end{eqnarray}
with $\omega^0_{m,n}=-\omega^0_{n,m}$ and
$\omega^2_{m,n}=-\omega^2_{n,m}$. The coboundary condition reads
\begin{eqnarray}
  \omega^0_{m,n}C^mC^n=\gamma(\eta^0_mC^m),\quad
  \omega^1_{m,n}C^m\xi^n=\gamma(\eta^1_m\xi^m). 
\end{eqnarray}

We have $\{\dover{}{C^0},\gamma\}=\cN_{C,\xi}$ with
$\cN_{C,\xi}=m(C^m\dover{}{C^m}+\xi^m\dover{}{\xi^m})$. It follows
that all cocycles of $\cN_{C,\xi}$ degree different from $0$ are
coboundaries, $\gamma \omega_N=0$, $\cN_{C,\xi}\omega_N=N\omega_N$,
$N\neq 0$ implies that
$\omega_N=\gamma(\frac{1}{N}\dover{}{C^0}\omega_N)$. Without loss of
generality we can thus assume that
$\omega^0_{m,n}C^mC^n=\omega^0_mC^mC^{-m}$ with
$\omega^0_m=-\omega^0_{-m}$ and in particular $\omega^0_0=0$;
$\omega^1_{m,n}C^m\xi^n=\omega^1_mC^m\xi^{-m}$;
$\omega^2_{m,n}\xi^m\xi^n=\omega^2_m\xi^m\xi^{-m}$ with
$\omega^2_m=-\omega^2_{-m}$ and in particular $\omega^2_0=0$. By
applying $\dover{}{C^0}$ to the coboundary condition
$\omega^0_mC^mC^{-m}=\gamma(\eta^0_mC^m)$ we find that $0=m\eta^0_m
C^m$. The coboundary condition then gives
$\omega^0_mC^mC^{-m}=\gamma(\eta^0_0 C^0)=-m \eta^0_0 C^m C^{-m}$. By
adjusting $\eta^0_0$, we can thus assume without loss of generality
that $\omega^0_1=0$ and that the coboundary condition has been
entirely used. In the same way
$\omega^1_mC^m\xi^{-m}=\gamma(\eta^1_m\xi^m)$ implies first that
$\eta^1_m=0$ for $m\neq 0$ and then that one can assume that
$\omega^1_1=0$, with no coboundary condition left.

Taking into account the anticommuting nature of the ghosts, the
cocycle conditions become explicitly,
$\omega^0_{m}(2n+m)-\omega^0_{n}(2m+n)+\omega^0_{m+n}(n-m)=0$,
$\omega^1_{m}(2n-m)+\omega^1_{n}(n-2m)+\omega^1_{m-n}(n+m)=0$,
$\omega^2_{m}(2n+m)+\omega^2_{m+n}(n-m)=0$. Putting $m=0$ in the last
relation gives $\omega^2_m=0$, for $m\neq 0$ and thus for all $m$,
putting $m=1=n$ in the second relation gives $\omega^1_0=0$, while
$m=0$ gives $\omega^1_n n=-\omega^1_{-n}n$ and thus that
$\omega^1_n=-\omega^1_{-n}$ for all $n$. Changing $m$ to $-m$ and
using this symmetry property, the cocycle conditions for $\omega^0_m$
and $\omega^1_m$ give the same constraints. Putting $m=1$, one finds
the recurrence relation
$\omega^{0,1}_{n+1}=\frac{n+2}{n-1}\omega^{0,1}_n$, which gives a
unique solution in terms of $\omega^{0,1}_2$. The result follows by
setting $c_{1,2}=\half \omega^{0,1}_2$ and checking that the
constructed solution does indeed satisfy the cocycle condition.

\end{proof}

\vspace{5mm}

The above result can also be obtained by a ``flat'' limit ($l\to
\infty$) of the $AdS_3$ result of the previous chapter. First, we
write the the 2D conformal charge algebra
(\ref{eq:ads3algabstract}) in term of the new
generators $\tilde l_m=l_m^+-l_{-m}^-$, $\tilde t_m=\frac{1}{l}(l_m^++l_{-m}^-)$:
\begin{eqnarray} 
\left\{\begin{array}{l}
i[\tilde l_m,\tilde l_n]=(m-n)
  \tilde l_{m+n}+\frac{c^+-c^-}{12}m(m+1)(m-1)\delta^0_{m+n}, \cr i[\tilde l_m,
\tilde t_n]=(m-n)\tilde  t_{m+n}+\frac{c^++c^-}{12l}m(m+1)(m-1)\delta^0_{m+n},\cr i [\tilde t_m,
\tilde t_n]=\frac{1}{l^2} \left( (m-n) \tilde l_{m+n} + \frac{c^+-c^-}{12} m
    \left( m^2-1\right)\right)\,.
\end{array}\right.
\end{eqnarray}
Second, we take the limit $l \to \infty$: the new generators reduces
to the generators of $\mathfrak{bms}_3$ if $u= l
\tau$ and the algebra goes to (\ref{eq:bms3algabstract}) with $c_1=c^+-c^-$ and $c_2=\frac{c^++c^-}{l}$.

\section{Solution space}
\label{sec:solution-space}

Following \cite{Tamburino:1966ys}, the equations of motion are organized
in terms of the Einstein tensor $G_{\alpha\beta}=R_{\alpha\beta}-\half
g_{\alpha\beta} R$ as
\begin{gather}
  G_{r\alpha}=0,  \qquad
G_{AB}-\half g_{AB} g^{CD}G_{CD}=0,  \label{eq:571b}\\
G_{uu}=0=G_{uA},  \label{eq:571c}\\
g^{CD}G_{CD}=0,  \label{eq:571d} 
\end{gather}
and the Bianchi identities are written as 
\begin{equation}
\label{eq:561}
 0= 2\sqrt{-g}{G_\alpha^\beta}_{;\beta}=2(\sqrt{-g}G_\alpha^\beta)_{,\beta}+
\sqrt{-g}G_{\beta\gamma}{g^{\beta\gamma}}_{,\alpha}. 
\end{equation}

For a metric of the form \eqref{eq:43d}, we have 
\begin{equation*}
\begin{gathered}
\Gamma^\lambda_{rr}=\delta^\lambda_r2\beta_{,r},\quad
\Gamma^{u}_{\lambda r}=0,\quad
\Gamma^r_{\phi r}=\beta_{,\phi }+ n ,\quad \Gamma^\phi _{\phi r}=\frac{1}{r} ,\\
\Gamma^u_{\phi \phi }= e^{-2\beta+2 \varphi} r,\quad
\Gamma^\phi _{\phi \phi}= e^{-2\beta+2 \varphi} U r+\d_\phi\varphi, \\ 
\Gamma^\phi _{ur}=
-\frac{1}{r} U +r^{-2}e^{2\beta- 2 \varphi}(\d_\phi \beta-n ),\quad 
\Gamma^u_{u\phi }=\beta_{,\phi }-n - e^{-2\beta+2 \varphi} rU ,\\
\Gamma^r_{ur}=-\half (\d_r+2\beta_{,r})\frac{V}{r}-(\beta_{,\phi }+n )
U ,\quad 
\Gamma^\phi _{\phi u}=\d_u \varphi   +U (\beta_{,\phi }-n )- e^{-2\beta+2\varphi}rU^2,\\
\Gamma^{u}_{uu}=2\beta_{,u}+\half(\d_r+2\beta_{,r})
\frac{V}{r}+2U  n +e^{-2\beta+2\varphi}r U^2 ,\\
\Gamma^r_{\phi \phi }= e^{-2\beta+2 \varphi}(r^2\d_\phi  U +r^2
\d_\phi \varphi  U 
+r^2\d_u \varphi+V ),\\
\Gamma^r_{u\phi }=-\frac{V_{,\phi }}{2r}-n \frac{V}{r}-
e^{-2\beta+2\varphi}U [r^2\d_\phi  U +r^2 \d_\phi \varphi  U -r^2\d_u \varphi+V ],\\
\Gamma^\phi _{uu}=2U \beta_{,u}+\half
U (\d_r+2\beta_{,r})\frac{V}{r}+2U^2  n +re^{-2\beta+2 \varphi}U^3
-U_{,u}-2\d_u \varphi  U\\-\half
e^{2\beta-2 \varphi}r^{-2}(\d_\phi +2\d_\phi \beta)\frac{V}{r}-U (\d_\phi +\d_\phi \varphi) U,\\
\Gamma^r _{uu}=- \half (\d_u - 2 \d_u \beta)\frac{V}{r} + \half
\frac{V}{r}(\d_r + 2 \d_r \beta) \frac{V}{r} + V r e^{2 \varphi - 2
  \beta} U (\d_r + \frac{1}{r})U + r^2 e^{-2\beta + 2 \varphi} U^2
\d_u \varphi \\
+ \half U (\d_\phi + 2 \d_\phi \beta)\frac{V}{r} 
+ \half r^2 e^{-2 \beta + 2 \varphi} U (\d_\phi + 2 \d_\phi \varphi) U,
\end{gathered}
\end{equation*}
where the notation $n=\half r^2 e^{2\varphi-2\beta} \d_r U$ has been
used.

We start with $G_{rr}=0$. From
\begin{equation*}
G_{rr}=R_{rr}=\frac{2}{r} \d_r \beta,
\end{equation*}
we find $\beta=0$ by taking the fall-off conditions into
account. From
\begin{equation*}
  G_{r\phi}=R_{r\phi}=(\d_r +\frac{1}{r}) n + \frac{1}{r}\d_\phi \beta,
\end{equation*}
we then obtain, by using the previous result, that $n=\frac{N}{r}$
where the integration constant $N=N(u,\phi)$.  Using the definition of
$n$, we get $U=-r^{-2} e^{-2\varphi} N$.  From
$G_{ru}=-g^{\phi\phi}R_{\phi\phi}$ and
\begin{multline*}
R_{\phi\phi}= e^{-2\beta+2\varphi} 
\left( (\d_r - \frac{1}{r})V + 2r \d_u \varphi + 2r(\d_\phi+\d_\phi \varphi) U \right)\\
-2 \d_\phi^2 \beta + 2 \d_\phi \beta \d_\phi \varphi 
- 2 (\d_\phi \beta - n)^2 -2 \d_\phi\varphi n + 2 \d_\phi n\\
  = e^{2\varphi} \left( (\d_r - \frac{1}{r})V + 2r \d_u \varphi\right)- 2 r^{-2 } N^2,
\end{multline*}
we get $\d_r (\frac{V}{r}) = - 2 \d_u \varphi + 2 r^{-3} e^{-2\varphi} N^2$
and then 
\begin{equation*}
V=-2r^2 \d_u \varphi + r M - r^{-1} e^{-2\varphi}N^2,
\end{equation*}
for an additional integration constant $M=M(u,\phi)$. 

When $G_{rr}=G_{r\phi}=G_{ru}=0$, the Bianchi identity \eqref{eq:561}
for $\alpha=r$ implies that $G_{\phi\phi}=0$. This implies in turn
that $R=0$. The Bianchi identity for $\alpha=\phi$ then gives $\d_r (r
G_{u\phi})=0$. When $G_{u\phi}=0$, the Bianchi identity for $\alpha=u$
gives $\d_r (r G_{uu})=0$. To solve the remaining equations of motion,
there thus remain only the constraints
\begin{equation*}
  \lim_{r \rightarrow \infty} r R_{u\phi} =0,\qquad \qquad 
\lim_{r \rightarrow \infty} r R_{uu} =0.
\end{equation*}
to be fulfilled. From 
\begin{equation*}
R_{u\phi} = \frac{1}{r} \left( -(\d_u + \d_u \varphi) N + 
\half \d_\phi M\right) + O(r^{-2}), 
\end{equation*}
we get 
\begin{equation*}
N=e^{-\varphi}\, \Xi(\phi)+ e^{-\varphi} \int^u_{u_0} d\tilde u \, e^\varphi \half \d_\phi M.
\end{equation*}
while 
\begin{equation*}
R_{uu} = \frac{1}{r} \left( -\half(\d_u + 2\d_u \varphi) M + 
e^{-2\varphi}\d_u(\d_\phi^2 \varphi-\half (\d_\phi \varphi)^2)\right) + O(r^{-2})
\end{equation*}
implies
\begin{equation*}
M = e^{-2\varphi}[\Theta (\phi) -(\d_\phi
  \varphi)^2+2\d_\phi^2 \varphi]. 
\end{equation*}
We thus have shown:

{\em For metrics of the form \eqref{eq:43d} with $\lim_{r \rightarrow
    \infty}\beta=0$, the general solution to the equations of motions
  is given by
\begin{equation}
\begin{gathered}
  \label{eq:44b}
  ds^2=s_{uu}du^2-2dudr+2s_{u\phi}dud\phi+r^2e^{2\varphi}d\phi^2,\\
s_{uu}=
e^{-2\varphi}\big[\Theta-(\d_\phi\varphi)^2+2\d^2_\phi\varphi\big]-2r\d_u
\varphi,\\
s_{u\phi}=e^{-\varphi}\Big[\Xi+\int^u_{u_0} d\tilde u
e^{-\varphi}\big[\half\d_\phi\Theta-\d_\phi\varphi[\Theta-
(\d_\phi\varphi)^2+3\d^2_\phi\varphi]+\d^3_\phi
\varphi\big]\Big],
\end{gathered}
\end{equation}
where $\Theta=\Theta(\phi)$ and $\Xi=\Xi(\phi)$ are arbitrary
functions. 
}

\section{Conformal properties of solution space}
\label{sec:conf-transf-solutbis}

By computing $\cL_\xi s_{\mu\nu}$, we find that the asymptotic
symmetry algebra $\mathfrak{bms}_3$ acts on solution space according
to 
\begin{equation}
\begin{split}
  -\delta\, \Theta & = Y\d_\phi \Theta+2 \d_\phi Y \Theta -
  2 \partial^3_\phi Y, \label{eq:B11}\\
  -\delta\, \Xi & = Y \partial_\phi\Xi+ 2 \d_\phi Y \Xi  +\half
  T\d_\phi\Theta+ \d_\phi T \Theta-\d^3_\phi
  T, \\
 -\delta\,\varphi &=0. 
\end{split}
\end{equation}
For the extended algebra, the first two relations are unchanged, while
$-\delta\varphi=\omega$. 

\section{Centrally extended surface charge algebra}
\label{sec:centr-extend-charge}

Let us again take $ \varphi=0$ in this section.  For the surface
charges computed at the circle at infinity $u=cte, r=cte\to\infty$,
one starts again from \eqref{eq:appasymptcharge}. In this case, on can
again prove linearity of the charges (\ref{eq:applinearity}) and simplify the expression
for the charges to (\ref{eq:applinearcharges}). The background line element,
which is used to raise and lower indices, is
\begin{equation}
  \label{eq:20a}
  d\bar s^2=-du^2-2dudr+r^2 d\phi^2, 
\end{equation}
This gives
\begin{multline}
  \cQ_\xi[g-\bar g,\bar g]= \frac{1}{16 \pi G}\lim_{r \to \infty} \int_0^{2\pi}
  rd\phi \Big[\xi^r(\bar D^u h-\bar D_\sigma
  h^{u\sigma} +\bar D^r h^u_r-\bar D^u h^r_r)\\-\xi^u(\bar D^r h-\bar
  D_\sigma h^{r\sigma}-\bar D^r h^u_u+\bar D^u h^r_u)+\xi^\phi(\bar D^r
  h^{u}_\phi-\bar D^u h^{r}_\phi) +\frac{1}{2}h(\bar
  D^r\xi^u-\bar D^u\xi^r)\\+\half h^{r\sigma}(\bar D^u\xi_\sigma-\bar
  D_\sigma\xi^u)-\half  h^{u\sigma}(\bar D^r\xi_\sigma-\bar
  D_\sigma\xi^r)\Big].\label{eq:intch2}
\end{multline}
Using
\begin{equation}
\bar D^u h-\bar D_\sigma
  h^{u\sigma} +\bar D^r h^u_r-\bar D^u h^r_r
=-\frac{1}{r} h_{ur},\nonumber
\end{equation} 
\begin{equation}
  \bar D^r h-\bar D_\sigma h^{r\sigma}-\bar D^r h^u_u+\bar D^u h^r_u
 =-\frac{1}{r}h_{uu} +\frac{2}{r} h_{ur}+ \frac{1}{r^{2}}\d_\phi h_{u\phi},\nonumber
\end{equation}
\begin{equation}
  \bar D^r h^u_\phi-\bar D^u h^r_\phi=(\frac{1}{r}-\d_r) h_{u\phi},\nonumber
\end{equation}
\begin{equation}
  \bar D^r\xi^u-\bar D^u\xi^r=-2\d_\phi Y+O(1),\nonumber
\end{equation}
\begin{equation}
\half h^{r\sigma}(\bar D^u\xi_\sigma-\bar
  D_\sigma\xi^u)-\half  h^{u\sigma}(\bar D^r\xi_\sigma-\bar
  D_\sigma\xi^r)=-2\d_\phi Y h_{ur}+\frac{1}{r}h_{u\phi}Y+O(r^{-2}), \nonumber
\end{equation}
we get
\begin{multline}
   \cQ_\xi[g-\bar g,\bar g]=\frac{1}{16 \pi G}\lim_{r \to \infty}
   \int_0^{2\pi}d\phi\, 
\Big[ (r h_{ur}+u h_{uu})\d_\phi Y+ h_{uu}T +2 h_{u\phi}
Y\Big]\\
 \approx\frac{1}{16 \pi G}
   \int_0^{2\pi}d\phi\, \left((\Theta+1) T +2\Xi Y\right).  \label{eq:36b}
\end{multline}
The surface charges thus provide the inner product that turns the
linear spaces of solutions and asymptotic symmetries into dual
spaces. It follows that the solutions that we have constructed are all
non-trivial as different solutions carry different charges. 

{}From $\delta^g_{\xi_1}\xi_2^u=0$,
$\delta^g_{\xi_1}\xi_2^u=O(r^{-2})$ and
$\delta^g_{\xi_1}\xi_2^r=O(r^{-1})$, it follows that only the Lie
bracket $[\xi_1,\xi_2]$ contributes on the right hand side of
\eqref{bracket1}-\eqref{bracket2}, $\cQ_{[\xi_1,\xi_2]_M}[g-\bar
g,\bar g]=\cQ_{[\xi_1,\xi_2]}[g-\bar g,\bar g]$. Using \eqref{eq:36b},
\eqref{eq:B11} and integrations by parts in $\d_\phi$ to evaluate the
left hand side, one indeed finds
\begin{align}
  \cQ_{\xi_1}[\cL_{\xi_2} g,\bar g] & \approx
  \cQ_{[\xi_1,\xi_2]}[g-\bar g,\bar g] + K_{\xi_1 , \xi_2} ,\\
  K_{\xi_1 , \xi_2} &= \frac{1}{8\pi G } \int_0^{2\pi}d \phi \,
  \Big[\d_\phi Y_1(T_2+\d^2_\phi T_2)-\d_\phi Y_2(T_1+\d^2_\phi T_1)\Big].
\end{align}
where $K_{\xi_1 , \xi_2}$ is the central charge.

The associated numerical value for the classical central charge is obtained by
evaluating $K_{\xi_1 , \xi_2}$ for the generators $l_n$ and $t_n$. We obtain
\begin{equation}
c_1=0,\quad c_2=\frac{3}{G}\,.
\end{equation}
which is the result originally derived in
\cite{Barnich:2007nx}. We saw at the end of section
\ref{sec:bms3algebra} that the general form of the $\mathfrak{bms}_3$
algebra can be obtained by a ``flat'' limit coming from the 2D
conformal algebra. This flat limit applied to the anti-de-Sitter case
studied in the previous chapter ($c^\pm=\frac{3l}{2G}$) leads to the same value for the
$\mathfrak{bms}_3$ central charges.

An important question is a complete understanding of the physically
relevant representations of $\mathfrak{bms}_3$. Note that in the
present gravitational context, the Hamiltonian is associated with
$t_0$, so that one is especially interested in representations with a
lowest eigenvalue of $t_0$. This question should be tractable, given
all that is known on both the Poincar\'e and Virasoro subalgebras of
$\mathfrak{bms}_3$.

It turns out that $\mathfrak{bms}_3$ is isomorphic to the Galilean
conformal algebra in $2$ dimensions $\mathfrak{gca}_2$
\cite{Bagchi:2010zr}. In a different context, a class of
non-unitary representations of $\mathfrak{gca}_2$ have been studied in
some details\cite{Bagchi:2010ly}.

\chapter{${\rm \bf BMS_4/CFT_2}$ correspondence} 
\label{chap:bms4}

This last chapter is devoted to the study of asymptotically flat
spacetimes at null infinity  in 4 dimensions. As we said in the
introduction, this is the first example where the asymptotic symmetry
group is 
enhanced with respect to the isometry group of the background metric
and becomes infinite-dimensional
\cite{Bondi:1962zr,Sachs:1962ly,Sachs:1962fk}.  In this case, the induced
metric is $2$-dimensional because the boundary is a null surface. The
asymptotic symmetry group of non singular transformations is the
well-known Bondi-Metzner-Sachs group. It consists of the semi-direct
product of the group of globally defined conformal transformations of
the unit $2$-sphere, which is isomorphic to the orthochronous
homogeneous Lorentz group, times the infinite-dimensional abelian
normal subgroup of so-called supertranslations.

We start by an analysis of the symmetry group of asymptotically
flat spacetimes at null infinity in the gauge fixed form introduced by
Sachs. The asymptotic Killing vectors form a representation of the
group mentioned above through the modified Lie bracket introduced in
chapter \ref{chap:ads3}. This is when one focus on globally defined
transformations. There is a further enhancement when one focuses on infinitesimal
transformations and does not require the associated finite
transformations to be globally well-defined. The symmetry algebra is
then the semi-direct sum of the infinitesimal local conformal
transformations of the $2$-sphere with the abelian ideal of
supertranslations, and now both factors are infinite-dimensional.

The rest of this chapter is an attempt to apply the analysis of the
two previous cases ($AdS_3$ and $BMS_3$) to this new $BMS_4$. We first
show that there is no possible central extension. We then solve the
equations of motion and show how the group is represented on solution
space. The surface charges are non-integrable and we use the proposition
of Wald and Zoupas to deal with the non-integrable
term. Unfortunately, this spoils the usual covariant definition of the
algebra.

The last section is a proposition for a new bracket for the charges. It
leads to an algebra of charges that form a representation of the
asymptotic symmetries up to a general field-dependent extension.

\section{Asymptotically flat $4$-d spacetimes at null infinity}
\label{sec:asympt-flat-4}

Let $x^0=u,x^1=r,x^2=\theta,x^3=\phi$ and $A,B,\dots=2,3$. Following
mostly Sachs \cite{Sachs:1962fk} up to notation, the metric $g_{\mu\nu}$ of an
asymptotically flat spacetime can be written in the form
\begin{equation}
  \label{eq:bms4sachsmetric}
  ds^2=e^{2\beta}\frac{V}{r} du^2-2e^{2\beta}dudr+
g_{AB}(dx^A-U^Adu)(dx^B-U^Bdu)
\end{equation}
where $\beta,V,U^A,g_{AB} (\text{det}\, g_{AB})^{-1/2}$ are $6$
functions of the coordinates, with $\text{det}\, g_{AB}=r^4
b(u,\theta,\phi)$ for some fixed function $b(u,\theta,\phi)$.  The
inverse to the metric
\begin{equation*}
g_{\mu\nu}=  \begin{pmatrix} e^{2\beta} \frac{V}{r}+g_{CD}U^CU^D & 
  -e^{2\beta} & -g_{BC} U^C \\
  -e^{2\beta} & 0 & 0 \\
-g_{AC} U^C & 0 & g_{AB} 
\end{pmatrix}
\end{equation*}
is given by 
\begin{equation*}
g^{\mu\nu}=  \begin{pmatrix} 0 & 
  -e^{-2\beta} & 0 \\
  -e^{-2\beta} & 
-\frac{V}{r} e^{-2\beta} & -U^B e^{-2\beta} \\
0 & -U^A e^{-2\beta} & g^{AB}
\end{pmatrix}.
\end{equation*}
The fall-off conditions are 
\begin{equation}
  \label{eq:3a}
g_{AB}dx^Adx^B=r^2\bar \gamma_{AB}dx^Adx^B+O(r),
\end{equation}
Sachs chooses $\bar\gamma_{AB}={}_0\gamma_{AB}$ to be the metric on
the unit $2$ sphere, ${}_0\gamma_{AB}dx^Adx^B=d\theta^2+\sin^2\theta
d\phi^2$ and $b=\sin^2\theta$, but the geometrical analysis by Penrose
\cite{Penrose:1963uq} suggests to be somewhat more general and use
a metric that is conformal to the latter, for instance,
$\bar\gamma_{AB}dx^Adx^B=e^{2\varphi}(d\theta^2+\sin^2\theta
d\phi^2)$, with $\varphi=\varphi(u,x^A)$. We will choose the
determinant condition more generally to be
$b(u,x^A)=\text{det}\bar\gamma_{AB}$. In particular, in the above
example, on which we now focus, $b=e^{4\varphi}\sin^2\theta$.

The rest of the fall-off conditions are given by
\begin{equation}
\label{eq:BMS4falloff}
\beta=O(r^{-2}),\quad V/r=-2r\dot\varphi-e^{-2\varphi}+
  \bar\Delta\varphi+O(r^{-1}), 
\quad U^A=O(r^{-2}).
\end{equation}
Here, a dot denotes the derivative with respect to $u$, $\bar D_A$
denotes the covariant derivative with respect to $\bar
\gamma_{AB}$. We denote by $\bar \Gamma^A_{BC}$ the associated
Christoffel symbols and by $\bar\Delta$ the associated
Laplacian. Similarly, ${}_0 D_A,{}_0\Gamma^A_{BC},{}_0\Delta$
correspond to ${}_0\gamma_{AB}$.  Note that $g^{AB}g_{BC}=\delta^A_C$
and that the condition on the determinant implies
\begin{gather}
  \label{eq:BMS4detcond}
\left\{\begin{array}{l}
  g^{AB}\d_rg_{AB}=4 r^{-1},\\
g^{AB}\d_u g_{AB}=\bar\gamma^{AB}\d_u\bar\gamma_{AB}=4\dot\varphi, \\
g^{AB}\d_C g_{AB}=\bar
\gamma ^{AB}{}\d_C\bar
\gamma_{AB}={}_0\gamma^{AB}\d_C\,{}_0\gamma_{AB}+4\d_C\varphi,
\end{array}\right.
\end{gather}
where $\bar\gamma^{AB}\bar\gamma_{BC}=\delta^A_C={}_0\gamma^{AB}{}_0
\gamma_{BC}$.  In terms of the metric and its inverse, the fall-off
conditions read
\begin{gather*}
\left\{\begin{array}{l}
g_{uu}=-2r\dot\varphi-e^{-2\varphi}+\bar\Delta\varphi+O(r^{-1}),\ 
g_{ur}=-1+O(r^{-2}),\
g_{uA}=O(1), \\
g_{rr}=0=g_{rA},\quad 
g_{AB}=r^2\bar\gamma_{AB}+O(r),\\
g^{ur}=-1+O(r^{-2}),\quad g^{uu}=0=g^{uA},\\
g^{rr}=2r\dot\varphi+e^{-2\varphi}-\bar\Delta\varphi+O(r^{-1}),\ g^{rA}=O(r^{-2}),\
g^{AB}=r^{-2}\bar\gamma^{AB}+O(r^{-3}).
\end{array}\right.
\end{gather*}

\section{Asymptotic symmetries}

With the choice $\varphi=0$, Sachs studies the vector fields that
leave invariant this form of the metric with these fall-off
conditions. More precisely, he finds the general
solution to the equations
\begin{gather}
  \label{eq:bms4symasympt1}
  \cL_\xi g_{rr}=0,\quad \cL_\xi g_{rA}=0,\quad \cL_\xi g_{AB}
  g^{AB}=0,\\
\cL_\xi
g_{ur}=O(r^{-2}),\quad \cL_\xi g_{uA}=O(1),\quad  \cL_\xi
g_{AB}=O(r),\quad 
\cL_\xi g_{uu}=O(r^{-1}). \label{eq:bms4symasympt2}
\end{gather}
For arbitrary $\varphi$, the general solution to \eqref{eq:bms4symasympt1} is given by 
\begin{gather}
  \label{eq:bms4vectasymp}
\left\{\begin{array}{l}
  \xi^u=f,\\
\xi^A=Y^A+I^A, \quad  I^A=- f_{,B} \int_r^\infty dr^\prime(
e^{2\beta} g^{AB}),\\
\xi^r=-\half r (\psi+\chi-f_{,B}U^B+2f\d_u\varphi),
\end{array}\right.
\end{gather}
with $\d_r f=0=\d_r Y^A$ and where $ \psi= \bar D_A Y^A$, $\chi= \bar
D_A I^A$.  This gives the expansions
\begin{gather}
  \label{eq:bms4expvectasymp}
  \left\{\begin{array}{l} \xi^u=f,\quad
      \xi^A=Y^A-r^{-1} f_{,B}\bar\gamma^{BA}+O(r^{-2}),\\
      \xi^r=-r(f\dot\varphi+\frac{1}{2}\psi)+\half \bar\Delta f +O(r^{-1}).
\end{array}\right.
\end{gather}
The first equation of \eqref{eq:bms4symasympt2} then implies that
\begin{equation}
\dot f =f\dot\varphi+\half \psi\iff f=e^{\varphi}\big[T+
 \half\int_0^udu^\prime
e^{-\varphi}\psi\big],\label{eq:bms4eqf}
\end{equation}
with $T=T(\theta,\phi)$, while the second requires $\d_u Y^A=0$ and
thus $Y^A=Y^A(\theta,\phi)$. The third one implies that $Y^A$ is a
conformal Killing vector of $\bar\gamma_{AB}$ and thus also of
${}_0\gamma_{AB}$. The last equation of \eqref{eq:bms4symasympt2} is then
satisfied without additional conditions. For the computation, one uses
that $\bar\Delta =e^{-2\varphi}{}_0\Delta $ and $\psi={}_0 \psi +2 Y^A\d_A
\varphi$, with ${}_0\psi={}_0D_A Y^A$ and the following properties
of conformal Killing vectors $Y^A$ on the unit $2$-sphere,
\begin{equation}
  \label{eq:BMS4confvect}
  2{}_0 D_B{}_0 D_C Y_A={}_0 \gamma_{CA}{}_0 D_B{}_0\psi+{}_0 \gamma_{AB}
{}_0 D_C{}_0\psi-{}_0 \gamma_{BC}{}_0 D_A{}_0 \psi+2
  Y_C {}_0 \gamma_{BA}-2 Y_A {}_0 \gamma_{BC},  
\end{equation}
where the indices on $Y^A$ are lowered with $\bar\gamma_{AB}$.  This
implies in particular ${}_0 \Delta Y^A=-Y^A$ and also that ${}_0
\Delta {}_0\psi =-2{}_0\psi$.

\vspace{5mm}

By definition, the algebra $\mathfrak{bms}_4$ is the semi-direct sum
of the Lie algebra of conformal Killing vectors $Y^A\dover{}{x^A}$ of the unit
$2$-sphere with the abelian ideal consisting of functions
$T(x^A)$ on the $2$-sphere. The bracket is defined through 
\begin{equation}
\label{eq:32}
\begin{split}
(\hat Y,\hat T)&=[(Y_1,T_1),(Y_2,T_2)], \\
\hat Y^A&= Y^B_1\d_B
Y^A_2-Y^B_1\d_B Y^A_2,\\
\hat T&=Y^A_1\d_A
  T_2-Y^A_2\d_A T_1 +\half (T_1\,{}_0\psi_2-T_2\,{}_0\psi_1).
\end{split}
\end{equation}

Let $\cI={\mathbb R}\times S^2$ with coordinates $u,\theta,\phi$. On
$\cI$, consider the scalar field $\varphi$ and the vectors fields
$\bar\xi(\varphi,T,Y)=f\dover{}{u} +Y^A\dover{}{x^A}$, with $f$ given
in \eqref{eq:bms4eqf} and $Y^A$ an $u$-independent conformal Killing
vector of $S^2$. It is straightforward to check that these vector
fields form a faithful representation of $\mathfrak{bms}_4$ for the
standard Lie bracket.

Consider then the modified Lie bracket
\begin{equation}
  \label{eq:43}
  [\xi_1,\xi_2]_M=[\xi_1,\xi_2]-\delta^g_{
    \xi_1}\xi_2+\delta^g_{ \xi_2}\xi_1,
\end{equation}
where $\delta^g_{\xi_1}\xi_2$ denotes the variation in $\xi_2$
under the variation of the metric induced by $\xi_1$, $\delta^g_{
  \xi_1}g_{\mu\nu}=\cL_{\xi_1}g_{\mu\nu}$, 

{\em Spacetime vectors $\xi$ of the form \eqref{eq:bms4vectasymp}, with
  $Y^A(x^B)$ a conformal Killing vectors of the $2$-sphere and
  $f(u,x^B)$ satisfying \eqref{eq:bms4eqf} provide a faithful
  representation of $\mathfrak{bms}_4$ when equipped with the modified
  Lie bracket $[\cdot,\cdot]_M$.}

Indeed, for the $u$ component, there is no modification due to the
change in the metric and the result follows by direct computation:
$[\xi_1,\xi_2]^u_{(M)}=\hat f$, where $\hat f$ corresponds to $f$ in
\eqref{eq:bms4eqf} with $T$ replaced by $\hat T$ and $Y$ by $\hat Y$. By
evaluating $\cL_\xi g^{\mu\nu}$, we find
\begin{gather}
\left\{\begin{array}{l}
  \label{eq:bms4varcompmetric}
 \delta_\xi\varphi=0,\\
  \delta_\xi\beta=\xi^\alpha\d_\alpha\beta+\half \big[\d_u f
  +\d_r\xi^r+\d_A f U^A\big],\\
  \delta_\xi U^A = \xi^\alpha \partial_\alpha U^A + U^A (\partial_u f
  + \partial_B f U^B) -\d_B\xi^A U^B\\\hspace{1cm} - \partial_u \xi^A
  - \partial_r \xi^A \frac{V}{r} + \partial_B \xi^r
  g^{AB}e^{2\beta}.
\end{array}
\right.
\end{gather}
It follows that 
\begin{gather*}
\left\{\begin{array}{l}
  \delta^g_{\xi_1}\xi^A_2=-\d_B
  f_2\int^\infty_r dr^\prime
    e^{2\beta}(2\delta_{\xi_1}\beta g^{AB}+\cL_\xi g^{AB}),\\
\delta^g_{\xi_1}\xi^r_2=-\half r\big[\bar D_A (\delta^g_{\xi_1}\xi^A_2)
-\d_A f_2 \delta_{\xi_1} U^A\big].
\end{array}
\right.
\end{gather*}
We have $\lim_{r\to\infty}[\xi_1,\xi_2]^A_M=\hat Y^A$ and, using
$\d_r\xi^A=g^{AB}e^{2\beta}\d_B f$, \eqref{eq:bms4eqf} together with the expression
of $\xi^r$ in \eqref{eq:bms4expvectasymp}, it follows by a straightforward
computation that $\d_r([\xi_1,\xi_2]^A_M)=g^{AB} e^{2\beta}\d_B\hat
f$, which gives the result for the $A$ components.  Finally, for the
$r$ component, we need
\[\d_r(\frac{\xi^r}{r})=-\half\big(\d_r\chi-\d_B f \d_r U\big).\]
We then find $\lim_{r\to\infty}\frac{[\xi_1,\xi_2]^r_M}{r}=-\half
(\hat \psi+2\hat f \d_u\varphi)$, where $\hat \psi$ corresponds to
$\psi$ with $Y^A$ replaced by $\hat Y^A$, while
$\d_r(\frac{[\xi_1,\xi_2]^r_M}{r})=-\half (\d_r \hat \chi-\d_B\hat
f\d_r U^B)$, where $\hat \chi$ corresponds to $\chi$ with $f$ replaced
by $\hat f$. This gives the result for the $r$ component and concludes
the proof.

More generally, one can also consider the transformations that leave
the form of the metric \eqref{eq:bms4sachsmetric} invariant up to a conformal
rescaling of $g_{AB}$, i.e., up to a rescaling of $\varphi$ by
$\omega(u,x^A)$. They are generated by spacetime vectors satisfying
\begin{gather}
  \label{eq:bms4symasptphi1}
  \cL_\xi g_{rr}=0,\quad \cL_\xi g_{rA}=0,\quad \cL_\xi g_{AB}
  g^{AB}=4\omega,\\
\cL_\xi
g_{ur}=O(r^{-2}),\quad \cL_\xi g_{uA}=O(1),\quad  \cL_\xi
g_{AB}=2\omega g_{AB}+O(r),\nonumber\\
\cL_\xi g_{uu}=-2r\dot\omega-2\omega e^{-2\varphi}+2\omega\bar\Delta
\varphi + O(r^{-1}). \label{eq:bms4symasptphi2}
\end{gather}

Equations \eqref{eq:bms4symasptphi1}, \eqref{eq:bms4symasptphi2} then imply that the
vectors are given by \eqref{eq:bms4vectasymp} and \eqref{eq:bms4eqf} with the
replacement $\psi\to \psi-2\omega$.

With this replacement, the vector fields $\bar\xi=f \dover{}{u}+
Y^A\dover{}{x^A}$ on $\cI={\mathbb R}\times S^2$ equipped with the
modified bracket provide a faithful representation of the extension of
$\mathfrak{bms}_4$ defined by elements $(Y,T,\omega)$ and bracket
$[(Y_1,T_1,\omega_1),(Y_2,T_2,\omega_2)]=(\hat Y,\hat T,\hat\omega)$,
with $\hat Y,\hat T$ as before and $\hat \omega=0$.

Indeed, the result is obvious for the $A$ components. Furthermore,
\[\delta^g_{\bar \xi_1} f_2= \omega_1 f_2+\half e^{\varphi}\int_0^udu^\prime
e^{-\varphi}[-\omega_1(\psi_2-2\omega_2)+
2Y^A_2\d_A \omega_1].\] At $u=0$, we get
$[\bar\xi_1,\bar\xi_2]^u_M|_{u=0}=e^\varphi|_{u=0} \hat T$, while
direct computation shows that $\d_u ([\bar\xi_1,\bar\xi_2]^u_M)=\hat
f\dot\varphi +\half \bar D_B\hat Y^B$, as it should. 

Following the same reasoning as before, one can then also show that
the spacetime vectors \eqref{eq:bms4vectasymp} with the replacement
discussed above and equipped with the modified Lie bracket provide a
faithful representation of the extended $\mathfrak{bms}_4$ algebra. 

Indeed, we have $\xi=\bar\xi +I^A\dover{}{x^A}+\xi^r\dover{}{r}$.
Furthermore, $[\xi_1,\xi_2]_M^u=[\bar\xi_1,\bar\xi_2]^u_M=\hat f$ as
it should.  In the extended case, the variations of $\beta,U^A$ are
still given by \eqref{eq:bms4varcompmetric}.  We then have
$\lim_{r\to\infty}[\xi_1,\xi_2]_M^A=\hat Y^A$ and find, after some
computations, $\d_r ([\xi_1,\xi_2]_M^A)=g^{AB}e^{2\beta} \d_B\hat f$,
giving the result for the $A$ components.  Finally, for the $r$
component, we find
$\lim_{r\to\infty}\frac{[\xi_1,\xi_2]^r_M}{r}=-\half (\hat \psi+2\hat
f \dot\varphi)$, while $\d_r(\frac{[\xi_1,\xi_2]^r_M}{r})=-\half (\d_r
\hat \chi-\d_B\hat f\d_r U^B)$, which concludes the proof.

\section{Extended $\mathfrak{bms}_4$ Lie algebra}
\label{sec:mathfr-algebra}

As we showed in the previous section, the algebra $\mathfrak{bms}_4$ is the semi-direct sum
of the Lie algebra of conformal Killing vectors $Y^A\dover{}{x^A}$ of the unit
$2$-sphere with the abelian ideal consisting of functions
$T(x^A)$ on the $2$-sphere, the bracket being defined as
\begin{equation}
\label{eq:bms4algebrasphere}
\begin{split}
(\hat Y,\hat T)&=[(Y_1,T_1),(Y_2,T_2)], \\
\hat Y^A&= Y^B_1\d_B
Y^A_2-Y^B_1\d_B Y^A_2,\\
\hat T&=Y^A_1\d_A
  T_2-Y^A_2\d_A T_1 +\half (T_1\,{}_0\psi_2-T_2\,{}_0\psi_1).
\end{split}
\end{equation}

In terms of the standard complex coordinates $\zeta=
e^{i\phi}\cot{\frac{\theta}{2}}$, the metric on the sphere is
conformally flat, 
  \begin{equation}
    \label{eq:6}
    d\theta^2+\sin^2\theta
    d\phi^2=P^{-2}d\zeta
    d\bar\zeta,\quad P(\zeta,\bar\zeta)=\half(1+\zeta\bar\zeta), 
  \end{equation}
and, since conformal Killing vectors are invariant under conformal
rescalings of the metric, the conformal Killing vectors of the unit
sphere are the same as the conformal Killing vectors of the
Riemann sphere. 

Depending on the space of functions under consideration, there are
then basically two options which define what is actually meant by
$\mathfrak{bms}_4$.

The first choice consists in restricting oneself to globally
well-defined transformations on the unit or, equivalently, the Riemann
sphere. This singles out the global conformal transformations, also
called projective transformations, and the associated group is
isomorphic to $SL(2,\mathbb C)/\mathbb Z_2$, which is itself
isomorphic to the proper, orthochronous Lorentz group. Associated with
this choice, the functions $T(\theta,\phi)$, which are the generators
of the so-called supertranslations, have been expanded into spherical
harmonics. This choice has been adopted in the original work by Bondi,
van der Burg, Metzner and Sachs and followed ever since, most notably
in the work of Penrose and Newman-Penrose
\cite{Penrose:1963uq,Newman:1966kx}, where spin-weighted spherical
harmonics and the associated ``edth'' operator have made their
appearance. After attempts to cut this group down to the standard
Poincar\'e group, it has been taken seriously as an invariance group
of asymptotically flat spacetimes. Its consequences have been
investigated, but we believe that it is fair to say that this version
of the BMS group has had only a limited amount of success.

\vspace{5mm}

The second choice is motivated by
exactly the same considerations that are at the origin of the
breakthrough in two dimensional conformal quantum field theories
\cite{Belavin:1984vn}. It consists in focusing on local properties and
allowing the set of all, not necessarily invertible holomorphic
mappings. In this case, the conformal Killing
vectors are given by two copies of the Witt algebra, so that besides
supertranslations, there now also are superrotations. 

The general solution to the conformal Killing equations is
$Y^\zeta=Y(\zeta)$, $Y^{\bar\zeta}=\bar Y(\bar\zeta)$, with $Y$ and
$\bar Y$ independent functions of their arguments. The standard basis
vectors are chosen as
\begin{equation}
l_n=-\zeta^{n+1}\frac{\d}{\d\zeta},\quad \bar l_n=-\bar
\zeta^{n+1}\frac{\d}{\d\bar \zeta},\quad n\in \mathbb Z\label{eq:55}
\end{equation}
At the same time, let us choose to expand the generators of the
supertranslations in terms of
\begin{equation}
  T_{m,n}=P^{-1}\zeta^m\bar\zeta^n, 
\quad m,n\in\mathbb Z. 
\end{equation}
In terms of the basis vector $l_l\equiv (l_l,0)$ and
$T_{mn}=(0,T_{mn})$, the commutation relations for the complexified
$\mathfrak{bms}_4$ algebra read
\begin{equation}
\begin{gathered}
  \label{eq:bms4algebra1}
  [l_m,l_n]=(m-n)l_{m+n},\quad [\bar l_m,\bar l_n]=(m-n)\bar
  l_{m+n},\\
[l_l,T_{m,n}]=(\frac{l+1}{2}-m)T_{m+l,n},
\quad [\bar l_l,T_{m,n}]= (\frac{l+1}{2}-n)T_{m,n+l}, \\
[l_m,\bar l_n]=0, \quad [T_{m,n},T_{o,p}]=0.
\end{gathered}
\end{equation}

The $\mathfrak{bms}_4$ algebra contains as subalgebra the Poincar\'e
algebra, which we identify with the algebra of exact Killing vectors
of the Minkowski metric equipped with the standard Lie bracket.

Indeed, these vectors form the subspace of spacetime vectors
\eqref{eq:bms4vectasymp} for which (i) $\beta=0=U^A=\varphi$ while $V=-r$ and
$g_{AB}={}_0\gamma_{AB}$ and (ii) the relations in \eqref{eq:bms4symasympt2} hold
with $0$ on the right hand sides.  The former implies in particular
that $I^A=-\frac{1}{r}{}_0\gamma^{AB}\d_B f$, while a first
consequence of the latter is that the modified Lie bracket reduces the
standard one.

Besides the previous conditions that $Y^A$ is an $u$-independent
conformal Killing vector of the $2$-sphere,
$\cL_{Y}\,{}_0\gamma_{AB}={}_0D_C Y^C\,{}_0\gamma_{AB}$ and $f=T+\half
u\,{}_0\psi$ with $T_{,u}=0=T_{,r}$, we find the additional
constraints
\begin{gather}
  \label{eq:78}
  {}_0D_A \d_B\,{}_0\psi+{}_0D_B\d_A\,{}_0\psi={}_0
\gamma_{AB}\,{}_0\Delta\,{}_0\psi,\\
  {}_0D_A \d_B T+{}_0D_B\d_A T={}_0\gamma_{AB}\,{}_0\Delta T,\quad
  \d_A T=-\half \d_A({}_0\Delta T).
\end{gather}
In the coordinates $\zeta,\bar\zeta$, these constraints are equivalent
to $\d^3 Y=0=\bar\d^3 \bar Y$ and $\d^2\tilde T=0=\bar\d^2\tilde T$,
where $T=P^{-1}\tilde T$ and $\d=\frac{\d}{\d\zeta}$,
$\bar\d=\frac{\d}{\d\bar\zeta}$, so that the complexified Poincar\'e algebra is
spanned by the generators
\begin{equation}
  l_{-1},\,l_0,\,l_1,\quad \bar l_{-1},\, \bar l_0,\, \bar l_1,\quad
T_{0,0},\,T_{1,0},\T_{0,1},\,T_{1,1},\label{eq:61}
\end{equation}
and the non vanishing commutation relations read
\begin{equation}
\begin{gathered}
  \label{eq:79}
[l_{-1},l_0]=-l_{-1},\ [l_{-1},l_1]=-2 l_0,\
[l_0,l_1]=-l_1,\\
[l_{-1},T_{1,0}]=-T_{0,0},\ [l_{-1},T_{1,1}]=-T_{0,1},\ [\bar
l_{-1},T_{0,1}]=-T_{0,0},\ [\bar l_{-1},T_{1,1}]=-T_{1,0},\\
[l_{0},T_{0,0}]=\half T_{0,0},\ [l_{0},T_{0,1}]=\half T_{0,1}, \
[l_{0},T_{1,0}]=-\half T_{1,0},\ [l_{0},T_{1,1}]=-\half T_{1,1},\\
[\bar l_{0},T_{0,0}]=\half T_{0,0},\ [\bar l_{0},T_{0,1}]=-\half T_{0,1}, \
[\bar l_{0},T_{1,0}]=\half T_{1,0},\ [\bar l_{0},T_{1,1}]=-\half T_{1,1},\\
[l_{1},T_{0,0}]= T_{1,0},\ [l_{1},T_{0,1}]= T_{1,1},\ [\bar
l_{1},T_{0,0}]=T_{0,1},\ [\bar l_{1},T_{1,0}]=T_{1,1}. 
\end{gathered}  
\end{equation} 
In particular for instance, the generators for translations are
associated to $\half(T_{1,1}+T_{0,0})=1$,
$\half(T_{1,1}-T_{0,0})=\cos\theta$, $\half ( T_{1,0}+
T_{0,1})=\sin\theta\cos\phi$, $\frac{1}{2i}( T_{1,0}-
T_{0,1})=\sin\theta\sin\phi$. Note that in order for the asymptotic
symmetry algebra to contain the Poincar\'e algebra as a subalgebra, it
is essential not to restrict the generators of supertranslations to
the sum of holomorphic and antiholomorphic functions on the Riemann
sphere.

The quotient algebra of $\mathfrak{bms}_4$ by the abelian ideal of
infinitesimal supertranslations is no longer given by the Lorentz
algebra but by two copies of the Witt algebra. It follows that the
problem with angular momentum in general relativity
\cite{Winicour:1980ys}, at least in its group theoretical formulation,
disappears as now the choice of an infinite number of conditions is
needed to fix an infinite number of rotations.

The considerations above apply for all choices of $\varphi$ which is
freely at our disposal. In the original work of Bondi, van der Burg,
Metzner and Sachs, and in much of the subsequent work, the choice
$\varphi=0$ was preferred. From the conformal point of view, the
choice
\begin{equation}
\varphi=\ln{[\half(1+\zeta\bar\zeta)]}
\end{equation}
is interesting as it turns $\bar\gamma_{AB}$ into the flat metric on
the Riemann sphere with vanishing Christoffel symbols,
\begin{equation}
  \bar\gamma_{AB}dx^Adx^B=d\zeta d\bar\zeta. 
\end{equation}
In this case, $\psi =\d_A Y^A $,
\begin{equation}
f=\tilde T+\half u \psi,\label{eq:BMS4fsimple}
\end{equation}
 with $\tilde
T=PT$. In terms of $\tilde T$, we get instead of the last of
\eqref{eq:bms4algebrasphere}
\begin{equation}
\hat{\tilde T}=Y^A_1\tilde T_2+\half \tilde T_1 \d_A
Y_2^A-(1\leftrightarrow 2).
\end{equation}
In terms of generators, the algebra \eqref{eq:bms4algebra1} is unchanged if one
now expands the supertranslations $\tilde T$ directly in terms of
$\tilde T_{m,n}=\zeta^m\bar\zeta^n$. 


\section{Central extensions of $\mathfrak{bms}_4$}

The $\mathfrak{bms}_4$ algebra can also be viewed as an abstract
structure defined on the 2-sphere with the help of tensor densities
introduced in section \ref{sec:bms3algebra}. 

In stereographic coordinates
$\zeta=e^{i\phi}\cot{\frac{\theta}{2}}$ and $\bar \zeta$ for the $2$ sphere,
the algebra may be realized through the vector fields $y=Y(\zeta)\d $,
$\bar y=\bar Y(\bar \zeta)\bar \d$, with $\d= \dover{}{\zeta}$, $\bar
\d=\dover{}{\bar \zeta}$. They act on tensor densities
$\cF_{\half,\half}$ of degree $(\half,\half)$,
$t=T(\zeta,\bar\zeta)e^{-\varphi_0}(d\zeta)^{-\half}(d\bar\zeta)^{-\half}$,
where $\varphi_0=\ln{\half(1+\zeta\bar\zeta)}$ through
\begin{eqnarray}
\rho(y) t &=& (Y\d T-\half \d Y
T)e^{-\varphi_0}(d\zeta)^{-\half}(d\bar\zeta)^{-\half}\,,\\
\rho(\bar y) t&=&(\bar
Y\bar \d T-\half \bar \d \bar Y
T)e^{-\varphi_0}(d\zeta)^{-\half}(d\bar\zeta)^{-\half}\,. 
\end{eqnarray}
The algebra $\mathfrak{bms}_4$ is then the semi-direct sum of the
algebra of vector fields $y,\bar y$ with the abelian ideal
$\cF_{\half,\half}$, the bracket being induced by the module action,
$[y,t]=\rho(y)t$, $[\bar y,t]=\rho(\bar y) t$.

When expanding $y=a^nl_n$, $\bar y=\bar a^n\bar l_n$, $t=b^{m,n} T_{m,n}$, 
with
\begin{equation}
  \label{eq:10}
  l_n=-\zeta^{n+1}\d,\quad \bar l_n=-\bar \zeta^{n+1}\bar\d,\quad 
T_{m,n}=\zeta^m\bar\zeta^n
e^{-\varphi_0}(d\zeta)^{-\half}(d\bar\zeta)^{-\half}\,,
\end{equation}
the enhanced symmetry algebra reads
\begin{eqnarray}
  [l_m,l_n]=(m-n)l_{m+n},&\quad& [\bar l_m,\bar l_n]=(m-n)\bar
  l_{m+n},\cr
 [l_l,T_{m,n}]=(\frac{l+1}{2}-m)T_{m+l,n},
&\quad& [\bar l_l,T_{m,n}]= (\frac{l+1}{2}-n)T_{m,n+l}, \cr
[l_m,\bar l_n]=0, &\quad& [T_{m,n},T_{o,p}]=0,
\end{eqnarray}
where $m,n\dots\in \mathbb Z$, which is exactly the algebra we saw in
the previous section.

\vspace{5mm}

The only non trivial central
extensions of $\mathfrak{bms}_4$ are the usual central extensions of
the 2 copies of the Witt algebra, i.e., they appear in the commutators
$[l_m,l_{-m}]$ and $[\bar l_m,\bar l_{-m}]$.  Contrary to three
dimensions, there are no central extensions involving the generators
for supertranslations. In other words, up to equivalence, the most general central extension of
$\mathfrak{bms}_4$ is
\begin{eqnarray}
  \label{eq:37}
&  [l_m,l_n]=(m-n)l_{m+n} + \frac{c}{12} m (m-1)(m+1)
\delta^0_{m+n},\cr & [\bar l_m,\bar l_n]=(m-n)\bar
  l_{m+n}+ \frac{\bar c}{12} m (m-1)(m+1)
\delta^0_{m+n},\cr
& [l_l,T_{m,n}]=(\frac{l+1}{2}-m)T_{m+l,n},
\quad [\bar l_l,T_{m,n}]= (\frac{l+1}{2}-n)T_{m,n+l}\cr & [l_m,\bar l_n]=0 ,\quad [T_{m,n},T_{o,p}]=0\,.
\end{eqnarray}

\begin{proof}
For $\mathfrak{bms}_4$, the Chevally-Eilenberg differential is given
by
\begin{eqnarray}
  \label{eq:4bis}
  \gamma &=&-\half C^mC^{k-m}(2m-k)\dover{}{C^k}-\half \bar C^m\bar
  C^{k-m}(2m-k)\dover{}{\bar C^k}\nonumber
 \\ &&-C^m\xi^{k-m,n}(\frac{3m+1}{2}-k)\dover{}{\xi^{k,n}}
- \bar C^n\xi^{m,k-n}(\frac{3n+1}{2}-k)\dover{}{\xi^{m,k}}
\,,
\end{eqnarray}
in the space $\Lambda(C,\bar C,\xi)$ of polynomials in the
anticommuting ``ghost'' variables $C^m,\bar C^n,\xi^{m,n}$. The
grading is given by the eigenvalues of the ghost number operator,
$N_{C,\xi}=C^m\dover{}{C^m}+\bar C^m\dover{}{\bar
  C^m}+\xi^{m,n}\dover{}{\xi^{m,n}}$, the differential $\gamma$ being
homogeneous of degree $1$ and $H^2(\mathfrak{bms}_4)\simeq
H^2(\gamma,\Lambda(C,\bar C,\xi))$.  Furthermore, when counting only the
ghosts $\xi^m$ associated with supertranslations,
$N_\xi=\xi^{m,n}\dover{}{\xi^{m,n}}$, the differential $\gamma$ is homogeneous
of degree $0$, so that the cohomology decomposes into components of
definite $N_\xi$ degree. The cocycle condition then becomes
\begin{eqnarray}
\gamma(\omega^0_{m,n}C^mC^n+\bar\omega^0_{m,n}\bar C^m\bar
C^n+\omega^{-1}_{m,n}C^m\bar C^n)&=&0,\nonumber\\
\gamma(\omega^1_{k,mn}C^k\xi^{m,n}+\bar \omega^1_{k,mn}\bar
C^k\xi^{m,n})&=&0,\nonumber \\
\gamma(\omega^2_{mn,kl}\xi^{m,n}\xi^{k,l})&=&0, 
\label{eq:5bis}
\end{eqnarray}
with $\omega^0_{m,n}=-\omega^0_{n,m}, \bar \omega^0_{m,n}=-\bar \omega^0_{n,m}$ and
$\omega^2_{mn,kl}=-\omega^2_{kl,mn}$. 
The coboundary condition reads
\begin{eqnarray}
  \label{eq:6bis}
 \omega^0_{m,n}C^mC^n+\bar\omega^0_{m,n}\bar C^m\bar
C^n+\omega^{-1}_{m,n}C^m\bar C^n&=&\gamma(\eta^0_mC^m+\bar\eta^0_m\bar
C^m),\nonumber\\
\omega^1_{k,mn}C^k\xi^{m,n}+\bar \omega^1_{k,mn}\bar
C^k\xi^{m,n}&=&\gamma(\eta^1_{mn}\xi^{m,n}). 
\end{eqnarray}

We have $\{\dover{}{C^0},\gamma\}=\cN_{C,\xi}$ with
$\cN_{C,\xi}=mC^m\dover{}{C^m}+(m-\half)\xi^{m,n}\dover{}{\xi^{m,n}}$ and
also $\{\dover{}{\bar C^0},\gamma\}=\bar \cN_{\bar C,\xi}$ with
$\bar \cN_{\bar C,\xi}=n\bar C^n\dover{}{\bar
  C^n}+(n-\half)\xi^{m,n}\dover{}{\xi^{m,n}}$. 
It follows again that
all cocycles of either $\cN_{C,\xi}$ or $\bar \cN_{\bar C,\xi}$ degree different from $0$ are
coboundaries. Without loss of generality we can thus assume that
$\omega^0_{m,n}C^mC^n+\bar\omega^0_{m,n}\bar C^m\bar
C^n+\omega^{-1}_{m,n}C^m\bar C^n=\omega^0_{m}C^mC^{-m}+\bar\omega^0_{m}\bar C^m\bar
C^{-m}+\omega^{-1}_{0,0}C^0\bar C^0$ with 
$\omega^0_m=-\omega^0_{-m}, \bar \omega^0_m=-\bar \omega^0_{-m}$ and
in particular $\omega^0_0=0=\bar \omega^0_0$; none of monomials with
one $\xi^{m,n}$ and either one $C^k$ or one $\bar C^k$ can be of degree $0$, so 
$\omega^1_{k,mn}=0=\bar \omega^1_{k,mn}$;
$\omega^2_{mn,kl}\xi^{m,n}\xi^{k,l}=\omega^2_{m,n}\xi^{m,n}\xi^{-m+1,-n+1}$ with 
$\omega^2_{m,n}=-\omega^2_{-m+1,-n+1}$. Both the cocycle and the
coboundary condition for $\omega^0_{m}C^mC^{-m}+\bar\omega^0_{m}\bar C^m\bar
C^{-m}+\omega^{-1}_{0,0}C^0\bar C^0$ split. For
$\omega^{-1}_{0,0}C^0\bar C^0$ there is no coboundary condition, while 
the cocycle condition implies $\omega^{-1}_{0,0}=0$. The rest of the
analysis proceeds as in the previous subsection, separately for
$\omega^0_{m}C^mC^{-m}$ and $\bar\omega^0_{m}\bar C^m\bar
C^{-m}$, with the standard central extension for $[l_m,l_{-m}]$ and
$[\bar l_m,\bar l_{-m}]$.

We still have to analyze
$\gamma(\omega^2_{m,n}\xi^{m,n}\xi^{-m+1,-n+1})=0$. This condition
gives
$\omega^2_{m,n}(\frac{3l-1}{2}+m)+\omega^2_{l+m,n}(\frac{l+1}{2}-m)=0$
and also
$\omega^2_{m,n}(\frac{3l-1}{2}+n)+\omega^2_{m,l+n}(\frac{l+1}{2}-n)=0$. 
Putting $m=0$ in the first relation gives
$\omega^2_{0,n}(\frac{3l-1}{2})+\omega^2_{l,n}(\frac{l+1}{2})=0$. Putting
$l=-1$ then implies $\omega^2_{0,n}=0$ and then also
$\omega^2_{l,n}=0$ for $l\neq -1$. But
$\omega^2_{-1,n}=-\omega^2_{2,-n+1}=0$ which shows that
$\omega^2_{m,n}=0$ for all $m,n$ and concludes the proof.
\end{proof}

\section{Solution space}
\label{sec:BMS4solutionspace}

We start by assuming only that we have a metric of the form
\eqref{eq:bms4sachsmetric} and that the determinant condition holds. 
Following again \cite{Tamburino:1966ys}, the equations of motion are
organized in terms of the Einstein tensor
$G_{\alpha\beta}=R_{\alpha\beta}-\half g_{\alpha\beta} R$ as
\begin{gather}
  G_{r\alpha}=0,  \qquad
G_{AB}-\half g_{AB} g^{CD}G_{CD}=0,  \label{eq:BMS4EOM1}\\
G_{uu}=0=G_{uA},  \label{eq:BMS4EOM2}\\
g^{CD}G_{CD}=0.  \label{eq:BMS4EOM3} 
\end{gather}
Due to the form of the metric and the determinant condition, equation
\eqref{eq:BMS4EOM3} is a consequence of \eqref{eq:BMS4EOM1} on
account of the Bianchi identities. Indeed, the latter can be written
as
\begin{equation}
\label{eq:56}
 0= 2\sqrt{-g}{G_\alpha^\beta}_{;\beta}=2(\sqrt{-g}G_\alpha^\beta)_{,\beta}+
\sqrt{-g}G_{\beta\gamma}{g^{\beta\gamma}}_{,\alpha}. 
\end{equation}
When \eqref{eq:BMS4EOM1} hold and $\alpha=1$, we get
$G_{AB}{g^{AB}}_{,r}=0=\half g_{AB}{g^{AB}}_{,r} g^{CD}G_{CD}$. This
implies that \eqref{eq:BMS4EOM3} holds by using \eqref{eq:BMS4detcond}. 

The remaining Bianchi identities then reduce to
$2(\sqrt{-g}G_A^\beta)_{,\beta}=0=2(\sqrt{-g}G_u^\beta)_{,\beta}$. The
first gives $(r^2 G_{uA})_{,r}=0$. This means that if $r^2G_{uA}=0$
for some constant $r$, it vanishes for all $r$. When $G_{uA}=0$, the
last Bianchi identity reduces to $(r^2G_{uu})_{,r}=0$, so that again,
$r^2G_{uu}=0$ everywhere if it vanishes for some fixed $r$.

Let $k_{AB}=\half g_{AB,r}$, $l_{AB}=\half g_{AB,u}$, $n_A=
\frac{1}{2} e^{-2\beta}g_{AB} U^B_{,r}$ with indices on these
variables and on $U^A$ lowered and raised with the 2 dimensional
metric $g_{AB}$ and its inverse.  Define $K^A_B$ through the relation
$k^A_B=\frac{1}{r}\delta^A_B+\frac{1}{r^2}K^A_B$.  In particular, the
determinant condition implies that $k=\frac{2}{r}$ and thus that
$K^D_D=0$.  Similarly, if
$l^D_B=\half\bar\gamma^{DA}\bar\gamma_{AB,u}+\frac{1}{r}L^D_B$, the
determinant condition implies in particular that $L^D_B$ is traceless,
$L^D_D=0$. Note that for a traceless $2\times 2$ matrix ${M^T}^A_B$,
we have ${M^T}^A_C{M^T}^C_B=\half {M^T}^C_D{M^T}^D_C \delta^A_B$.

For a metric of the form \eqref{eq:bms4sachsmetric}, we have 
\begin{equation*}
\begin{gathered}
\Gamma^\lambda_{rr}=\delta^\lambda_r2\beta_{,r},\quad
\Gamma^{u}_{\lambda r}=0,\quad
\Gamma^r_{Ar}=\beta_{,A}+ n_A,\quad \Gamma^A_{Br}=k^A_B,\\
\Gamma^u_{AB}= e^{-2\beta} k_{AB},\quad
\Gamma^A_{BC}= e^{-2\beta} U^A k_{BC}+{}^{(2)}\Gamma^A_{BC}, \\ \Gamma^A_{ur}=
-k^A_BU^B+e^{2\beta}(\d^A\beta-n^A),\quad 
\Gamma^u_{uA}=\beta_{,A}-n_A- e^{-2\beta} k_{AB}U^B,\\
\Gamma^r_{ur}=-\half (\d_r+2\beta_{,r})\frac{V}{r}-(\beta_{,A}+n_A) U^A,\\
\Gamma^A_{Bu}=l^A_B +\half {}^{(2)}D^A U_B-\half {}^{(2)}D_B U^A+U^A(\beta_{,B}-n_B)- e^{-2\beta}k_{BC}U^A
U^C,\\
\Gamma^{u}_{uu}=2\beta_{,u}+\half(\d_r+2\beta_{,r})
\frac{V}{r}+2U^A n_A+e^{-2\beta}k_{AB} U^AU^B,\\
\Gamma^r_{AB}= e^{-2\beta}(\half{}^{(2)}D_A U_B
+\half{}^{(2)}D_B U_A+l_{AB}+k_{AB}\frac{V}{r} ),\\
\Gamma^r_{uA}=-\big(\frac{V_{,A}}{2r}+\frac{V}{r}n_A+
e^{-2\beta}U^B[\frac{1}{2}{}^{(2)}D_A U_B+\frac{1}{2}{}^{(2)}D_B
U_A+l_{AB}+\frac{V}{r}k_{AB}]\big),\\
\Gamma^A_{uu}=2U^A\beta_{,u}+\half
U^A(\d_r+2\beta_{,r})\frac{V}{r}+2U^A n_B U^B+U^A k_{BC}e^{-2\beta}U^BU^C
\\-U^A_{,u}-2l^A_B U^B-\half
e^{2\beta}(\d^A+2\d^A\beta)\frac{V}{r}-\half {}^{(2)}D^A(U^CU_C),\\
\Gamma^r_{uu}=-\half(\d_u-2\beta_{,u})\frac{V}{r}+\half
\frac{V}{r}(\d_r+2\beta_{,r})\frac{V}{r}+\half
U^A(\d_A+2\beta_{,A})\frac{V}{r}+ 2\frac{V}{r} U^A n_A\\+\frac{V}{r}
e^{-2\beta}k_{AB} U^AU^B+e^{-2\beta}l_{AB}U^AU^B +e^{-2\beta} U^AU^B
{}^{(2)}D_A U_B. 
\end{gathered}
\end{equation*}

To write the equations of motion, we use that $|{}^{(4)} g|=e^{4\beta}
|{}^{(2)}g|$ and
\begin{equation*}
  R_{\mu\nu}=\big[\d_\alpha+(2\beta+\half\ln
  |{}^{(2)}g|)_{,\alpha}\big]\Gamma^\alpha_{\mu\nu}-\d_\mu\d_\nu(2\beta+\half\ln
  |{}^{(2)}g|) -\Gamma^\alpha_{\nu\beta}\Gamma^\beta_{\mu\alpha}. 
\end{equation*}
The equation $G_{rr}\equiv R_{rr}=0$ then becomes
\begin{equation}
  \d_r\beta=-\frac{1}{2r}+\frac{r}{4} k^A_B k^B_A=\frac{1}{4r^3} K^A_B K^B_A\iff
  \beta=-\int^\infty_r dr^\prime\frac{1}{4 {r^\prime}^3} K^A_B K^B_A\,.
\end{equation}
This equation determines $\beta$ uniquely in terms of $g_{AB}$ because
the fall-off condition \eqref{eq:BMS4falloff} excludes the arbitrary function of
$u,x^A$ allowed by the general solution to this equation. 

The equations $G_{rA}\equiv R_{rA}=0$ read
\begin{equation}
\begin{gathered}
  \label{eq:54}
  \d_r(r^2n_A)=J_A, \\
  J_A=r^2(\d_r -\frac{2}{r})\beta_{,A}-{}^{(2)}D_B K^B_A
  =\d_A(-2r\beta+\frac{1}{4r}K^B_CK^C_B)-{}^{(2)}D_B K^B_A. 
\end{gathered}
\end{equation}
In the original approach \cite{Bondi:1962zr,Sachs:1962ly}, it was
assumed in particular that the metric $g_{AB}$ admits an expansion in
terms powers of $r^{-1}$ starting at order $r^2$. We will assume 
\begin{equation}
  \label{eq:BMS4expendg}
  g_{AB}=r^2\bar\gamma_{AB}+rC_{AB}+D_{AB}+\frac{1}{4} \bar\gamma_{AB}
C^C_DC^D_C+o(r^{-\epsilon}), 
\end{equation}
where indices on $C_{AB},D_{AB}$ are raised with the inverse of
$\bar\gamma_{AB}$.  In \cite{Tamburino:1966ys}, it was then shown
explicitly how \eqref{eq:BMS4expendg} is related to the conformal approach
\cite{Penrose:1963uq,Penrose:1965ve} and imposed through
differentiability conditions at null infinity.

Under the assumption \eqref{eq:BMS4expendg}, $C^D_D=0=D^C_C$ and
\begin{equation}
\begin{split}
  \label{eq:BMS4solution2}
  K^A_B&=-\half C^A_B- r^{-1}D^A_B+o(r^{-1-\epsilon}),\\
  \beta&=-\frac{1}{32}r^{-2}C^A_BC^B_A-\frac{1}{12}
    r^{-3} C^A_BD^B_A+o(r^{-3-\epsilon}),\\
  J_A&=\half \bar D_BC^B_A+ r^{-1} \bar
  D_BD^B_A +o(r^{-1-\epsilon}).
\end{split}
\end{equation}
These equations then imply
$n_A=\frac{1}{2}r^{-1}\bar D_BC^B_A+r^{-2} (\ln r \bar
D_BD^B_A+N_A)+ o(r^{-2-\epsilon})$ and involve the
arbitrary functions $N_A(u,x^B)$ as integration ``constants''.
Because $U^A$ has to vanish for $r\to\infty$, we get from the
definition of $n_A$
\begin{equation}
  \label{eq:BMS4solution3}
  U^A=-\frac{1}{2}r^{-2}\bar D_BC^{BA}-
\frac{2}{3}r^{-3}\Big[(\ln r+
  \frac{1}{3})\bar
D_BD^{BA}-\half C^{A}_{B} \bar D_CC^{CB}
  +N^A\Big]
  +o(r^{-3-\varepsilon}),
\end{equation}
where the index on $N_A$ has been raised with $\bar\gamma^{AB}$. 

It is straightforward to verify that if one trades the coordinate $r$
for $s=r^{-1}$, the only non vanishing components of the
``unphysical'' Weyl tensor at the boundary are given by 
\begin{equation}
\lim_{s\to
  0}(s^2W_{sAsB})=-D_{AB},\label{eq:BMS4Weyl}
\end{equation}
(see e.g.~\cite{Winicour:1985qf} for a
detailed discussion). In \cite{Sachs:1962ly}, the condition $D_{AB}=0$
was imposed in order to avoid a logarithmic $r$-dependence in the
solution to the equations of motion and to avoid singularities on the
unit sphere. When one dispenses with this latter restriction, absence
of a logarithmic $r$-dependence is guaranteed through the requirement
$\bar D_BD^{BA}=0$. In the coordinates $\zeta,\bar\zeta$ and with the
parametrization $\bar\gamma_{AB}dx^Adx^B=e^{2\tilde\varphi}d\zeta
d\bar\zeta$, this is equivalent to
\begin{equation}
  \label{eq:BMS4introd}
  D_{\zeta\zeta}=d(u,\zeta), \quad D_{\bar
  \zeta\bar \zeta}=\bar d(u,\bar \zeta), \quad 
D_{\zeta\bar \zeta}=0. 
\end{equation}
A more complete analysis of the field equations when allowing for a
logarithmic or, more precisely, a ``polyhomogeneous'' dependence in $r$
can be found in \cite{Chrusciel:1992bh}.

Starting from 
\begin{multline*}
  R_{AB}= (\d_r+2\beta_{,r}+\frac{2}{r})\Gamma^r_{AB}-
k^C_A\Gamma^r_{BC}-k^C_B\Gamma^r_{AC}
+{}^{(2)}R_{AB}-2{}^{(2)}D_B\beta_{,A}
 \\+(\d_u+2\beta_{,u}+l )\Gamma^u_{AB}-\Gamma^u_{uA}\Gamma^u_{uB}
-\Gamma^r_{rA}
  \Gamma^r_{rB}\\
-\Gamma^C_{uA}\Gamma^u_{BC}
  -\Gamma^C_{uB}\Gamma^u_{AC}+ {}^{(2)}D_C( e^{-2\beta}U^C
  k_{AB})\\-e^{-4\beta}U^C k_{BD}U^D
  k_{AC}+2
   e^{-2\beta}\beta_{,C}U^C k_{AB},
\end{multline*}
we find 
\begin{multline}
  \label{eq:BMS4solutionspace1}
  g^{DA}R_{AB}=e^{-2\beta}
\Big[(\d_r+\frac{2}{r})(l^D_B+k^D_B\frac{V}{r}+\half{}^{(2)}D_B
  U^D+\half {}^{(2)}D^D U_B)\\+k^D_A{}^{(2)}D_B
  U^A-k^A_B{}^{(2)}D_A U^D+(\d_u+l)k^D_B+{}^{(2)}D_C (U^C k^D_B)\Big]\\+{}^{(2)}R^D_{B}
  -2({}^{(2)}D_B\d^D\beta+\d^D\beta\d_B\beta +n^Dn_B). 
\end{multline}

When taking into account the previous equations, $G_{ur}\equiv
R_{ur}+\half e^{2\beta} R=0$ reduces to $g^{AB}R_{AB}=0$. Explicitly,
we find from the trace of \eqref{eq:BMS4solutionspace1}
\begin{equation}
\begin{gathered}
  \d_r V=J,\\
  J=e^{2\beta}r^2({}^{(2)}\Delta
  \beta+\d^D\beta\d_D\beta+n^D n_D-\half{}^{(2)}R)
-2rl -\frac{r^2}{2}(\d_r+\frac{4}{r}) {}^{(2)}
  D_B U^B\\
  =-2rl+
e^{2\beta}r^2\big[{}^{(2)}\Delta
  \beta+(n^A-\d^A\beta)(n_A-\d_A\beta)\\-{}^{(2)}D_A
  n^A-\half{}^{(2)}R\big] -2r {}^{(2)} D_B U^B
\\=-2rl-\frac{1}{2}\bar R+o(r^{-1-\epsilon}), 
\end{gathered}
\end{equation}
where we have used the previous equation to get the second line. 
This equation implies
\begin{equation}
  \label{eq:BMS4solution4}
  \frac{V}{r}  =-r l-\frac{1}{2} \bar
  R + r^{-1}2M + o(r^{-1-\epsilon}),
\end{equation}
and implies a third arbitrary function of $M(u,x^B)$ as integration
constant.

We have $G_{AB}-\half g_{AB} g^{CD} G_{CD}=R_{AB}-\half g_{AB} g^{CD}
R_{CD}$. Taking into account the previous equations, it thus reduces
to the condition that the traceless part of \eqref{eq:BMS4solutionspace1} vanishes.
Using that $\d_u k^D_B=\d_r l^D_B-2(l^D_Ak^A_B-k^D_A
l^A_B)$, we get 
\begin{multline*}
  (\d_r+\frac{1}{r}) l^D_B-
  (l^D_Ak^A_B-k^D_Al^A_B)+\frac{1}{2}k^D_B l =\\
-\half \Big[ (\d_r+\frac{2}{r})(k^D_B\frac{V}{r}+\half{}^{(2)}D_B
  U^D+\half {}^{(2)}D^D U_B)\\+k^D_C{}^{(2)}D_B
  U^C-k^C_B{}^{(2)}D_C U^D+{}^{(2)}D_C (U^C k^D_B)\Big]\\+e^{2\beta} \Big[ n^Dn_B+
  {}^{(2)}D_B\d^D\beta+\d^D\beta\d_B\beta -\half {}^{(2)}R^D_{B}\Big]. 
\end{multline*}
The various definitions then give 
\begin{equation}
  \label{eq:BMS4quantummech}
  \d_r L^D_B-\frac{1}{r^2}
  (L^D_AK^A_B-K^D_AL^A_B)=J^D_B,
\end{equation}
where
\begin{multline}
J^D_B=-\frac{r}{2} \Big[ (\d_r+\frac{2}{r})(k^D_B\frac{V}{r}+\half{}^{(2)}D_B
  U^D+\half {}^{(2)}D^D U_B)\\+k^D_C{}^{(2)}D_B
  U^C-k^C_B{}^{(2)}D_C U^D+{}^{(2)}D_C (U^C k^D_B)\Big]+\\+r e^{2\beta} \Big[ n^Dn_B+
  {}^{(2)}D_B\d^D\beta+\d^D\beta\d_B\beta-\half {}^{(2)}R^D_{B}\Big]-\frac{1}{2}
e  \bar\gamma^{DA}
  \bar\gamma_{AB,u}-\frac{r}{2}k^D_B l
  \\-\frac{1}{2r}
  (K^D_A\bar\gamma^{AC}\bar\gamma_{CB,u}-
\bar\gamma^{DC}\bar\gamma_{CA,u}
  K^A_B).
\end{multline}
The previous equations imply
\begin{multline*}
  J^D_B=-\half (\d_r k^D_B+\frac{1}{r}k^D_B)V-\frac{r^2}{2}k^D_Be^{2\beta}\big[{}^{(2)}\Delta
  \beta+(n^A-\d^A\beta)(n_A-\d_A\beta)-{}^{(2)}D_A n^A-\half{}^{(2)}R\big]\\
-\half({}^{(2)}D_B
  U^D+ {}^{(2)}D^D U_B)-r  U^C {}^{(2)}D_C k^D_B+\frac{r}{2} k^D_C({}^{(2)}D^C
  U_B- {}^{(2)}D_B U^C)\\+\frac{r}{2}k^C_B({}^{(2)}D_C U^D-{}^{(2)}D^D
  U_C)+\frac{r}{2} {}^{(2)}D_C U^Ck^D_B
\\+r e^{2\beta} \Big[ (n^D-\d^D\beta)(n_B-\d_B\beta)+
  {}^{(2)}D_B\d^D\beta-\half ({}^{(2)}R^D_{B}+{}^{(2)}D_B
  n^D+ {}^{(2)}D^D n_B)\Big]\\-\frac{1}{2}
  \bar\gamma^{DA}
  \bar\gamma_{AB,u}+\frac{r}{2}k^D_Bl
 -\frac{1}{2r}
  (K^D_A\bar\gamma^{AC}\bar\gamma_{CB,u}-
\bar\gamma^{DC}\bar\gamma_{CA,u}
  K^A_B).
\end{multline*}

Let $\cO^{DA}_{BC}=-\frac{1}{r^2} (K^D_A\delta^C_B-\delta^D_AK
^C_B)$ and
$\cA\cR$ denote anti-radial ordering. 
Equation \eqref{eq:BMS4quantummech} without right-hand side has the same form as
the Schr\"odinger equation with time dependent Hamiltonian. If we define
\begin{equation}
U^{DA}_{BC}(r_<,r_>)=\cA\cR\exp{[-\int^{r_>}_{r_{<}}
  dr^\prime \cO^{DA}_{BC}(r^\prime)]}\label{eq:76},
\end{equation}
the solution to the inhomogeneous equation \eqref{eq:BMS4quantummech} with
non-vanishing $J^B_D$ can then be obtained by variation of constants
and reads
\begin{equation}
  \label{eq:BMS4solL}
  L^D_B(r)=U^{DA}_{BC}(r,\infty)[\half N^C_A+\int dr^\prime
  U^{CE}_{AF}(\infty,r^\prime) J^F_E(r^\prime)], 
\end{equation}
and involves two more integration constants encoded in $N^D_B(u,x^B)$.

In other words, the $r$-dependence of $g_{AB,u}$ is completely
determined up to two integration constants. It follows that the only
variables left in the theory whose $r$-dependence is undetermined are
the two functions contained in $R_{AB}(u_0,r,x^C)=g_{AB}(u_0,r,x^C)
-r^2\bar\gamma_{AB}(u_0,x^C)-rC_{AB}(u_0,x^C)
-D_{AB}(u_0,x^C)-\frac{1}{4}\bar\gamma_{AB} C^C_DC^D_C$ at some
  initial fixed $u_0$.

When expanding into orders in $r$, one finds in particular
\begin{equation*}
\begin{split}
  L^D_B&=\half(\bar\gamma^{DA}C_{AB,u}-
  C^{DA}\bar\gamma_{AB,u} )+\half
  r^{-1}\Big[\bar\gamma^{DA}\d_u (D_{AB}+\frac{1}{4}\bar\gamma_{AB}C^C_DC^D_C)
\\&  -C^{DA}C_{AB,u}
  -D^{DA}\bar\gamma_{AB,u}
  +\frac{1}{4}C^E_FC^F_E
  \bar\gamma^{DA}\bar\gamma_{AB,u}\Big]+o(r^{-1-\epsilon}),
  \\
  J^D_B&= \half \delta^D_Bl-\half\bar\gamma^{DA}\bar\gamma_{AB,u} +
  \frac{1}{4}r^{-1}[C^{DA}\bar\gamma_{AB,u}
  -\bar\gamma^{DC}\bar\gamma_{CA,u} C^A_B]\\&+\half
  r^{-2}[ lD^D_B+
  D^{DA}\bar\gamma_{AB,u}-\bar\gamma^{DC}
  \bar\gamma_{CA,u}D^A_B] +o(r^{-2+\epsilon}).
\end{split}
\end{equation*}
When injecting into the equation of motion \eqref{eq:BMS4quantummech}, the leading
order requires that 
\begin{equation}
  \label{eq:BMS4dug}
  \bar\gamma_{AB,u}=l\bar\gamma_{AB}, 
\end{equation}
or, in other words, that the only $u$ dependence in $\bar\gamma_{AB}$
is contained in the conformal factor. This agrees with the assumption
of section~\bref{sec:asympt-flat-4}, where the $u$-dependence of
$\bar\gamma_{AB}$ was contained in $\exp{2\varphi}$ and 
$l=2\d_u\varphi$, and also with the
discussion at the end of the previous subsection, where it was
contained in $\exp{2\tilde\varphi}$ and $l=2\d_u\tilde\varphi$. In
the following we always assume that 
\eqref{eq:BMS4dug} holds. In particular, this implies
\begin{equation*}
\begin{split}
  L^D_B&=\half(\bar\gamma^{DA}C_{AB,u}- lC^D_B
  )+\half r^{-1}[\bar\gamma^{DA}D_{AB,u}
  -C^{DA}C_{AB,u} \\ &-lD^D_B
  +\frac{1}{2}C^{EF}\d_uC_{EF}
  \delta^D_B]+o(r^{-1-\epsilon}),
  \\
  J^D_B&= \half r^{-2} lD^D_B+o(r^{-2+\epsilon}).
\end{split}
\end{equation*}
When taking into account the next order of \eqref{eq:BMS4quantummech} and comparing
to the general solution \eqref{eq:BMS4solL}, we get 
\begin{equation}
  \label{eq:BMS4defNAB}
  \d_u D_{AB}=0, \qquad 
N_{AB}=\d_u C_{AB}-C_{AB} l,
\end{equation}
where the index on $N^A_B$ has been lowered with
$\bar\gamma_{AC}$. This implies in turn that 
\begin{equation*}
  l^A_B=\half l\delta^A_B+\half r^{-1} N^A_B-\frac{1}{4} r^{-2}[
 C^{A}_{C}N^{C}_{B}- N^A_CC^C_B
+2lD^A_B] +o(r^{-2-\epsilon}). 
\end{equation*}

At this stage, equations \eqref{eq:BMS4EOM1} have been solved, and then
\eqref{eq:BMS4EOM3} holds automatically on account of the Bianchi
identities. Furthermore $g^{CD}G_{CD}=0$ reduces to $R_{ur}=0$ and we
also have $R=0$.  Under these assumptions, we only need to discuss the
$r$ independent part of $r^2G_{uA}=0$ and then of $r^2G_{uu}=0$, which
reduce to $r^2 R_{uA}=0$ and $r^2 R_{uu}=0$, respectively. The
$r$-independent part fixes the $u$ dependence of $N_A$ and $M$ in
terms of the other fields. Explicitly,
\begin{multline*}
  R_{uA}=(-\d_u+l)\beta_{,A} -\d_A l-(\d_u+l)n_A  
 \\ +n_B {}^{(2)}D^B U_A- \beta_{,B}  {}^{(2)}D_A
  U^B  +2U^B(\beta_{,B}\beta_{,A}+n_Bn_A)
\\+ {}^{(2)}D_B
  \Big[l^B_A+\half  {}^{(2)}D^B U_A
-\half  {}^{(2)}D_A
  U^B+U^B(\beta_{,A}-n_A)\Big]+2n_Bl^B_A
\\  -(\d_r+2\beta_{,r}+\frac{2}{r}
  )(\frac{V_{,A}}{2r})-\frac{V}{r}(\d_r+\frac{2}{r}
  )  n_A+k_A^B(\frac{V_{,B}}{r}+2\frac{V}{r} n_B )
\\-e^{-2\beta} (\d_r+\frac{2}{r})\big[U^B(\half
   {}^{(2)}D_A U_B+\half  {}^{(2)}D_B U_A+l_{AB}+\frac{V}{r}k_{AB})\big]
\\
-e^{-2\beta}U^B\Big[(\d_u+l)k_{AB}-2l_A^Ck_{CB}-2k_A^Cl_{CB}
-2k_A^Ck_{CB}\frac{V}{r}\\+ {}^{(2)}D_C(k_{AB}
  U^C)-k_{AC} {}^{(2)}D^C U_B-k_{BC} {}^{(2)}D^C U_A\Big],
\end{multline*}
and the term proportional to $r^{-2}$ yields
\begin{multline}
  \label{eq:BMS4duNA}
(\d_u+l)N_A
 =\d_AM+\frac{1}{4}C_A^B\d_B\bar
  R +\frac{1}{16}\d_A\big[N^B_C C^C_B\big]
-\frac{1}{4} \bar D_AC^C_BN^B_C\\
-\frac{1}{4}\bar D_B\big[C^{B}_{C}N^C_{A}-N^B_CC^C_A\big]-\frac{1}{4}
  \bar D_B \big[ \bar D^B \bar D_CC^C_A -\bar D_A \bar
  D_CC^{BC}\big]\\-\frac{1}{32}
l\d_A\big[C^B_CC^C_B\big]
 +\frac{1}{16}\d_A lC^B_CC^C_B 
+\half \bar D_B\big[lD^B_A\big].
\end{multline}
Similarly, 
\begin{multline*}
   R_{uu}=(\d_u+2\beta_{,u}+l)\Gamma^u_{uu}+(\d_r+2\beta_{,r}+
\frac{2}{r})\Gamma^r_{uu} +
(\d_A+2\beta_{,A}+{}^{(2)}\Gamma^{B}_{BA})\Gamma^A_{uu} \\-2\beta_{,uu}-\d_u l 
-(\Gamma^{u}_{uu})^2-2\Gamma^{u}_{uA}\Gamma^A_{uu}-(\Gamma^r_{ur})^2 
-2\Gamma^r_{uA}\Gamma^A_{ur}-\Gamma^A_{uB}\Gamma^B_{uA},
\end{multline*}
and the term proportional to $r^{-2}$ yields
\begin{multline}
  \label{eq:BMS4duM}
  (\d_u+\frac{3}{2}l) M=-\frac{1}{8} N^A_BN^B_A
  -\frac{1}{8}lC^A_BN^B_A-\frac{1}{32} l^2 C^A_BC^B_A +\frac{1}{8}
  \bar \Delta \bar R \\
  +\frac{1}{4}\bar D_A\bar D_C N^{CA} +\frac{1}{8} l \bar D_A \bar D_C
  C^{CA}+\frac{1}{4}\bar D_C l\bar D_A C^{CA}.
\end{multline}
All these considerations can be summarized as follows:

{\em For a metric of the form \eqref{eq:bms4sachsmetric} satisfying the determinant
  condition and with $g_{AB}$ as in \eqref{eq:BMS4expendg}, the general
  solution to Einstein's equations is parametrized by the $2$
  dimensional background metric $\bar\gamma_{AB}(u,x^C)$ satisfying
  \eqref{eq:BMS4dug}, by the mass and angular momentum aspects
  $M(u,x^A),N_A(u,x^B)$ satisfying \eqref{eq:BMS4duM},\eqref{eq:BMS4duNA}, by
  the traceless symmetric news tensor $N_{AB}(u,x^C)$ defined in
  \eqref{eq:BMS4defNAB}, and by the traceless symmetric tensors $D_{AB}(x^C)$,
  $C_{AB}(u_0,r,x^C)$, $R_{AB}(u_0,r,x^C)$.

  For such spacetimes, the only non vanishing components of the
  unphysical Weyl tensor at the boundary are given by \eqref{eq:BMS4Weyl}.
  When logarithmic terms are required to be absent in the metric,
  $D_{AB}(x^C)$ has to satisfy $\bar D_BD^B_A=0$. In the coordinates
  $\zeta,\bar\zeta$ and the parametrization
  $\bar\gamma_{AB}dx^Adx^B=e^{2\tilde\varphi} d\zeta d\bar \zeta$,
  this leads to \eqref{eq:BMS4introd} with $d=d(\zeta)$ and $\bar d=\bar d(\bar
\zeta)$ by also taking \eqref{eq:BMS4defNAB} into account.}

In particular, let us now use the parametrization
$\bar\gamma_{AB}dx^Adx^B=e^{2\tilde\varphi}d\zeta d\bar\zeta$. The
determinant condition then reads $\text{det}\, g_{AB}=-e^{4\tilde
  \varphi}\frac{r^4}{4}$.  Even though we will not use it explicitly
below, let us point out that the determinant condition can be
implemented for instance by choosing the Beltrami representation,
\begin{equation*}
\begin{gathered}
    h=\frac{g_{\zeta\zeta}}{g_{\zeta\bar\zeta}+f},\quad 
  \bar h=\frac{g_{\bar\zeta\bar\zeta}}{g_{\zeta\bar\zeta}+f}, \\
  g_{\zeta\zeta}=\frac{2f h}{1-y},\quad g_{\bar\zeta\bar\zeta}=
\frac{2f \bar h}{1-y}, \quad g_{\zeta\bar\zeta}=\frac{f(1+y)}{1-y},\\
  g^{\zeta\zeta}=-\frac{2 \bar h}{f(1-y)},\quad
  g^{\bar\zeta\bar\zeta}=-\frac{2 h}{f(1-y)},\quad
  g^{\zeta\bar\zeta}=\frac{1+y}{f(1-y)},
\end{gathered}
\end{equation*}
where $f=\sqrt{-{}^{(2)}g}$, $y=h\bar h$, with
$f=\frac{r^2}{2}e^{2\tilde\varphi}$ fixed, while $h=O(r^{-1})=\bar
h$. Alternatively, one can choose
\begin{equation*}
\begin{gathered}
   g_{\zeta\zeta}=f e^{i\alpha}\sinh\rho ,\quad g_{\bar\zeta\bar\zeta}=
f e^{-i\alpha}\sinh\rho, \quad g_{\zeta\bar\zeta}=f \cosh\rho,\\
  g^{\zeta\zeta}=-f^{-1}e^{-i\alpha}\sinh\rho, \quad
  g^{\bar\zeta\bar\zeta}=-f^{-1}e^{i\alpha}\sinh\rho ,\quad
  g^{\zeta\bar\zeta}=f^{-1}\cosh\rho,
\end{gathered}
\end{equation*}
where $\rho=O(r^{-1})$ and $\alpha=O(r^0)$.  

In the parametrization with the conformal factor introduced with
respect to the Riemann sphere, we can write
\begin{equation}
  \label{eq:BMS4compC}
\begin{gathered}
  C_{\zeta\zeta}=e^{2\tilde\varphi} c,\quad
  C_{\bar\zeta\bar\zeta}=e^{2\tilde\varphi} \bar c,\quad
  C_{\zeta\bar\zeta}=0,\\
D_{\zeta\zeta}= d,\quad
  D_{\bar\zeta\bar\zeta}= \bar d,\quad
  D_{\zeta\bar\zeta}=0.
\end{gathered}
\end{equation}
Equations \eqref{eq:BMS4solution2}, \eqref{eq:BMS4solution3} and \eqref{eq:BMS4solution4} read 
\begin{equation}
\begin{split}
  \beta&=-\frac{1}{4}r^{-2}c\bar c -\frac{1}{3}r^{-3}e^{-2\tilde\varphi}
  (d\bar  c+\bar d c) +o(r^{-3-\epsilon}), \\
  U^\zeta&=-\frac{2}{r^{2}}e^{-4\tilde\varphi}\d(e^{2\tilde\varphi}\bar
  c)-\\&\hspace{2cm}-\frac{2}{3r^3}\Big[(\ln r+\frac{1}{3})4e^{-4\tilde\varphi}\d
  \bar d -4e^{-4\tilde\varphi}\bar c
\bar \d(e^{2\tilde\varphi} c)+ \cN^\zeta\Big]+o(r^{-3-\epsilon}),\\
  \frac{V}{r}&=-2r\d_u\tilde\varphi + 4e^{-2\tilde\varphi}\d\bar\d
  \tilde\varphi + r^{-1}2M +o(r^{-1-\epsilon}).
\end{split}
\end{equation}
and the evolution equations become
\begin{equation}
\begin{gathered}
  \d_u (e^{3\tilde\varphi} M)=\d_u \Big(e^{\tilde\varphi}\big[\d^2\bar
  c+\bar \d^2 c+2\d\tilde\varphi \d\bar c+2\bar\d\tilde\varphi\bar \d
  c+2\d^2\tilde\varphi\bar c+2\bar\d^2\tilde\varphi
  c\big]\Big)\\
  -e^{\tilde\varphi}\d_u(e^{\tilde\varphi}c)\d_u(e^{\tilde\varphi}\bar
  c)+2e^{\tilde\varphi}\Big( c\big[\d_u(\bar\d \tilde\varphi)^2
  -\d_u\bar\d^2\tilde\varphi\big] +\bar
  c\big[\d_u(\d\tilde\varphi)^2-\d_u\d^2\tilde\varphi\big]\Big)
  \\+e^{-\tilde\varphi}\Big(-4(\d\bar\d)^2\tilde\varphi+8\big[
  (\d\bar\d \tilde\varphi)^2+\d\tilde\varphi\d\bar\d^2\tilde\varphi+
  \bar\d\tilde\varphi\d^2\bar\d\tilde\varphi -2\bar\d\tilde\varphi
  \d\tilde\varphi\d\bar\d\tilde\varphi\big]\Big),
  \\
  \d_u (e^{2\tilde\varphi} N_{\bar\zeta})=e^{2\tilde\varphi}\Big[\bar
  \d M +\frac{1}{4}\big[(\bar\d \bar c+5\bar c\bar\d)\d_u c-(3
  c\bar\d+7\bar\d c) \d_u\bar c \big] +2\bar\d\tilde\varphi(\bar c\d_u
  c - c\d_u\bar c) \\-\half\d_u\tilde\varphi \bar\d (c\bar
  c)+\bar\d\d_u\tilde\varphi c\bar c \Big]
  +2(\d\d_u\tilde\varphi +\d_u\tilde\varphi\d)\bar d \\
  +\bar \d^3c+2\bar \d^3\tilde\varphi c +4\bar \d^2\tilde\varphi\bar
  \d c-4\bar\d \tilde\varphi\bar\d^2\tilde\varphi c
  -4(\bar\d\tilde\varphi)^2\bar \d c\\
  -\d^2\bar\d\bar c -2(\d\tilde\varphi \d+\d^2\tilde\varphi)\bar\d\bar
  c-2(\d\bar\d\tilde\varphi-\bar
  \d\tilde\varphi\d-2\bar\d\tilde\varphi \d\tilde\varphi ) \d\bar c
  \\-2(\d^2\bar\d\tilde\varphi+2\d\bar\d^2\tilde\varphi
  -2\bar\d\tilde\varphi\d^2\tilde\varphi
  -4\bar\d\tilde\varphi\d\bar\d\tilde\varphi )\bar c .
\end{gathered}
\end{equation}

Let us now set $\tilde\varphi=0$. Note that one can re-introduce an
arbitrary $\tilde\varphi$ through the finite coordinate transformation
generated by $\xi^u=-u\tilde\varphi$, $\xi^A=-\xi^u_{,B}\int ^\infty_r
dr^\prime (e^{2\beta} g^{AB})$, $\xi^r=-\half r(\d_A \xi^A
-2\tilde\varphi -f_{,B} U^B)$. The above relations then simplify to
\begin{equation}
\begin{split}
  \label{eq:30til}
  \beta&=-\frac{1}{4}r^{-2}c\bar c -\frac{1}{3}r^{-3} (d\bar
  c+\bar d c) +o(r^{-3-\epsilon}), \\
  U^\zeta&=-2r^{-2}\d \bar c-\frac{2}{3}r^{-3}\Big[(\ln
  r+\frac{1}{3})4\d\bar d-4 \bar c \bar \d c+ N^\zeta\Big]
  +o(r^{-3-\epsilon}),\\
  \frac{V}{r}&= r^{-1}2 M+o(r^{-1-\epsilon}),\\
  \d_u M&=\big[\d^2\dot{\bar c}+\bar \d^2\dot c\big] -\dot c\dot{\bar c},
  \\
  \d_u N_{\bar\zeta}&=
\bar \d M +\frac{1}{4}\big[(\bar\d \bar
  c+5\bar c\bar\d)\dot c-(3 c\bar\d+7\bar\d c) \dot{\bar c} \big]
+\bar \d^3c-\d^2\bar\d\bar c.
\end{split}
\end{equation}
When defining $\tilde M=M-\bar\d^2 c- \d^2 \bar c$ and $\tilde
N^\zeta=-\frac{1}{12}[2 N^\zeta+7\bar c \bar \d c+3 c\bar \d \bar c]$, the
evolution equations become
\begin{equation}
\begin{split}
  \label{eq:30q}
  \d_u\tilde M&=-\dot c\dot {\bar c},
  \\
  3\d_u \tilde N^\zeta&=-\bar\d \tilde M-2\bar \d^3  c-(\bar\d \bar c+3
  \bar c\bar \d )\dot{ c}.
\end{split}
\end{equation}

\section{Realization of $\mathfrak{bms}_4$ on solution space }
\label{sec:mass-asympt-flat}

In order to compute how $\mathfrak{bms}_4$ is realized on solution
space we need to compute the Lie derivative of the metric on-shell. We
will do so for the extended transformations defined by
\eqref{eq:bms4symasptphi1}-\eqref{eq:bms4symasptphi2} and use $-\delta
\bar\gamma_{AB}=2\omega\bar \gamma_{AB}$. Let
$\tilde\psi=\psi-2\omega$. This gives
\begin{equation}
  \label{eq:BMS4varC}
  -\delta C_{AB}=[f\d_u+\cL_Y -\half(\tilde \psi+fl)] C_{AB}-2\bar D_A\bar D_B f+\bar \Delta f\bar\gamma_{AB},
\end{equation}
where \eqref{eq:BMS4defNAB} should be used to eliminate $\d_u C_{AB}$ in favor
of $N_{AB}$ and
\begin{equation}
  \label{eq:BMS4varD}
  -\delta D_{AB}=\cL_Y D_{AB}, 
\end{equation}
where we have used that 
\begin{equation*}
  \begin{split}
    \bar D_A\bar D_C f C^C_B+\bar D_B\bar D_C f C^C_A-
    \bar\gamma_{AB}\bar D_C\bar D_C f C^{CD}-\bar\Delta f C_{AB}=0,\\
    \bar D_A f \bar D_C C^C_B+\bar D_B f \bar D_C C^C_A+ \bar D_C f
    \bar D_A C^C_B+ \bar D_C f \bar D_B C^C_A-\\- 2\bar D^C f\bar D_C
    C_{AB}-2\bar\gamma_{AB} \bar D_C f\bar D_D C^{CD} =0,
  \end{split}
\end{equation*}
which can be explicitly checked in the parametrization
$\bar\gamma_{AB}dx^Adx^B=e^{2\tilde\varphi} d\zeta d\bar \zeta$ with
$C_{AB}$ defined in \eqref{eq:BMS4compC}.  By taking the time derivative of
\eqref{eq:BMS4varC} and using \eqref{eq:BMS4defNAB}, \eqref{eq:bms4eqf} with $\psi$
replaced by $\tilde \psi$, one finds the transformation law for the
news tensor,
\begin{multline}
  -\delta N_{AB}= [f\d_u + \cL_Y] N_{AB}-(\bar D_A\bar D_B
  \tilde \psi-\half \bar \Delta\tilde \psi \bar\gamma_{AB})\\ 
+\frac{1}{4} (2 f\dot l
  +  f l^2+ \tilde \psi l -4\dot\omega+2
  Y^C\bar D_C l) C_{AB}\\+l (\bar D_A\bar D_B f-\half \bar \Delta
  f\bar\gamma_{AB}) -f(\bar D_A\bar D_B l -\half \bar\Delta l
  \bar\gamma_{AB}).
\end{multline}

We have $g_{uA}=\half \bar D_B C^B_A+\frac{2}{3} r^{-1} \big[(\ln
r+\frac{1}{3}) \bar D_B D^B_A+\frac{1}{4} C_{A}^B \bar D_C
C^{C}_B+N_A\big]+o(r^{-1-\epsilon})$, and by computing $\cL_\xi
g_{uA}$ on-shell, we find to leading oder that  
$-\delta (\bar D_B C^B_A)=[f\d_u+\cL_Y+\half(lf+\tilde\psi)] \bar D_B
C^B_A-\half \d_B (lf+\tilde\psi) C^B_A+\d_C
f(N^C_A+lC^C_A)-\d_A(\bar\Delta f)-\d_A f \bar R$. This is consistent
with \eqref{eq:BMS4varC} by using the generalization of \eqref{eq:BMS4confvect} which reads
\begin{equation}
  \label{eq:35a}
  2\bar D_B\bar  D_C Y_A=\bar  \gamma_{CA}\bar  D_B\psi+\bar \gamma_{AB}
\bar D_C\psi-\bar\gamma_{BC}\bar  D_A \psi+\bar R
  Y_C \bar \gamma_{BA}-\bar R Y_A \bar \gamma_{BC},  
\end{equation}
and implies $\bar \Delta Y^A=-\half \bar RY^A$,$\bar \Delta\psi = -
\bar D_A(\bar R Y^A)$. The logarithmic term gives 
$-\delta (\bar D_B D^B_A)=(f\d_u+\cL_Y+lf+\tilde\psi)\bar D_B D^B_A$,
which is again consistent with \eqref{eq:BMS4varD}, while the $r^{-1}$
terms, when combined with the previous transformations, give
\begin{multline}
  \label{eq:33}
  -\delta N_A=[f\d_u +\cL_Y+\tilde\psi+fl]N_A -\half [ f\bar D_B l
  +\bar D_B \tilde \psi +(\tilde\psi+lf)
  \bar  D_B ] D^B_A\\
  +3\bar D_A f M -\frac{3}{16}\bar D_A f
  N^B_CC^C_B
+\half \bar D_B f N^B_C C^C_A
+\frac{1}{32}(  \bar D_A fl-f\bar D_A l-\bar D_A\tilde\psi)
  (C^B_CC^C_B)\\+ \frac{1}{4} (\bar D_Bf \bar R+\bar D_B \bar\Delta
  f)C^B_A -\frac{3}{4}\bar D_B f(\bar D^B\bar D_C C^C_A-\bar
  D_A\bar D_C C^{BC})
  \\+\half (\bar D_A\bar D_B f
-\frac{1}{2}\bar \Delta f\bar\gamma_{AB} )\bar D_C C^{CB}+\frac{3}{8} \bar
  D_A(\bar D_C\bar D_B f C^{CB}) . 
\end{multline}
Here $\d_u N_A$ should be eliminated by using 
\eqref{eq:BMS4duNA}. 
In the same way, from the order $r^{-1}$ of $\cL_\xi g_{uu}$, we get 
\begin{multline}
  -\delta M =[f\d_u+Y^A\d_A+\frac{3}{2} (\tilde\psi+
  fl)]M\\+\frac{1}{4}\d_u[
  \bar D_C \bar D_B f C^{CB} +2\bar D_B f\bar D_C C^{CB}]
+\frac{1}{4}[\bar D_A f l-f\bar D_A l-\bar D_A \tilde\psi]\bar D_B
C^{BA}
\\
+\frac{1}{4}\d_A
  f(\d^A\bar R-C^{AB}\bar D_B l)+\frac{1}{4}
  l[\bar D_C \bar D_B f C^{CB} +\bar D_B f\bar D_C C^{CB}], 
\end{multline}
where $\d_u M$ should be replaced by its expression from
\eqref{eq:BMS4duM}. 

Let us now discuss these transformations in the parametrization
$\zeta,\bar\zeta$ with $\tilde\varphi=0=\omega$ so that
$\bar\gamma_{AB}dx^Adx^B=d\zeta d\bar\zeta$. From the leading and
subleading orders of $\cL_\xi g_{\zeta\zeta},\cL_\xi g_{\bar \zeta\bar
  \zeta}$, we get
\begin{equation}
  \label{eq:50}
  \begin{gathered}
-\delta c=f\dot c+Y^A\d_A c+(\frac{3}{2} \d  Y-\half  \bar \d 
\bar Y)c-2 \d^2 f  ,\\
-\delta d=Y^A\d_A  d+2\d Y d,
  \end{gathered}
\end{equation}
with $f$ given in \eqref{eq:BMS4fsimple} and 
the complex conjugate relation holding for $\bar c, \bar d$. 
In particular, for the news function we find 
\begin{equation}
  -\delta\dot c=f\ddot c+Y^A\d_A \dot c+2 \d  Y\dot c-
  \d^3  Y  ,
\end{equation}
From the subleading term of
$\cL_\xi g^{r\zeta}$ and the leading term of $\cL_\xi g_{uu}$ and  we get
\begin{multline}
  -\delta \tilde \cN^\zeta=Y^A \d_A \tilde \cN^\zeta
  + (\d Y + 2\bar \d \bar Y )\tilde \cN^\zeta +\frac{1}{3} \d^2 Y \bar d\\
- \bar
  \d f (\tilde M +2
  \bar \d^2  c + \bar c \dot{ c} )-\frac{f}{3}\big[\bar\d \tilde M+2\bar \d^3
   c+(\bar\d  \bar c+3 \bar c\bar \d)\dot{ c}\big], 
\end{multline}
\begin{equation}
  -\delta \tilde M  = -f
   \dot c\dot{\bar c}+  Y^A \d_A \tilde M + \frac{3}{2}
  \psi \tilde M +\bar c \d^3 Y+  c \bar \d^3 \bar Y 
    +  4 \d^2 \bar \d^2 \tilde T. 
\end{equation}

As can be understood by comparing with the $3$ dimensional anti-de
Sitter and flat cases, this computation already contains information
on the central extensions in the surface charge algebra through the
inhomogeneous part of the transformation laws for the fields. Although
we know that $BMS_4$ is free of central extensions in general, in our
case we could have field dependent extensions. Signs of this kind of
extension are present in the variation of $M$ for instance where we
can see a Schwarzian derivative multiplied by the field $c$.

\section{$\mathfrak{bms}_4$ charges}

In chapter \ref{chap:ads3} and  \ref{chap:bms3}, we were able to integrate the expression for
the charges and use equation (\ref{eq:applinearcharges}) to compute
them. The charges of $BMS_4$ are non-integrable; we have to start from expression (\ref{eq:gravsurfacecharge}) with
$h=\delta g$. We mean by $\delta g$ the
variation of the metric $g$ generated by a small variation of the
parameters characterizing it asymptotically:
\begin{equation}
\cX^\Gamma\equiv \{C_{AB},N_{AB},D_{AB},M,N_A\}.
\end{equation}
In this section, we will consider metrics with $\d_u \varphi=0$. Our integration two sphere is
given by $u=u_0$ constant and $r=cst \to \infty$. We will also use the
notation $\int d^2 \Omega^\varphi = \int dx^2 dx^3 \sqrt{\bar \gamma}
= \int_0^\pi d \theta\int_0^{2\pi} d\phi \sin{\theta} e^{2\varphi}$.

\vspace{5mm}

We start by writing
\begin{multline}
\ndelta  \cQ_\xi[h, g]= \frac{1}{16 \pi G}\lim_{r \to \infty} \int d^2
  \Omega^\varphi \ r^2 e^{2 \beta}\Big[\xi^r( D^u h- D_\sigma
  h^{u\sigma} + D^r h^u_r- D^u h^r_r)\\-\xi^u( D^r h-
  D_\sigma h^{r\sigma}- D^r h^u_u+ D^u h^r_u)+\xi^A( D^r
  h^{u}_A- D^u h^{r}_A) +\frac{1}{2}h(
  D^r\xi^u- D^u\xi^r)\\+\half h^{r\sigma}( D^u\xi_\sigma-
  D_\sigma\xi^u)-\half  h^{u\sigma}( D^r\xi_\sigma-
  D_\sigma\xi^r)\Big].\label{eq:intch} 
\end{multline}

We have 
\begin{eqnarray}
 D^u h- D_\sigma
  h^{u\sigma} + D^r h^u_r- D^u h^r_r
 & = & g^{ur}g^{AB}( D_r h_{AB}-D_A h_{rB})\nonumber\\
&=& - e^{-2 \beta} \left( g^{AB} \partial_r h_{AB} -
  k^{AB}h_{AB}+ e^{-2\beta} g^{AB} k_{AB} h_{ru}\right) \nonumber \\
&=& \frac{1}{4r^3} C^{AB} \delta C_{AB} + o(r^{-3-\epsilon})\,,
\end{eqnarray} 
\begin{eqnarray}
  -\left( D^r h- D_\sigma h^{r\sigma}- D^r h^u_u+ D^u h^r_u\right)&=&
  D^A h^r_A- D^r h^A_A\nonumber\\
  &= &g^{ur}g^{AB}( D_A h_{uB}- D_u h_{AB}) + O(r^{-3})\nonumber\\
  &=&g^{ur}\left( g^{AB}\, {}^{(2)}D_B h_{uA} - h_{ur} g^{AB}(l_{AB} +
    k_{AB}\frac{V}{r})\right. \nonumber \\ && \quad - k h_{uu} - g^{AB} \partial_u h_{AB} + g^{AB}
  h_{CA} l^C_B\Big)+ O(r^{-3})\nonumber \\ 
  &=&\frac{1}{r^2} \left(4 \delta M - \frac{1}{2} \bar D_A\bar D_B \delta C^{AB} +
    \frac{1}{2} \delta \partial_u (C^{AB}C_{AB}) \right. \nonumber \\ &&
  \left. - \frac{1}{2} \partial_u
    C_{AB} \delta C^{AB} - C^{AB} \partial_u \delta C_{AB} \right) + 
o(r^{-2-\epsilon})\,,\nonumber\\
\end{eqnarray}
\begin{eqnarray}
  D^r h^u_A- D^u h^r_A&=& (g^{ur})^2\left( \Gamma^C_{rA}h_{uC} - 
\partial_r h_{Au}\right) + g^{ur} g^{rB}
  \left( \Gamma^C_{rB}h_{AC} - \partial_r h_{AB}\right) + O(r^{-3})\nonumber \\
  &=& \frac{1}{2 r} \bar D_B \delta C^B_A + \frac{2}{3r^2} (2 \text{ ln }r -
  \frac{1}{3}) \bar D_B \delta D^B_A  + o(r^{-2-\epsilon}) \nonumber\\ && \quad + \frac{1}{r^2}
  \left( \frac{4}{3}\delta N_A + \frac{1}{3} \delta (C_{AB} \bar D_C C^{BC}) -
    \frac{1}{4} C_{AB} \bar D_C \delta C^{BC}\right)
\end{eqnarray}
\begin{multline}
\half h^{r\sigma}( D^u\xi_\sigma-
  D_\sigma\xi^u)-\half  h^{u\sigma}( D^r\xi_\sigma-
  D_\sigma\xi^r)=\nonumber\\
\frac{1}{2} \left(h^u_u + h^r_r \right)\left(D^u \xi^r - D^r \xi^u \right)+ 
\frac{1}{2} h^r_A\left(D^u \xi^A - D^A \xi^u \right). \nonumber
\end{multline}
\begin{equation}
\frac{1}{2}(h-h^u_u - h^r_r )(
  D^r\xi^u- D^u\xi^r)  = \frac{1}{2}g^{AB}h_{AB}(
  D^r\xi^u- D^u\xi^r) = 0
\end{equation}
\begin{eqnarray}
D^u \xi^A - D^A \xi^u & =& g^{ur} \partial_r \xi^A - g^{AB} \partial_B
\xi^u + (g^{ur} \Gamma^A_{rC} - g^{AB}\Gamma^u_{BC})\xi^C+O(r^{-3})\nonumber \\
&=& \frac{-2}{r} Y^A + \frac{1}{r^2} C^A_CY^C+O(r^{-3})
\end{eqnarray}
\begin{multline}
\frac{1}{2}h^r_A  =  -\frac{1}{4} \bar D_B \delta C^B_A - \frac{1}{3r}
(\text{ln } r + \frac{1}{3}) \bar D_B \delta D^B_A  \\ + \frac{1}{r}
\left(-\frac{1}{3} \delta N_A - \frac{1}{12} \delta (C_{AB} \bar D_C C^{BC})
  + \frac{1}{4} \delta C_{AB} \bar D_C C^{BC} \right)+o(r^{-1-\epsilon})\,.
\end{multline}
Putting everything together, we get
\begin{eqnarray}
   \cQ_\xi[h,\bar g]&=&\frac{1}{16 \pi G}\lim_{r \to \infty} \int
   d^2 \Omega^\varphi \,
\Big[ r \left( Y^A \frac{1}{2 } \bar D_B \delta C^B_A + Y^A
  \frac{1}{2 } \bar D_B \delta C^B_A  \right) \nonumber \\
&& \quad + Y^A \bar D_B \delta D^B_A \left(\frac{4}{3} \text{ ln } r -
  \frac{2}{9} +\frac{2}{3} \text{ ln } r + \frac{2}{9}
\right) -\frac{\psi}{8} C^{AB} \delta C_{AB}\nonumber\\
&& \quad + f \left(4 \delta M - \frac{1}{2} \bar D_A\bar D_B \delta C^{AB} +
\frac{1}{2} \delta \partial_u (C^{AB}C_{AB}) \right. \nonumber \\ &&
\qquad \left. - \frac{1}{2} \partial_u
C_{AB} \delta C^{AB} - C^{AB} \partial_u \delta C_{AB}\right)\nonumber\\ &&
\quad +Y^A \left( \frac{4}{3}\delta N_A + \frac{1}{3} \delta (C_{AB}
  \bar D_C C^{BC}) -
  \frac{1}{4} C_{AB} \bar D_C \delta C^{BC}\right) \nonumber \\ &&
\quad -2 Y^A \left(-\frac{1}{3} \delta N_A - \frac{1}{12} \delta
  (C_{AB} \bar D_C C^{BC})
  + \frac{1}{4} \delta C_{AB} \bar D_C C^{BC} \right) \nonumber \\ &&
\quad -\frac{1}{2} \bar D_A f \bar D_B \delta C^{AB} -\frac{1}{4}
C_{AB} Y^A \bar D_C
\delta C^{BC}\Big]. 
\end{eqnarray}
This can be simplified to
\begin{eqnarray}
   \cQ_\xi[h,\bar g]&=&\frac{1}{16 \pi G}\int
   d^2 \Omega^\varphi \,
\Big[ -\frac{\psi}{8} C^{AB} \delta C_{AB} +Y^A 2\delta N_A
-\frac{1}{2} \bar D_A f \bar D_B \delta C^{AB}\nonumber\\
&& \quad + f \left(4 \delta M - \frac{1}{2} \bar D_A\bar D_B \delta C^{AB} +
\frac{1}{2} \delta \partial_u (C^{AB}C_{AB}) \right. \nonumber \\ &&
\qquad \left. - \frac{1}{2} \partial_u
C_{AB} \delta C^{AB} - C^{AB} \partial_u \delta C_{AB} \right)\Big]\\ &=&\frac{1}{16 \pi G}\delta \int
   d^2 \Omega^\varphi \,
\Big[ -\frac{\psi}{16} C^{AB} C_{AB} + 2Y^A N_A+ 4 f M \Big] \nonumber \\
&&+\frac{1}{16 \pi G}\int
   d^2 \Omega\,
\Big[ \frac{f}{2} \partial_u
C_{AB} \delta C^{AB}\Big]. 
\end{eqnarray}

This result can be rewritten as
\begin{equation}
  \label{eq:BMS4seppcharges}
   \ndelta\cQ_\xi[h,g] = \delta \left(Q_{s}[\cX]\right)+\ndelta
   \Theta_{s}[\cX,\delta\cX]\, ,
 \end{equation}
 where the integrable part of the surface charge one-form is
 given by 
 \begin{equation}
   \label{eq:BMS4charges}
   Q_{s}[\cX]=\frac{1}{16 \pi G}\int
   d^2 \Omega^\varphi\,\Big[ 4 f M+ Y^A \big(2 N_A + 
\frac{1}{16}\partial_A (C^{CB} C_{CB})\big) \Big]\,,
 \end{equation}
and the non-integrable part is due to the news tensor,
\begin{equation}
  \ndelta\Theta_{s}[\cX,\delta\cX]=\frac{1}{16 \pi G}\int
  d^2 \Omega^\varphi\,
  \Big[ \frac{f}{2} N_{AB} \delta C^{AB}\Big]\,.
\end{equation}

The separation into an integrable and a non-integrable part in \ref{eq:BMS4seppcharges} is not
uniquely defined as this equation also holds in terms $Q_s' = Q_s -
N_s$, $\Theta_s' = \Theta_s + \delta N_s$ for some $N_s[\cX]$.

These charges are very similar and should be compared to those
proposed earlier in\cite{Wald:2000ly} in the context of a closely
related, but slightly different approach to asymptotically flat
spacetimes. 

\section{Algebra of charges and extension}

One of the big advantage of the covariant techniques to compute
charges is that they allow us to define an algebra for those
charges. Moreover, when using the equivalence of the Hamiltonian and
the covariant approaches, one can infer that this algebra coincides
with the Dirac bracket $\left\{ Q_{s_1},Q_{s_2}\right\}^*$ of the charges. Unfortunately, in presence of this non-integrable term, the
usual definition of algebra (\ref{eq:asympchargesalgebra}) fails.

In the non integrable case, we propose as a definition
\begin{equation}
\label{eq:bms4defalgebra}
\left\{Q_{s_1},Q_{s_2}  \right\}[\cX] = (-\delta_{s_2}) Q_{s_1}[\cX] + \ndelta
\Theta_{s_2} [\cX,-\delta_{s_1} \cX]\,. 
\end{equation}

\vspace{5mm}

\begin{theorem}
\label{theo:antisym} 
The $BMS_4$ charges (\ref{eq:BMS4charges}) satisfy
\begin{equation}
\left\{Q_{s_1},Q_{s_2}  \right\}[\cX] = Q_{[s_1,s_2]} [\cX]+
K_{s_1,s_2}[\cX], \label{eq:BMS4skew}
\end{equation}
where the field dependent central extension is 
\begin{multline}
K_{s_1,s_2}[\cX]= \frac{1}{32 \pi G}\int
   d^2 \Omega^\varphi\,
\Big[  ( f_1
  \d_A f_2 - f_2 \d_A f_1)\d^A \bar R+\\+C^{BC}  (f_1 \bar D_B \bar
  D_C \psi_2 - f_2 \bar D_B \bar D_C
  \psi_1) \Big] \,, \label{eq:BMS4centralcha}
\end{multline}
\end{theorem}

\begin{proof}
\footnotesize

We will start by computing the usual factor,
\begin{multline}
-\delta_{s_2} Q_{s_1}[\cX]=\frac{1}{16 \pi G}\int
   d^2 \Omega^\varphi\,
\Big[ Y_1^A \left(2 (-\delta_{s_2}) N_A + \frac{1}{16}\partial_A
  (-\delta_{s_2}) (C^{CB} C_{CB})\right)\\+ 4 f_1 (-\delta_{s_2}) M \Big]\,,
\end{multline}
and organize according to the different types of terms that appear:
\begin{itemize}

\item terms containing $M$
\begin{eqnarray}
-\delta_{s_2} Q_{s_1}[\cX]_M&=&\frac{1}{16 \pi G}\int
   d^2 \Omega^\varphi\,
\Big[ Y_1^A 2 (f_2 \partial_A M + 3 \partial_A f_2 M) + 4 f_1
(Y_2^A\d_AM+\frac{3}{2} \psi_2M)\Big]\nonumber\\
&=&\frac{1}{16 \pi G}\int
   d^2 \Omega^\varphi\, 4M
\Big[ - \half \bar D_A ( Y_1^A f_2) + \frac{3}{2} Y^A_1 \partial_A f_2
- \bar D_A (f_1
Y_2^A)+\frac{3}{2} f_1\psi_2)\Big]\nonumber\\
&=&\frac{1}{16 \pi G}\int
   d^2 \Omega^\varphi\, 4M
\Big[ Y^A_1 \partial_A f_2  - \half \psi_1 f_2  -Y_2^A\partial_A f_1
+\frac{1}{2} f_1\psi_2)\Big]\nonumber\\
&=&\frac{1}{16 \pi G}\int
   d^2 \Omega^\varphi\, 4M f_{[s_1,s_2]}\,,
\end{eqnarray}

\item terms containing $N_A$
\begin{eqnarray}
-\delta_{s_2} Q_{s_1}[\cX]_N&=&\frac{1}{16 \pi G}\int
   d^2 \Omega^\varphi\,
\Big[ 2Y_1^A (\cL_{Y_2} + \psi_2) N_A \Big]\nonumber \\
&=&\frac{1}{16 \pi G}\int
   d^2 \Omega^\varphi\,
\Big[ 2Y_1^A (Y_2^B \bar D_B + \psi_2) N_A + 2 Y_1^A \bar D_A Y_2^B N_B \Big]\nonumber \\
&=&\frac{1}{16 \pi G}\int
   d^2 \Omega^\varphi\, 2 N_A
\Big[ -Y_2^B \bar D_B  Y_1^A + Y_1^B \bar D_B Y_2^A\Big]\nonumber \\
&=&\frac{1}{16 \pi G}\int
   d^2 \Omega^\varphi\, 2 N_A Y^A_{[s_1,s_2]}\,,
\end{eqnarray}

\item terms containing $D_{AB}$
\begin{eqnarray}
-\delta_{s_2} Q_{s_1}[\cX]_D&=&\frac{1}{16 \pi G}\int
   d^2 \Omega^\varphi\, 2 Y_1^A 
\Big[  -\half [\bar D_B \psi_2 +\psi_2
  \bar D_B ] D^B_A\Big]\nonumber \\
&=&\frac{1}{16 \pi G}\int
   d^2 \Omega^\varphi\, 2 Y_1^A 
\Big[ -\half \bar D_B(\psi_2 D^B_A)\Big]\nonumber \\
&=&\frac{1}{16 \pi G}\int
   d^2 \Omega^\varphi\, \bar D^B Y_1^A \psi_2 D_{AB}=0\,,
\end{eqnarray}

\item terms containing the news
\begin{eqnarray}
  -\delta_{s_2} Q_{s_1}[\cX]_{news}&=&\frac{1}{16 \pi G}\int
  d^2 \Omega^\varphi\,
  \Big[ 2Y_1^A \left( -\frac{3}{16} \bar D_A f_2
    N^B_CC^C_B
    +\half \bar D_B f_2 N^B_C C^C_A\right) \nonumber \\ && \qquad +2 Y^A_1
  f_2 \left( \frac{1}{16}\d_A\big[N^B_C C^C_B\big]
    -\frac{1}{4} \bar D_AC^C_BN^B_C
    -\frac{1}{4} \bar D_B\big[C^{B}_{C}N^C_{A}-N^B_CC^C_A\big]\right)
  \nonumber 
\\ && \qquad - \psi_1\frac{1}{8}
  C^{AB} f_2 N_{AB}+ 4 f_1 \left(\frac{1}{4} \bar D_B \bar D_C f_2 N^{BC} + \half
    \bar D_B f_2 \bar D_C N^{BC} \right)\nonumber \\ && \qquad+ 4 f_1
  f_2 
\left( -\frac{1}{8} N^A_BN^B_A
    +\frac{1}{4}\bar D_A \bar D_C N^{CA} \right)\Big] \nonumber \\
  &=&\frac{1}{16 \pi G}\int
  d^2 \Omega^\varphi\,\frac{-1}{2} N^{BC} f_2 
  \Big[ f_1 N_{BC} + \cL_{Y_1} C_{BC} - \half \psi_1 C_{BC} -2 \bar D_B \bar D_C
  f_1\Big]\nonumber \\ && + \frac{1}{16 \pi G}\int
  d^2 \Omega^\varphi\,\half N^{C}_B C_{CA} \Big[ Y_1^A \bar D^Bf_2 +
  Y^B_1 \bar D^A f_2 - \bar \gamma^{AB} Y_1^D \bar D_Df_2\Big]\,.
\end{eqnarray}
The second line is zero. This is coming from the following identity
for the symmetrized product of two traceless tensors in 2
dimensions,
\begin{equation}
\half (C^{A}_B K^{B}_C+K^A_B C^B_C) = \half \delta^{A}_{C} C^B_D K^D_B\,,
\end{equation}
and the conformal Killing equation for the $Y^A$. The first line can
be recognized as, 
\begin{eqnarray}
-\delta_{s_2} Q_{s_1}[\cX]_{news}&=&\frac{1}{16 \pi G}\int
   d^2 \Omega^\varphi\,\frac{-1}{2} N^{BC} f_2 
\Big[-\delta_{s_1} C_{BC} \Big]\nonumber \\
&=&- \ndelta \Theta_{s_2} [-\delta_{s_1} \cX,\cX] \,,
\end{eqnarray}

\item the rest
\begin{eqnarray}
  -\delta_{s_2} Q_{s_1}[\cX]_R&=&\frac{1}{16 \pi G}\int
  d^2 \Omega^\varphi\,
  \Big[ 2Y_1^A  \Big( -\frac{1}{32}\bar D_A\psi_2
  C^B_CC^C_B +f_2 \frac{1}{4}C_A^B\d_B
  \bar  R\nonumber\\ &&\qquad -\frac{1}{4}f_2
  \bar D_B \left( \bar D^B \bar D_CC^C_A -\bar D_A 
    \bar D_CC^{BC}\right)\nonumber\\ &&\qquad + \frac{1}{4} ( \bar
  D_Bf_2 \bar R+\bar 
  D_B \bar \Delta
  f_2)C^B_A -\frac{3}{4} \bar D_B f_2(\bar D^B\bar D_C C^C_A-
  \bar D_A \bar D_C C^{BC})
  \nonumber\\ &&\qquad +\half (\bar D_A\bar D_B f_2
  -\frac{1}{2} \bar \Delta f_2\bar \gamma_{AB} ) \bar D_C C^{CB}+\frac{3}{8} 
  \bar D_A(\bar D_C\bar D_B f_2 C^{CB}) \Big) \nonumber \\ && \qquad -\psi_1 \frac{1}{8}
  C^{CB} \left( [\cL_{Y_2} -\half\psi_2] C_{CB}-2 \bar D_C \bar D_B f_2+
    \bar \Delta f_2\bar 
\gamma_{CB}\right)\nonumber \\ && \qquad+ 4 f_1 \left(f_2 \frac{1}{8}
    \bar \Delta  \bar R +\frac{1}{4}\d_A f_2 \d^A \bar R +\frac{1}{8}
    \bar D_C \bar D_B \psi_2 C^{CB} \right) \Big]\nonumber \\ 
  &=&\frac{1}{16 \pi G}\int
  d^2 \Omega^\varphi\,
  \Big[ -Y_1^A \frac{1}{16}\bar D_A\psi_2
  C^B_CC^C_B -\psi_1 \frac{1}{8}
  C^{CB} \left( [\cL_{Y_2} -\half\psi_2] C_{CB}\right) \nonumber\\
  &&\qquad +C^{BC} \Big( \half f_1 \bar D_B \bar D_C \psi_2 +\psi_1
  \frac{1}{4} \bar D_C\bar D_B f_2 + \half f_2 Y_{1B} \d_C \bar R \nonumber\\
  &&\qquad  -\frac{3}{4} \psi_1 \bar D_B \bar D_C f_2 - \bar D_C
  (Y_1^A \bar D_A\bar 
  D_B f_2)
  +\frac{1}{2} \bar D_C ( Y_{1B}\bar \Delta f_2 ) \nonumber\\ &&\qquad
  +\frac{1}{2} Y_{1C}( \bar D_Bf_2\bar R+
  \bar D_B \bar \Delta
  f_2) + \half \bar D_C \bar \Delta (Y_{1B} f_2) - \half \bar D_C \bar D_A \bar D_B (Y_1^A f_2)
  \nonumber\\ &&\qquad -\frac{3}{2} \bar D_C \bar D_A (Y_{1B} \bar D^Af_2) +
  \frac{3}{2} \bar D_C \bar D_A (Y_1^A D_B f_2)\Big) \nonumber\\ &&\qquad+\half ( f_1
  \d_A f_2 - f_2 \d_A f_1)\d^A \bar R\Big]\,.
\end{eqnarray}
Using the commutation rule for covariant derivatives, this gives
\begin{eqnarray}
-\delta_{s_2} Q_{s_1}[\cX]_C&=&\frac{1}{16 \pi G}\int
   d^2 \Omega^\varphi\,
\Big[ -\frac{1}{16}(Y_1^A  \bar D_A\psi_2-Y_2^A  \bar D_A\psi_1)
  C^B_CC^C_B \nonumber\\ && \qquad+\half ( f_1
  \d_A f_2 - f_2 \d_A f_1)\d^A \bar R
   +C^{BC} \Big(\half (f_1 \bar D_B \bar D_C \psi_2 - f_2 \bar D_B \bar D_C
  \psi_1) \nonumber \\ && \qquad
+ \frac{1}{4}f_2 Y_{1B} \d_C \bar R + \half f_2 \bar D_C \bar \Delta Y_{1B}
   + \frac{1}{4}\bar D_Cf_2 Y_{1B} \bar R + \half
  \bar D_C f_2 \bar \Delta Y_{1B}\Big) \Big]\nonumber\\
&=&\frac{1}{16 \pi G}\int
   d^2 \Omega^\varphi\,
\Big[ -\frac{1}{16}\psi_{[s_1,s_2]}
  C^B_CC^C_B +\half ( f_1
  \d_A f_2 - f_2 \d_A f_1)\d^A \bar R\nonumber\\
  &&\qquad +C^{BC} \half (f_1 \bar D_B \bar D_C \psi_2 - f_2 \bar D_B \bar D_C
  \psi_1) \Big]\,,
\end{eqnarray}
where in the last line we have used the identity $\bar \Delta Y^A =
-\half \bar R Y^A$ satisfied by conformal Killing vectors.

\end{itemize}

Summing everything, we obtain
\begin{eqnarray}
-\delta_{s_2} Q_{s_1}[\cX]
&=&\frac{1}{16 \pi G}\int
   d^2 \Omega^\varphi\,
\Big[ -\frac{1}{16}\psi_{[s_1,s_2]}
  C^B_CC^C_B +\half ( f_1
  \d_A f_2 - f_2 \d_A f_1)\d^A \bar R\nonumber\\
  &&\qquad +C^{BC} \half (f_1 \bar D_B \bar D_C \psi_2 - f_2 \bar D_B
  \bar D_C
  \psi_1) +4M f_{[s_1,s_2]}+2 N_A Y^A_{[s_1,s_2]}\Big]\nonumber \\
&& \qquad- \Theta_2 [-\delta_1 \cX,\cX] \nonumber \\
&=& Q_{[s_1,s_2]}- \ndelta \Theta_2 [-\delta_1 \cX,\cX] +K_{s_1,s_2}[\cX]\,,
\end{eqnarray}
with $K_{s_1,s_2}[\cX]$ defined in \eqref{eq:BMS4centralcha}.
\normalsize
\end{proof}

\begin{theorem}
The central extension satisfies the
cocycle condition
\begin{equation}
\label{eq:BMS4cocycle}
K_{[s_1,s_2],s_3} + \delta_{s_3} K_{s_1,s_2} + {\rm cyclic}\ (1,2,3)
=0.
\end{equation}
\end{theorem}

\begin{proof}
\footnotesize

Let us treat the two parts of $K_{s_1,s_2}[\cX]$ separately:
\begin{itemize}

\item for the second part $\hat K_{s_1,s_2}= \frac{1}{16 \pi G}\int
   d^2 \Omega^\varphi\,C^{BC} \half (f_1 \bar D_B \bar D_C \psi_2 -
   f_2 \bar D_B \bar D_C
  \psi_1) $, we have
\begin{eqnarray}
  A&=&\int
  d^2 \Omega^\varphi\,\Big[(-\delta_{s_3} C^{BC})  (f_1 \bar D_B \bar
  D_C \psi_2 - f_2 \bar D_B \bar D_C
  \psi_1) +{\rm cyclic}\ (1,2,3) \Big]\nonumber \\
  &=&\int
  d^2 \Omega^\varphi\,\Big[([f_3\d_u+\cL_{Y_3} -\half \psi_3] C_{AB}-2
  \bar D_A\bar D_B f_3+\Delta f_3\bar \gamma_{AB})\nonumber \\ && \qquad
  \qquad \qquad  (f_1 \bar D^B \bar D^A \psi_2 - f_2 \bar D^B \bar D^A
  \psi_1) +{\rm cyclic}\ (1,2,3)\Big]\nonumber \\
  &=&\int
  d^2 \Omega^\varphi\,\Big[-C_{BC} \bar D_A \left((Y_1^A f_2 - Y_2^A
    f_1) \bar D^B \bar D^C \psi_3 \right)\nonumber \\ && 
  +2 C_{BC} \left(\bar D_A Y^B_1 f_2 -\bar D_A Y^B_2 f_1
  \right) \bar D^A
  \bar D^C \psi_3  -\half C_{BC} \left(\psi_1 f_2 -\psi_2 f_1 \right) \bar D^B
  \bar D^C \psi_3\nonumber \\ && 
  +2 \left(\bar D_C f_1 f_2 - \bar D_C f_2 f_1 \right) \left(
    \bar \Delta \bar D^C
    \psi_3 - \half \bar D^C \bar \Delta \psi_3\right)+{\rm cyclic}\ (1,2,3) \Big],
\end{eqnarray}
The second term is given by
\begin{eqnarray}
  B&=&\int
  d^2 \Omega^\varphi\,\Big[C^{BC}  (f_{[s_1,s_2]} \bar D_B \bar D_C
  \psi_3 - f_3 \bar D_B \bar D_C
  \psi_{[s_1,s_2]})+{\rm cyclic}\ (1,2,3) \Big]\nonumber \\
  &=&\int
  d^2 \Omega^\varphi\,C^{BC} \Big[\bar D_A \left( (Y^A_1 f_2 - Y^A_2
    f_1)\bar  D_B \bar D_C
    \psi_3\right)  - \frac{3}{2} \left( \psi_1
    f_2 - \psi_2 f_1\right) \bar D_B \bar D_C \psi_3 \nonumber \\ &&  -
  \left(Y^A_1 f_2 - Y^A_2 f_1 \right) \bar D_A \bar D_B \bar D_C
  \psi_3
- f_3 \bar D_B \bar D_C \left(Y^A_1 \bar D_A \psi_2
    - Y^A_2 \bar D_A \psi_1 \right)+{\rm cyclic}\ (1,2,3) \Big]\nonumber \\
  &=&\int
  d^2 \Omega^\varphi\,C^{BC} \Big[\bar D_A \left( (Y^A_1 f_2 - Y^A_2
    f_1) \bar D_B \bar D_C
    \psi_3\right) - \frac{3}{2} \left( \psi_1
    f_2 - \psi_2 f_1\right) \bar D_B \bar D_C \psi_3 \nonumber \\ && \qquad
  -2 \left( f_1 \bar D_B Y^A_2 - f_2 \bar D_B Y^A_1\right) \bar D_C
  \bar D_A \psi_3+{\rm cyclic}\ (1,2,3)\Big]\,.
\end{eqnarray}
Summing the two, we get
\begin{eqnarray}
A+B & = & \int
   d^2 \Omega^\varphi\,C^{BC} \Big[-2 \left( f_1 (\bar D_B Y^A_2 +
     \bar D^A Y_{2B}) - f_2 (\bar D_B Y^A_1 + \bar D^A
       Y_{1B})\right) \bar D_C \bar D_A \psi_3 \nonumber\\ && - 2 \left( \psi_1
       f_2 - \psi_2 f_1\right) \bar D_B \bar D_C \psi_3 +2 \left(\bar
       D_C f_1 f_2 - \bar D_C f_2 f_1 \right) \left( \bar \Delta \bar D^C
   \psi_3 - \half \bar D^C \bar \Delta \psi_3\right)\nonumber \\
     &&
\qquad \qquad +{\rm cyclic}\ (1,2,3) \Big]\nonumber \\ && \hspace*{-2cm}
=\int
   d^2 \Omega^\varphi\,\Big[ 2 \left(\bar D_C f_1 f_2 - \bar D_C f_2
     f_1 \right) 
\left( \bar \Delta \bar D^C
   \psi_3 - \half \bar D^C \bar \Delta \psi_3\right)+{\rm cyclic}\ (1,2,3) \Big]\,.
\end{eqnarray}
We can then use the following identities $ \bar \Delta \psi = -\bar
D_A(\bar R Y^A)$ and $\Delta \bar D^C \psi = \bar D^C \Delta\psi +
\half \bar R \bar D^C \psi$ to simplify the above to
\begin{eqnarray}
A&+&B  = \int
   d^2 \Omega^\varphi\,\Big[ 2 \left(\bar D_C f_1 f_2 - \bar D_C f_2 f_1 \right) \left(
     -\half \bar D^C (Y_3^A \bar D_A R) - \psi_3\half \bar D^C \bar R
   \right)+{\rm cyclic}\ (1,2,3) 
   \Big]\nonumber \\ \hspace*{-4cm}
&=&\int
   d^2 \Omega^\varphi\,\Big[ 2 \left(\bar D_C f_1 f_2 - \bar D_C f_2 f_1 \right) \left(
     -\half \cL_{Y_3} \bar D^C\bar R - \psi_3\bar D^C \bar R \right)+{\rm cyclic}\ (1,2,3)
   \Big]\,.
\end{eqnarray}

\item for the first part $\tilde K_{s_1,s_2}= \frac{1}{16 \pi G}\int
   d^2 \Omega^\varphi\, \half ( f_1
  \d_A f_2 - f_2 \d_A f_1)\d^A \bar R $, the condition (\ref{eq:BMS4cocycle})
  leads to 
\begin{eqnarray}
C & = & \int d^2 \Omega^\varphi\,\Big[ f_{[s_1,s_2]} \d_A f_3 - f_3
\d_A f_{[s_1,s_2]} \d^A \bar R
+{\rm cyclic}\ (1,2,3)\Big]\nonumber \\ 
&=& \int d^2 \Omega^\varphi\,\Big[ \cL_{Y_1} (f_2 \d_A f_3- f_3 \d_A f_2 ) -
\psi_1 (f_2 \d_A f_3- f_3 \d_A f_2 )\d^A \bar R+
{\rm cyclic}\ (1,2,3)\Big]\nonumber \\ 
&=& \int d^2 \Omega^\varphi\,\Big[ (f_2 \d_A f_3- f_3 \d_A f_2 ) (-\cL_{Y_1} -
2 \psi_1)\d^A \bar R +{\rm cyclic}\ (1,2,3)\Big]\,.
\end{eqnarray}

\end{itemize}

The different contributions then sum up to zero, $A+B+C=0$. 
\normalsize
\end{proof}

An open question is if and in what sense the proposed bracket is
indeed a Dirac bracket. The point we want to make here is first of all that
theorem \ref{theo:antisym} assure the skew-symmetricity of our new
definition (\ref{eq:bms4defalgebra}). Furthermore,
\eqref{eq:BMS4skew} and \eqref{eq:BMS4cocycle} imply the Jacobi
identity for this bracket, provided that the transformation associated with $Q_{[s_1,s_2]} [\cX]+
K_{s_1,s_2}[\cX]$ is just $\delta_{[s_1,s_2]}$ or in other words that
the field dependent central extension does not generate a
transformation.

When defining as before,
$\left\{Q^\prime_{s_1},Q^\prime_{s_2} \right\}^*[\cX] =
(-\delta_{s_2}) Q^\prime_{s_1}[\cX] + \Theta^\prime_{s_2}
[-\delta_{s_1} \cX,\cX]$, one gets
$  \left\{Q^\prime_{s_1},Q^\prime_{s_2}  \right\}^*= Q^\prime_{[s_1,s_2]}+
K^\prime_{s_1,s_2}$,
where 
\begin{equation}
K^\prime_{s_1,s_2}=K_{s_1,s_2}+\delta_{s_2}
N_{s_1}-\delta_{s_1} N_{s_2}+N_{[s_1,s_2]}.
\end{equation}
Note that $\delta_{s_2} N_{s_1}-\delta_{s_1} N_{s_2}+N_{[s_1,s_2]}$ is
a trivial field dependent $2$-cocycle in the sense that it
automatically satisfies the cocycle condition \eqref{eq:BMS4cocycle}.

\section{Conclusion and outlook}
\label{sec:conclusion-outlook}

In this chapter, we have shown that the symmetry algebra of
asymptotically flat $4$ dimensional spacetimes is $\mathfrak{bms}_4$,
an algebra that contains both the Poincar\'e algebra and the non
centrally extended Virasoro algebra in a completely natural way. As a
first non trivial effect, we have computed the detailed transformation
properties of the data characterizing solution space. 

Using a covariant method, we constructed the associated surface
charges. They agree with the usual definitions found in the
literature. In order to define an algebra, we introduced a new
bracket. With this new definition, the algebra of the charges form a
representation of $BMS_4$ up to a general field-dependent
extension. More work is still needed to fully justify the proposal
for this Dirac bracket.

We believe that our understanding of the symmetry structure and its
action on solution space goes some way in getting quantitative control
on ``structure X'' \cite{witten:98xx}, i.e., on a holographic
description of gravity with zero cosmological constant.

In the future, it should be
interesting to analyze in more details the consequences of our results
on local conformal invariance for the non extremal Kerr/CFT
correspondence and for the gravitational S-matrix for instance.

\appendix

\chapter{Surface charges}
\label{app:surfacecharges}

\section{Hamiltonian approach}
\label{app:regge-teit-revis}

\subsection{Regge-Teitelboim revisited}

We present here an adaptation of the original
Hamiltonian method of \cite{Regge:1974kx,Benguria:1977fk}.

Let $\cL_H=a_A \dot z^A-h-\gamma_{\alpha} u^\alpha$, with $h$ a first
class Hamiltonian density and $\gamma_\alpha$ first class constraints
and define $\phi^i=(z^A,u^\alpha)$. Even though it is not so for our
theory, let us first run through the arguments in the case where one
has Darboux coordinates for the symplectic structure, i.e., when
$\sigma_{AB}=\dover{a_B}{z^A}-\dover{a_A}{z^B}$ is the constant
symplectic matrix. The gauge transformation generated by the
smeared constraints is given by
\begin{equation}
\delta_{\varepsilon} z^A = \sigma^{AB} \frac{\delta (
  \varepsilon^\alpha \gamma_\alpha)}{\delta z^B}
\end{equation}
where $\sigma^{AB}$ is the inverse of $\sigma_{AB}$. This
transformation is extended to the Legendre multiplicators in order to leave the action invariant (see e.g. \cite{Henneaux:1992fk}).

We furthermore suppose that we are in a source-free
region of spacetime. In this case one can show that
\begin{equation}
  \label{eq:app107b}
  \delta_\varepsilon z^A\vddl{\cL_H}{z^A}+\delta_\varepsilon
  u^\alpha\vddl{\cL_H}{u^\alpha}= -\d_0\Big(\gamma_\alpha
  \varepsilon^\alpha\Big)-\d_is^i_\varepsilon,
\end{equation}
where $s^i_\varepsilon=s^i_\varepsilon[z,u]$ vanishes when the
Hamiltonian equations of motion, including constraints, are satisfied,
$s^i_\varepsilon\approx 0$. This identity merely expresses the general
fact that the Noether current $s^\mu_\varepsilon$ associated to a
gauge symmetry can be taken to vanish when the equations of motions
hold (see e.g.~\cite{Henneaux:1992fk}, chapter 3),
$s^\mu_\varepsilon\approx 0$, and that the integrand of the generator
is given by (minus) the constraints contracted with the gauge
parameters in the Hamiltonian formalism,
$s^0_\varepsilon=-\gamma_\alpha\varepsilon^\alpha$. An explicit
expression for $s^i_\varepsilon$ in terms of the structure functions
can for instance be found in Appendix D
of~\cite{Barnich:2008uq}. Using integrations by parts, one can write
the variations of the constraints under a change of the canonical
coordinates $z^A$ as an Euler-Lagrange derivative, up to a total
derivative,
 \begin{eqnarray}
  \label{eq:app29b}
\delta_z(\gamma_\alpha\varepsilon^\alpha) =\delta z^A
\vddl{(\gamma_\alpha\varepsilon^\alpha)}{
    z^A}-\d_i k^{i}_{\varepsilon}.
\end{eqnarray}
where $k^{i}_{\varepsilon}=k^{i}_{\varepsilon}[\delta z,z]$ depends
linearly on $\delta z^A$ and its spatial derivatives.  Taking the
time derivative of \eqref{eq:app29b} and using a variation $\delta_\phi$
of \eqref{eq:app107b} to eliminate
$\d_0\delta_z(\gamma_\alpha\varepsilon^\alpha)$, one finds
\begin{equation}
  \label{eq:app118}
  \d_i\Big(\d_0  k^{i}_{\varepsilon}-\delta_\phi
  s^i_\varepsilon\Big)=
  \d_0\Big(\delta z^A\vddl{(\gamma_\alpha\varepsilon^\alpha)}{
    z^A}\Big)+
  \delta_\phi\Big( \delta_\varepsilon z^A\vddl{\cL_H}{z^A}+\delta_\varepsilon
  u^\alpha\vddl{\cL_H}{u^\alpha}\Big). 
\end{equation}
One now takes $\varepsilon^\alpha_s$ to satisfy
$\delta_{\varepsilon_s }z^A_s=0=\delta_{\varepsilon_s}
u^\alpha_s$. Note that in the case of Darboux coordinates, this also
implies that $\vddl{(\gamma_\alpha\varepsilon^\alpha_s)}{
  z^A}=0$. This is where we differ from the original analysis. The
authors of \cite{Regge:1974kx} considered asymptotic symmetries: $\delta_{\varepsilon_s }z^A_s=0=\delta_{\varepsilon_s}
u^\alpha_s$ only as $r$ goes to infinity. We are
only considering exact reducibility parameters (in the case of
gravity, they are the killing vectors). If furthermore $z^A_s,u^\alpha_s$ is a solution of the
Hamiltonian equations of motion and the RHS of \eqref{eq:app118}
vanishes. By using a contracting homotopy with respect to $\delta
\phi^i$ and their spatial derivatives, one deduces that
\begin{equation}
  \label{eq:app119}
  \d_0  k^{i}_{\varepsilon_s}[\delta z,z_s]=(\delta_\phi
  s^i_{\varepsilon_s})[z_s]-\d_j k^{[ij]}_{\varepsilon_s},
\end{equation}
where $k^{[ij]}_{\varepsilon_s}=k^{[ij]}_{\varepsilon_s}[\delta\phi,
\phi_s]$ depends linearly on $\delta \phi^i$ and their spatial
derivatives. Finally, when $\delta z^A_s,\delta u^\alpha_s$ satisfy
the linearized Hamiltonian equations of motion, including constraints,
we find from \eqref{eq:app29b} and \eqref{eq:app119} that
\begin{equation}
  \label{eq:app120}
\d_i k^{i}_{\varepsilon_s}[\delta z_s,z_s]=0,\quad  \d_0
k^{i}_{\varepsilon_s}[\delta z_s,z_s]-\d_j t^{[ij]}_{\varepsilon_s}=0.
\end{equation}
At a fixed time $t=x^0$, consider a closed $2$ dimensional surface
$S$, $\d S=0$, for instance a sphere with radius $r$ and define the
surface charge 1-forms by
\begin{equation}
  \label{eq:121}
  \ndelta \cQ_{\varepsilon_s}[\delta z_s,z_s]=\oint_{S} d^3x_i\,
  k^{i}_{\varepsilon_s}[\delta z_s,z_s], 
\end{equation}
where $d^{3}x_i=\half\epsilon_{ijk}dx^j\wedge dx^k$. The first
relation of \eqref{eq:app120} implies that the surface charge 1-form only
depends on the homology class of the closed surface $S$, 
\begin{equation}
  \label{eq:40}
  \oint_{S_1} d^{3}x_m\,k^{m}_{\varepsilon_s}[\delta z_s,z_s]=
\oint_{S_2} d^{3}x_m\,k^{m}_{\varepsilon_s}[\delta z_s,z_s].
\end{equation}
Here $S_1-S_2=\d\Sigma$, where $\Sigma$ is a three-dimensional volume
at fixed time $t$ containing no sources. For instance, the surface
charge $1$-form does not depend on $r$.
The second relation of \eqref{eq:app120} implies that it is conserved in
time and so does not depend on $t$ either,
\begin{equation}
  \frac{d}{dt}\ndelta \cQ_{\varepsilon_s}[\delta z_s,z_s]=0.
\end{equation}

The objects obtained by this method are 1-forms on solution space. To
build the wanted charges, we still need to integrate them:
\begin{equation}
\label{eq:appasymptcharge4}
Q_{\varepsilon_s}[z_s] = \int_{Z_0}^{Z_s} \ndelta \cQ_{\varepsilon_s}[\delta_\gamma z_s,z_s]
\end{equation}
where the integration is done along a path $\gamma$ in solution space
going from a reference solution $Z_0$ to the solution upon
consideration $Z_s$. The variation $\delta_\gamma z_s$ is the
variation of the fields along the path. There is
an issue with integrability: the integral may be path dependent; the
1-form $\ndelta \cQ $ may not be closed, see
e.g.~\cite{Barnich:2008uq,Iyer:1994uq,Barnich:2003fk} for a
discussion. This explains the notation $\ndelta$.

\subsection{Linear theories}
\label{sec:linear-theories}

In the case of linear theories, the latter problem does not arise and
the whole analysis simplifies. One can replace \eqref{eq:app29b} by
\begin{equation}
  \label{eq:app29c}
\gamma_\alpha\varepsilon^\alpha =z^A
\vddl{(\gamma_\alpha\varepsilon^\alpha)}{
    z^A}-\d_i k^{i}_{\varepsilon}[z],
\end{equation}
where $\delta/\delta z^A$ are the (spatial) Euler-Lagrange derivatives
and $k^{i}_\varepsilon[z]$ depends linearly both on the phase space
variables $z^A$ and their spatial derivatives and on the gauge
parameters. One then uses \eqref{eq:app107b} directly to eliminate
$\d_0(\gamma_\alpha\varepsilon^\alpha)$ from the time derivative of
\eqref{eq:app29c}, to get
\begin{equation}
  \label{eq:108}
  \d_i\Big[\d_0 k^i_\varepsilon-s^i_\varepsilon\Big]=\d_0\Big[z^A
  \vddl{(\gamma_\alpha\varepsilon^\alpha)}{
    z^A}\Big]+\delta_\varepsilon
  z^A\vddl{\cL_H}{z^A}+\delta_\varepsilon
  u^\alpha\vddl{\cL_H}{u^\alpha}.
\end{equation}
For gauge parameters $\varepsilon^\alpha_s$ that satisfy
\begin{equation}
\delta_{\varepsilon_s }z^A=0=\delta_{\varepsilon_s} u^\alpha,\label{eq:app127}
\end{equation}
one then arrives at
\begin{equation}
  \label{eq:122}
\d_i k^{i}_{\varepsilon_s}[z]=-\gamma_\alpha\varepsilon^\alpha,\quad  \d_0
k^{i}_{\varepsilon_s}[z]=s^i_{\varepsilon_s}[z,u]-\d_j
k^{[ij]}_{\varepsilon_s}. 
\end{equation}
For a solution $z^a_s,u^\alpha_s$, the surface charges
\begin{equation}
\cQ_{\varepsilon_s}[z_s]=\oint_{S} d^3x_i\,
  k^{i}_{\varepsilon_s}[z_s],\label{eq:123}
\end{equation}
are again independent of $r$ and $t$. 

When this analysis is applied to the Hamiltonian formulation of
Pauli-Fierz theory, one finds the standard expressions 
\begin{equation}
k^{i}_{\varepsilon}[z] =2\xi_m\pi^{mi}-
  \xi^{\perp}(\delta^{mn}\d^i-\delta^{mi}\d^n) h_{mn}+
  h_{mn}(\delta^{mn}\d^i-\delta^{ni}\d^m)\xi^{\perp},\label{eq:128}
\end{equation}
while the only solutions to \eqref{eq:app127} are $\xi_{\mu
  s}=-\omega_{[\mu\nu]}x^\nu +a_\mu$, for some constants $a_\mu$,
$\omega_{[\mu\nu]}=-\omega_{[\nu\mu]}$. In this context of
flat space, Greek indices take values from $0$ to $3$ with
$\mu=(\perp,i)$. Indices $\mu$ are lowered and raised with
$\eta_{\mu\nu}=\text{diag} \,(-1,1,1,1)$.

\section{Covariant approach}
\label{sec:covariant-approach}
The rest of this appendix is devoted to a quick review of
the covariant approach developed in
\cite{Barnich:2002fk,Barnich:2003fk,Barnich:2008uq}. This approach can be applied
to any gauge theory but we will focus here 
on its application to gravity, with or without cosmological constant.

The starting point is the analysis of the linearized theory described
by $h_{\mu\nu}$ around a background
$g_{\mu\nu}$. It has been shown in \cite{Barnich:2004uq} that the
conserved surface charges are completely classified by the Killing
vectors $\xi^\mu$ of the background metric $g_{\mu\nu}$. Moreover,
they form a representation of the Lie algebra of those Killing
vectors. The explicit expression for the surface charges of the
linearized theory depends only on the Einstein equations of motion and
is given by
\begin{multline}
  \label{eq:gravsurfacecharge}
 \ndelta \cQ_\xi[h, g]= \frac{1}{16
    \pi G}\int_{\partial \Sigma}\,(d^{n-2}x)_{\mu\nu}\, \sqrt{- g}\Big[\xi^\nu D^\mu h
  -\xi^\nu D_\sigma h^{\mu\sigma} +\xi_\sigma D^\nu
  h^{\mu\sigma}
  \\
  +\frac{1}{2}h D^\nu\xi^\mu +\half
  h^{\nu\sigma}(D^\mu\xi_\sigma-D_\sigma\xi^\mu)
  -(\mu\leftrightarrow \nu)\Big]\,,
\end{multline}
where
\[(d^{n-k}x)_{\nu\mu}=\frac{1}{k!(n-k)!}
\epsilon_{\nu\mu\alpha_1\dots\alpha_{n-2}}
dx^{\alpha_1}\wedge\dots\wedge dx^{\alpha_{n-2}},\quad \epsilon_{01\dots
  n-1}=1.\]
The algebra of these charges is given by
\begin{equation}
\left\{\ndelta \cQ_{\xi_1}[h, g],\ndelta \cQ_{\xi_2}[h, g] \right\} = -
\ndelta \cQ_{\xi_2}[ \cL_{\xi_1} h, g] \approx \ndelta \cQ_{[\xi_1,\xi_2]}[h, g].
\end{equation}
The last equality is an equality on the equations of motion.

\vspace{5mm}

In view of these universal properties of the surface charges in the
linearized theory, the authors of
\cite{Barnich:2002fk} used them to
define surface charges in the full theory associated to asymptotic symmetries.
For an asymptotic Killing vector $\xi$, they define a charge $\cQ_\xi$
as
\begin{equation}
\label{eq:appasymptcharge}
\cQ_\xi [g] = \int_{\bar g}^g \ndelta \cQ_\xi[\delta_\gamma g, g]
\end{equation}
where $\ndelta \cQ_\xi$ is evaluated on an asymptotic surface $\delta
\Sigma$. For instance, in the case of $AdS_3$, $\delta \Sigma$ will be
a circle cross-section of the cylinder at $r \to \infty$. The
integration of (\ref{eq:appasymptcharge}) is done along a path
$\gamma$ in solution space going from a reference metric $\bar \gamma$
to the metric upon consideration $g$. The variation $\delta_\gamma g$
is the infinitesimal variation of the metric along the path. As
in (\ref{eq:appasymptcharge4}), there may be a problem of integrability.

In the simplest cases, the charges are ``linear'':
\begin{equation}
\label{eq:applinearity}
\ndelta \cQ_\xi[\delta_\gamma g, g] = \ndelta \cQ_\xi[\delta_\gamma g,
\bar g]
\end{equation}
and the integral can be done easily to give
\begin{equation}
\label{eq:applinearcharges}
\cQ_\xi [g] = \ndelta \cQ_\xi[g - \bar g,
\bar g]
\end{equation}
This happens in chapters \ref{chap:ads3}  and \ref{chap:bms3} with space-times in three dimensions that are asymptotically
$AdS_3$ or asymptotically flat at null infinity.

If the charges are integrable, the authors of \cite{Barnich:2008uq} showed that under some
technical conditions the charges form a representation of the algebra
of asymptotic Killing vectors up to a central extension:
\begin{eqnarray}
\label{eq:asympchargesalgebra}
\left\{Q_{\xi_1}[ g],Q_{\xi_2}[ g] \right\} & = & - \delta_{\xi_1}
Q_{\xi_2}[ g]\\
 & =& -\ndelta \cQ_{\xi_2}[ \cL_{\xi_1} g, g]\\
 &\approx & Q_{[\xi_1, \xi_2]}[ g] + K_{\xi_1,\xi_2}
\end{eqnarray}
where the central extension is given by:
\begin{multline}
  K_{\xi_1,\xi_2}=\frac{1}{16 \pi G}
\int_{\6\Sigma} (d^{n-2}x)_{\nu\mu}
\sqrt{- \5g}
\Big[
-2\5D_\rho\xi^\rho_2\, \5D^\nu\xi^{\mu}_1
+2\5D_\rho\xi^{\rho}_1\, \5D^\nu\xi^{\mu}_2
\\
+4\5D_\rho\xi^{\nu}_2\, \5D^\rho\xi^{\mu}_1
+(\5D^\rho \xi^{\nu}_1+\5D^\nu \xi^{\rho}_1)
(\5D^\mu \xi_{2\rho}+\5D_\rho \xi^\mu_2)
\\
-2\5R^{\mu\nu\rho\sigma}\xi_{2\rho}\xi_{1\sigma}
\Big].
\end{multline}

In addition, the covariant expression for the surface charges described
above coincides on-shell with those of the Hamiltonian formalism
\cite{Barnich:2002fk,Barnich:2008uq}. In this context, it follows from
the analysis of \cite{Brown:1986fk,Brown:1986zr,Regge:1974kx} that the
surface charge is, after gauge fixation, the
canonical generator of the associated asymptotic transformations in the Dirac
bracket.

\bibliography{/Users/Cedric/Documents/Physics/Papers/Cedric.bib}

\end{document}